\documentclass[onecolumn, draftcls,11pt]{IEEEtran}

\usepackage{amsfonts}                                                          
\usepackage{epsfig}
\usepackage{amsthm}
\usepackage{graphicx}        
\usepackage{latexsym}
\usepackage{amssymb}
\usepackage{amsmath}
\usepackage{parskip}
\usepackage{multirow}
\usepackage{cite}
\usepackage{mathtools}
\usepackage{epstopdf}
\usepackage{etoolbox}
\usepackage{array}
\usepackage{mathptmx}
\usepackage[colorlinks=true]{hyperref}
\usepackage[bottom]{footmisc}
\usepackage{tikz}

\usetikzlibrary{decorations.pathreplacing}

\usepackage{verbatim}
\usepackage{calrsfs}
\usepackage{booktabs}
\allowdisplaybreaks

\newcommand{\CDC}[1]{\textcolor{blue}{{#1}}}

\DeclareMathAlphabet{\pazocal}{OMS}{zplm}{m}{n}

\let\bbordermatrix\bordermatrix
\patchcmd{\bbordermatrix}{8.75}{4.75}{}{}
\patchcmd{\bbordermatrix}{\left(}{\left[}{}{}
\patchcmd{\bbordermatrix}{\right)}{\right]}{}{}

\newcommand{\sr}{\stackrel}

\newcommand{\rar}{\rightarrow}

\newcommand{\tri}{\sr{\triangle}{=}}

\newcommand{\be}{\begin{equation}}
\newcommand{\ee}{\end{equation}}
\newcommand{\bea}{\begin{eqnarray}}
\newcommand{\eea}{\end{eqnarray}}
\newcommand{\bes}{\begin{eqnarray*}}
\newcommand{\ees}{\end{eqnarray*}}
\newcommand{\bce}{\begin{center}}
\newcommand{\ece}{\end{center}}
\newcommand{\beae}{\begin{IEEEeqnarray}{rCl}}
\newcommand{\eeae}{\end{IEEEeqnarray}}
\newcommand{\nms}{\IEEEeqnarraynumspace}
\def\VR{\kern-\arraycolsep\strut\vrule &\kern-\arraycolsep}
\def\vr{\kern-\arraycolsep & \kern-\arraycolsep}

\newcommand{\ben}{\begin{enumerate}}
\newcommand{\een}{\end{enumerate}}

\newcommand{\hso}{\hspace{.1in}}
\newcommand{\hst}{\hspace{.2in}}

\newcommand{\noi}{\noindent}

\newtheorem{theorem}{Theorem}[section]
\newtheorem{problem}{Problem}[section]
\newtheorem{remark}{Remark}[section]
\newtheorem{corollary}{Corollary}[section]

\newtheorem{definition}{Definition}[section]
\newtheorem{lemma}{Lemma}[section]

\newtheorem{proposition}{Proposition}[section]

\newtheorem{observation}{Observation}[section]

\newtheorem{counterexample}{Counterexample}[section]

\begin{document}

\baselineskip=17pt


\title{Time-Invariant Feedback  Strategies Do  Not  Increase Capacity of AGN  Channels Driven by Stable and Certain Unstable   Autoregressive  Noise\footnote{This paper was presented in part at the 2020 IEEE International Symposium on Information Theory, Los Angeles, USA,  July 21-26, 2020}}

 \author{
   \IEEEauthorblockN{
     Charalambos D. Charalambous\IEEEauthorrefmark{1},
     Christos Kourtellaris\IEEEauthorrefmark{1},
     Sergey Loyka\IEEEauthorrefmark{2}\\}
   \IEEEauthorblockA{
     \IEEEauthorrefmark{1}Department of Electrical and Computer Engineering\\University of Cyprus, 75 Kallipoleos Avenue, P.O. Box 20537, Nicosia, 1678, Cyprus
     \\
     Email: chadcha@ucy.ac.cy, kourtellaris.christos@ucy.ac.cy}\\
   \IEEEauthorblockA{
     \IEEEauthorrefmark{2}School of Electrical Engineering and Computer Science\\ University of Ottawa, Ontario, Canada\\
     Email: sergey.loyka@uottawa.ca}
 }

\maketitle

\vspace*{-1.0cm}

\begin{abstract}
The capacity of additive Gaussian noise (AGN) channels, $Y_t=X_t+V_t, t=1, \ldots, n$, $\frac{1}{n} {\bf E}\big\{\sum_{t=1}^n |X_t|^2 \big\}\leq \kappa, \kappa \in [0,\infty)$,   with time-invariant channel input feedback strategies,  is characterized and conditions are identified for entropy rates, and limit of average power  to exist,  when the noise is described by {\it stable and unstable} autoregressive models, AR$(c)$, $V_t=cV_{t-1}+ W_t, V_0=v_0, t=1, \ldots, n$, 
where  $c\in (-\infty,\infty)$,  $W_t, t=1,\ldots, n$, is a zero mean, variance $K_W$,  independent Gaussian sequence, independent of $V_0$. For stable   AR$(c), c\in (-1,1)$ the conditions are necessary and sufficient for asymptotic stationarity of the processes $(X_t, Y_t), t=1, 2, \ldots$. 
New   closed form capacity formulas and lower bounds  are derived, for the AR$(c), c\in (-\infty,\infty)$ noise, which are  fundamentally different from existing formulas in the literature, and illustrate   multiple regimes of capacity, as a function of the parameters $(c,K_W,\kappa)$, as follows.\\
1) feedback increases capacity  for the regime, $c^2 \in (1, \infty),$  for   $\kappa > \frac{K_W\big(1+\sqrt{4c^2-3}\big)}{2\big(c^2-1\big)^2}$, \\
2) feedback does not increase capacity for the regime $c^2 \in (1, \infty)$, for $\kappa \leq \frac{K_W\big(1+\sqrt{4c^2-3}\big)}{2\big(c^2-1\big)^2}$, and \\
3)  feedback does not increase capacity for the regime  $c \in [-1,1]$, for $ \kappa \in [0,\infty)$.\\
\noi When compared  to \cite{kim2010}, our capacity formulas for AGN channels driven by a  
 stable AR$(c), c\in (-1,1)$ noise, state that feedback does not increase capacity,  contrary  to the main results of the characterizations of feedback capacity derived in \cite[Theorem~4.1 and Theorem~6.1]{kim2010}, for stationary or asymptotically stationary processes. We show that our disagreement with \cite{kim2010} is mainly attributed to the  identification of   necessary and/or  sufficient conditions, to ensure the asymptotic average power and   entropy rates exist,  expressed in terms of the known  detectability and stabilizability conditions of convergence of generalized difference Riccati equations (DREs) to analogous generalized algebraic Riccati equations (AREs), of Gaussian estimation problems, which are not accounted for in \cite[Theorem~4.1 and Theorem~6.1]{kim2010}.
\end{abstract}



 \section{Introduction, Motivation, and Main Results of the Paper}
\label{intro}
The feedback capacity of  additive Gaussian noise (AGN) channels driven by nonstationary, and stationary limited memory Gaussian noise, is   addressed, since the early 1970's,      in an anthology of  papers, under various assumptions  \cite{butman1969,tienan-schalkwijk1974,wolfowitz1975,butman1976,cover-pombra1989,ozarow1990,yang-kavcic-tatikonda2007,kim2010}. Two of the fundamental problems are related to

\begin{description}
\item[(Q1):] characterizations and   computations  of   feedback capacity of noiseless feedback codes, when the initial state $S_0=s_0$ of the noise is known or not known  to the encoder and the decoder, and 

\item[(Q2):]  bound on feedback capacity, based on linear feedback coding schemes of communicating  Gaussian random variables (RVs), $\Theta: \Omega \rar {\mathbb R}$, and coding schemes of communicating digital messages  $W:\Omega \rar   {\cal M}^{(n)} \tri  \left\{1, 2,\ldots,  \lceil M_n\rceil\right\}$, when the initial state $S_0=s_0$ of the noise is known to the encoder and the decoder.
\end{description}
\CDC{This paper is mainly concerned with question (Q1), for AGN channels driven by stable and unstable noise, when  {\it  channel input strategies are  time-invariant (not necessarily stationary).} The first objective is to identify necessary and/or sufficient conditions for  optimal channel input strategies to ensure the limiting average power and the information rates exist,  for both stable or unstable channel noise, and then to determine if additional conditions are needed for these  limits to be independent of the initial states of the noise.  Such conditions are identified. For stable (resp. unstable)  noise these conditions also  imply  the  input and output processes (resp. the input process and innovations process of the output) are asymptotically stationary, and the limiting average power and the entropy rates do not depend on the initial states of the channel or their  initial distributions.     The second main objective is the calculation of closed form capacity formulas for stable and unstable autoregressive unit memory noise. Such formulas are derived. They  illustrate  multiple regimes of capacity, and that feedback does not increase capacity, when the noise is stable (also for  unstable noise for certain values of power $\kappa$). }
      
\CDC{From  the main  results of this  paper follows   that  several of the  answers related to questions (Q1) and (Q2), given in \cite{kim2010}, that treats stable, stationary or asymptotically stationary  noises, are incorrect due to oversights related to convergence of asymptotic  average power and entropy rates. In particular, the oversights affect the validity of  feedback capacity described   in time-domain \cite[Theorem~6.1]{kim2010}, by a time-invariant optimization problem, with zero  variance   of the innovations of the channel input process; when conditions for convergence are imposed, then the value of feedback capacity stated    in   \cite[Theorem~6.1]{kim2010} is necessarily zero (otherwise the per unit time asymptotic limits of  the  information theoretic characterization of feedback capacity does not exists).
 Some of the   technical  oversights in  \cite{kim2010} are also repeated in \cite{liu-han2017,liu-han2019,gattami2019,li-elia2019}. These are  discussed  at various  parts of this  paper. For the reader's convenience these are clarified in Section~\ref{lite}.} \\
\CDC{To enlarge the scope of the current paper compared to previous literature \cite{yang-kavcic-tatikonda2007,kim2010,liu-han2017,liu-han2019,gattami2019,li-elia2019} that treat feedback capacity for stable noises, we analyze both the feedback and non-feedback capacity,  for stable and unstable noises  (using a unified approach),  although most of our emphasis is on feedback capacity.} To  avoid extensive notation, this  paper is focused on AGN channels, driven by  a correlated Gaussian noise, of the simplest form, the autoregressive unit memory, stable and unstable noise. However, the method of this  paper apply to more general limited memory noise models \cite{yang-kavcic-tatikonda2007,kim2010} (under appropriate assumptions).

\subsection{The Main Problems of Capacity of AGN Channels with Memory}
\label{intro_n}
In this section we introduce the precise mathematical formulation, and the underlying assumptions based on which we derive the results of the paper.  We consider the following time-varying AGN channel. 
\begin{align}
&Y_t=X_t+V_t,  \hst t=1, \ldots, n, \hst \frac{1}{n}  {\bf E}_{v_0} \Big\{\sum_{t=1}^{n} (X_t)^2\Big\} \leq \kappa, \hso \kappa \in [0,\infty) \label{AGN_in}
\end{align}
$X^n = \{X_1, X_2, \ldots, X_n\}$ is  the sequence of channel input random variables (RVs) $X_t :  \Omega \rar {\mathbb R}$,   \\
$Y^n = \{Y_1, Y_2, \ldots, Y_n \}$ is  the sequence of channel output RVs $Y_t :  \Omega \rar {\mathbb R}$,\\
   $V^n = \{ V_1, \ldots, V_n\}$ conditioned on the initial state $V_0=v_0$, is a  sequence of  jointly Gaussian distributed  RVs $V_t :  \Omega \rar {\mathbb R}$,   and $V^n \in N(0,K_{V^n|V_0})$, \\
$V_0=v_0$, is  known to the encoder and decoder\footnote{This assumption is explicit in \cite{yang-kavcic-tatikonda2007},  it is hidden in \cite[Theorem~6.1]{kim2010}, it is explicit in  \cite[Lemma~6.1 and comments above it]{kim2010}; see Section~\ref{app_sect_imp} for clarification.}\\
$N(0,K_{V^n|V_0})$ denotes the distribution of the Gaussian RV $V^n$ conditional on $V_0$, with  zero conditional mean, and conditional variance  $K_{V^n|V_0}$, \\
 ${\bf E}_{v_0}\{\cdot\}$ denotes expectation for fixed initial state  $V_0=v_0$. \\
A  time-varying unit memory Gaussian autoregressive noise, with initial state $V_0=v_0$, is defined by      
\begin{align}
\mbox{AR$(c_t)$}:\hso  \left\{ \begin{array}{l}
 V_t=c_tV_{t-1}+ W_t,\hso  V_0=v_0, \hso t=1, \ldots, n,  \\
W_t\in N(0,K_{W_t}), \hso t=1,\ldots, n, \hso \mbox{indep. Gaussian, indep. of $V_0\in N(0,K_{V_0})$,}\\
K_{V_0} \geq 0, \hso K_{W_t}>0, \hst  c_t \in (-\infty,\infty), \hso t=1, \ldots, n \hso \mbox{are non-random}.
\end{array} \right. \label{ykt_ar_1}
\end{align}
We denote by AR$(c), c\in (-\infty,\infty)$ the restriction of AR$(c_t)$ to the time-invariant  autoregressive noise, i.e., $K_{W_t}=K_{W}, t=1, \ldots, n, c_t=c\in (-\infty,\infty), t=0, \ldots, n$. For stable noise,   AR$(c), c \in (-1,1)$, the variance defined by  $K_{V_t}\tri {\bf E}\big(V_t\big)^2$, satisfies $K_{V_t}= c^2 K_{V_{t-1}}+K_W, K_{V_0}\geq 0, t=1, \ldots, n$. The stable AR$(c)$ noise   is called asymptotically stationary if $\lim_{n \longrightarrow \infty}K_{V_n}= \frac{K_W}{1-c^2},$ for all initial values $K_{V_0}\geq 0$, i.e., $|c|<1$.   \\
AR$(c_t)$ without an initial state is defined by  (\ref{ykt_ar_1}),  for $t=2, \ldots, n$, with  $V_1 \in N(0, K_{V_1}), K_{V_1}\geq 0$, independent of $W_t \in N(0,K_{W_t}), K_{W_t}>0,  t=2, \ldots, n$. Similarly, the stable AR$(c)$ noise without an initial state  is called asymptotically stationary if $K_{V_t}= c^2 K_{V_{t-1}}+K_W, K_{V_1}\geq 0, t=2, \ldots, n$,  converges,  $\lim_{n \longrightarrow \infty}K_{V_n}= \frac{K_W}{1-c^2},$ for all initial values $K_{V_1}\geq 0$, $|c|<1$. That is, the invariant distribution of the noise is  $N(0,\frac{K_W}{1-c^2}), c \in (-1,1)$.

At this stage, we introduce the feedback code and non-feedback code of the AGN channel.
  \\

\begin{definition} Feedback and non-feedback codes\\
\label{def_code}
(a)  A noiseless time-varying feedback code\footnote{A time-varying feedback code means the channel input distributions ${\bf P}_{X_t|X^{t-1}, Y^{t-1}, S_0}, t=1, \ldots, n$ are time-varying.} for the  AGN Channel, is denoted  by     ${\cal C}_{{\mathbb Z}^+}^{fb}\tri \big\{(n, \lceil M_n \rceil, s_0,\kappa, \epsilon_n): n=1,2, \ldots, \big\}$, and  consists of the following  elements and assumptions. \\
(i) The set of uniformly distributed messages $W : \Omega \rar  {\cal M}^{(n)} \tri  \left\{1, 2,\ldots, \lceil M_n \rceil\right\}$. \\
(ii) The set of codewords  of block length $n$,  defined by  the set\footnote{The superscript $e(\cdot) $ on ${\bf E}_{v_0}^e$ is used to denote that the distribution depends on the strategy $e(\cdot)\in {\cal E}_{[0,n]}(\kappa)$.} 
\begin{align}
{\cal E}_{[0,n]}(\kappa) \triangleq & \Big\{ X_1=e_1(W,V_0),X_2=e_2(W,V_0,X_1,Y_1), \ldots, X_n=e_n(W,V_0, X^{n-1}, Y^{n-1}): \nonumber \\
&  \frac{1}{n+1}   {\bf E}_{v_0}^e    \Big( \sum_{i=0}^n  (X_t)^2\Big) \leq \kappa \Big\}. 
\end{align}
(iii) The decoder functions  $(v_0,y^n) \longmapsto d_{n}(v_0,y^n)\in  {\cal M}^{(n)}$, with average  error  probability
\bea
{\bf P}_{error}^{(n)}(v_0) = {\mathbb P}\Big\{d_{n}(V_0,Y^n) \neq W\Big|V_0=v_0 \Big\}= \frac{1}{\lceil M_n \rceil} \sum_{w=1}^{\lceil M_n \rceil} {\bf  P}_{v_0}^e\Big(d_n(V_0,Y^n) \neq W\Big) \leq \epsilon_n. \label{g_cp_4}
\eea
where ${\bf P}_{v_0}^e$ means the distribution depends on $e(\cdot)\in {\cal E}_{[0,n]}(\kappa)$ and  $V_0=v_0$ is fixed.\\
(iv)  ``$X^n$ is causally related to $V^n$'' \cite[page 39, above Lemma~5]{cover-pombra1989}, which is equivalent to the following decomposition of the joint probability distribution of $(X^n, V^n)$ given $V_0$.
\begin{align}
{\bf P}_{X^n, V^n|V_0}
=&{\bf P}_{V^n|V_0} \prod_{t=1}^n  {\bf P}_{X_t|X^{t-1}, V^{t-1},V_0}. \label{g_cp_3}\\
=&{\bf P}_{V^n|V_0} \prod_{t=1}^n  {\bf P}_{X_t|X^{t-1}, Y^{t-1},V_0},  \hst \mbox{by   $Y_t=X_t+V_t$.}\label{g_cp_3_a}
\end{align}
 The coding   rate  is   $r_n\triangleq \frac{1}{n} \log \lceil M_n \rceil$.   Given an initial state $V_0=v_0$,  a rate $R(v_0)$ is called  an {\it achievable rate}, if there exists  a  code  sequence ${\cal C}_{{\mathbb Z}^+}^{fb}$, satisfying
$\lim_{n\longrightarrow\infty} {\epsilon}_n=0$ and $\liminf_{n \longrightarrow\infty}\frac{1}{n}\log \lceil M_n \rceil\geq R(v_0)$.\\
  The operational definition of the  feedback capacity of the  AGN channel, for fixed $V_0=v_0$, is $C(\kappa,v_0)\triangleq \sup \big\{R(v_0): R(v_0) \: \: \mbox{is achievable}\big\}$. \\
(b) A  time-varying code without feedback  for the  AGN Channel, denoted by ${\cal C}_{{\mathbb Z}^+}^{nfb}$, is the restriction of the time-varying feedback code ${\cal C}_{{\mathbb Z}^+}^{fb}$, to the subset    
  ${\cal E}_{[0,n]}^{nfb}(\kappa)\subset {\cal E}_{[0,n]}(\kappa)$, defined   by 
\begin{align}
{\cal E}_{[0,n]}^{nfb}(\kappa) \triangleq & \Big\{ X_1=e_1^{nfb}(W,V_0),X_2=e_2^{nfb}(W,V_0,X_1),  \ldots, X_n=e_n^{nfb}(W,V_0, X^{n-1}): \nonumber \\
&  \frac{1}{n}   {\bf E}_{v_0}^{e^{nfb}}    \Big( \sum_{i=1}^n  (X_t)^2\Big) \leq \kappa \Big\}.  
\end{align}
\end{definition}

Since the code sequence  ${\cal C}_{{\mathbb Z}^+}^{fb}$   depends on $V_0=v_0$, then in general, the   rate $R(v_0)$, and also $C(\kappa, v_0)$ depend on  $v_0$. 


{\it Feedback Capacity of Time-Varying Channel Input Strategies.} Consider the  feedback code of Definition~\ref{def_code}.(a), i.e., ${\cal C}_{{\mathbb Z}^+}^{fb}$. Given the elements of the set ${\cal E}_{[0,n]}(\kappa)$, by the  maximum entropy principle of Gaussian distributions, similar to   \cite{cover-pombra1989},  the upper bound holds\footnote{The superscript $e$ means the underlying distributions are induced by the channel distribution and the elements of the set $e(\cdot)\in {\cal E}_{[0,n]}(\kappa)$.}.
\begin{align}
I^e(W; Y^n|v_0) \leq H(Y^n|v_0)-H(V^n|v_0), \hso \mbox{if $H(Y^n|v_0)$ is evaluated at a Gaussian  ${\bf P}_{Y^n|V_0}$ } \label{cp_ub}
\end{align}
where $H(X|s)$ stands for differential entropy of RV $X$ conditioned on the initial state $S=s$.  Further, similar to \cite{cover-pombra1989},   the upper bound in (\ref{cp_ub}) is achieved  if the input $X^n$ is  jointly Gaussian for fixed $V_0=v_0$, satisfies the average power constraint,  and respects  (\ref{g_cp_3}). 
By the chain rule of mutual information, $I^e(W;Y^n|v_0)=\sum_{t=1}^n I^e(W;Y_t|Y^{t-1},v_0)$, and the data processing inequality, follows,
\begin{align}
\sup_{{\cal E}_{[0,n]}(\kappa)}I^e(W;Y^n|v_0)\leq 
 &\sup_{ {\bf P}_{X_t|X^{t-1}, Y^{t-1}, V_0}, t=1, \ldots, n: \hso \frac{1}{n}  {\bf E}_{v_0}\big\{\sum_{t=1}^n \big(X_t\big)^2 \big\}\leq \kappa} H(Y^n|v_0)- H(V^n|v_0), \hso \mbox{by (\ref{AGN_in})} \label{FTFIC_1a}
\end{align} 
where the supremum in the right hand side of (\ref{FTFIC_1a})  is taken over conditionally Gaussian time-varying distributions ${\bf P}_{X_t|X^{t-1}, Y^{t-1}, V_0}, t=1, \ldots, n$, such that $(X^n, Y^n)$ are jointly Gaussian for fixed $V_0=v_0$, and  (\ref{g_cp_3}) is respected.\\
Define,  as in \cite{cover-pombra1989}, the $n-$finite transmission feedback information (FTFI) capacity  of code ${\cal C}_{{\mathbb Z}^+}^{fb}$,  by
\begin{align}
C_n(\kappa,v_0) \tri   \sup_{ {\bf P}_{X_t|X^{t-1}, Y^{t-1}, V_0}, t=1, \ldots, n: \hso \frac{1}{n}  {\bf E}_{v_0}\big\{\sum_{t=1}^n \big(X_t\big)^2 \big\}\leq \kappa}  H(Y^n|v_0)-H(V^n|s_0) \label{FTFIC_1in}
\end{align}
provided the supremum element exists in the set.  From the  converse and direct coding theorems  in \cite[Theorem~1]{cover-pombra1989}, it then follows that the  characterization of  feedback capacity of code  ${\cal C}_{{\mathbb Z}^+}^{fb}$,   is given  by  
\begin{align}
 C(\kappa,v_0) =  \lim_{n \longrightarrow \infty} \frac{1}{n} C_n(\kappa,v_0) \label{FTFI_2}
\end{align}
provided the limit exists in $[0,\infty)$ i.e., it is finite.

{\it Capacity Without Feedback of Time-Varying  Channel Input Strategies.}
Let $C_n^{nfb}(\kappa,v_0)$ be defined as in (\ref{FTFIC_1in}), with the time-varying feedback distributions ${\bf P}_{X_t|X^{t-1}, Y^{t-1}, V_0}, t=1, \ldots, n$, replaced by the time-varying non-feedback distributions ${\bf P}_{X_t|X^{t-1}, V_0}, t=1, \ldots, n$, called $n-$finite transmission without feedback information (FTwFI) capacity. The  non-feedback  capacity  of the code ${\cal C}_{{\mathbb Z}^+}^{nfb}$ of Definition~\ref{def_code}.(b), is characterized  by  $C^{nfb}(\kappa,v_0) =  \lim_{n \longrightarrow \infty} \frac{1}{n} C_n^{nfb}(\kappa,v_0)$, provided the limit exists.

To the best of our knowledge, no  closed form formulas are available in the literature for $C_n(\kappa,v_0)$, $C(\kappa,v_0)$,  even for  AGN channels driven by a time-invariant  AR$(c), c\in (-1,1)$ noise. \CDC{Often, past literature is focused on stationary or  stable noise, stationary or asymptotically stationary joint input and output process, and investigates  the  variant of 
(\ref{FTFI_2}) with the limit and supremum operations interchanged  \cite[Theorem~7 and Corollary~7.1]{yang-kavcic-tatikonda2007},   \cite[Theorem~3.2]{kim2010} and  \cite{liu-han2017,liu-han2019,gattami2019,li-elia2019}. However, as it will be apparent in this paper,  conditions for existence of the limiting  average power and entropy rates are not correctly identified in \cite{yang-kavcic-tatikonda2007,kim2010,liu-han2017,liu-han2019,gattami2019,li-elia2019}, and this oversight led to incorrect characterizations of feedback capacity by a time-invariant optimization problem (i.e., \cite[Theorem~6.1]{kim2010}).}

 This brings us to the next definition  of capacity, where conditions for existence of the limits of average power and entropy rates are characterized, and they part of our problem formulation. 

{\it  Feedback Capacity of Time-Invariant Channel Input Strategies.} \CDC{We consider (\ref{FTFIC_1in}), (\ref{FTFI_2}) with the per unit time limit and  supremum operations interchanged,  and time-invariant codes and induced distributions, called strategies.
To ensure the feedback capacity (to be defined shortly) is well-posed,  we introduce the following condition:}
   
\CDC{
{\bf (C1)} Channel input strategies  with feedback are  time-invariant, the consistency condition   (\ref{g_cp_3}) holds, and the following limits exist:  \\
(i) $\lim_{n \longrightarrow \infty} \frac{1}{n} {\bf E}_{v_0}\big\{\sum_{t=1}^n \big(X_t\big)^2 \big\}\in [0,\infty)$,  (ii)
$\lim_{n \longrightarrow \infty}\frac{1}{n}\big\{H(Y^n|v_0)-H(V^n|v_0)\big\}\in [0,\infty)$.} 

We define the operational information  feedback capacity under condition {\bf (C1)}, as follows.  
\begin{align}
&C^\infty(\kappa,v_0) 
\tri  \sup_{\lim_{n \longrightarrow \infty} \frac{1}{n} {\bf E}\big\{\sum_{t=1}^n \big(X_t\big)^2 \big\}\leq \kappa, \:\; \mbox{subject to  {\bf (C1)}}} \lim_{n \longrightarrow \infty} \frac{1}{n} \Big\{ H(Y^n|v_0)-H(V^n|s_0)\Big\}  \label{inter}
\end{align}
where the supremum is taken over all jointly Gaussian  channel input processes $X^n, n=1,2,\ldots$ with feedback,   or distributions with feedback ${\bf P}_{X_t|X^{t-1}, Y^{t-1}, V_0}, t=1,2,  \ldots$, such that $(X^n,Y^n), n=1,2,\ldots,$  is jointly Gaussian, for $V_0=v_0$,   and  {\bf (C1)} holds.  \\
\CDC{In the definition of $C^\infty(\kappa,v_0)$ we do not assume joint stationarity of  $(X^n, Y^n, V^n), n=1,2,\ldots$,  because this is not required for the  limits to exist. Similarly, we do not impose conditions  to ensure $C^\infty(\kappa,v_0)=C^\infty(\kappa), \forall v_0$, i.e., is independent of the initial state $V_0=v_0$, and hence $\int C^\infty(\kappa,v_0){\bf P}_{V_0}(dv_0)$ is independent of $V_0 \in N(0, K_{V_0}), \forall K_{V_0}\geq 0$. Rather, for stable (resp. unstable) noise, we first identify necessary and/or sufficient conditions for {\bf (C1)} to hold, and then determine if  additional conditions are needed to ensure the optimal channel input process is asymptotically stationary, and such that it induces asymptotic stationarity of the channel output process (resp. innovations process of the output),   and    $C^\infty(\kappa,v_0)=C^\infty(\kappa), \forall v_0$. }\\
\CDC{For time-invariant stable or unstable noise, it will become apparent,  from properties of estimation theory of linear Gaussian systems (introduced at latter parts of the paper), that the well-known detectability and stabilizability conditions  of generalized  Kalman-filter equations \cite{caines1988,kailath-sayed-hassibi}  (see Definition~\ref{def:det-stab}), are necessary and sufficient for  the limits to exist, and for $C^\infty(\kappa,v_0)$ to be well-posed.  It will also be apparent that in prior literature  \cite{yang-kavcic-tatikonda2007,kim2010,liu-han2017,liu-han2019,gattami2019,li-elia2019},  conditions for existence of the limits in  {\bf (C1)} are not  imposed, and this omission leads to incorrect characterizations  of feedback rates by a  time-invariant optimization problem, which is not the asymptotic limit of  (\ref{inter}) (i.e., \cite[Theorem~6.1]{kim2010}).}


{\it Capacity Without Feedback of Time-Invariant Channel Input Strategies.}
 Similar to (\ref{inter}), we also analyze  the non-feedback capacity analog,   under condition ${\bf (C1)}$, which is defined as follows.   
\begin{align}
&C^{\infty,nfb}(\kappa,v_0) 
\tri  \sup_{\lim_{n \longrightarrow \infty} \frac{1}{n} {\bf E}\big\{\sum_{t=1}^n \big(X_t\big)^2 \big\}\leq \kappa, \:\; \mbox{subject to  {\bf (C1)}}} \lim_{n \longrightarrow \infty} \frac{1}{n} \Big\{ H(Y^n|v_0)-H(V^n|s_0)\Big\} \label{inter_nfb}
\end{align}
where the supremum is taken over all jointly Gaussian   channel input processes $X^n,n=1,2,\ldots,$ without feedback  or distributions without feedback, denoted  by ${\bf P}_{X_t|X^{t-1}, V_0}, t=1,2, \ldots$, such that  $(X^n,Y^n), n=1,2,\ldots$  is jointly Gaussian for $V_0=v_0$, {\bf (C1)} holds (with ${\bf P}_{X_t|X^{t-1}, Y^{t-1},  V_0}$ replaced by ${\bf P}_{X_t|X^{t-1}, V_0}$), and (\ref{g_cp_3}) is respected, for $n=1,2,\ldots$. To our knowledge,  for AGN channels driven by an  unstable  noise $V^n$,  no closed form expression of non-feedback capacity  is ever reported in the literature. \\
Given the above formulation, in this paper we obtain answers to the  questions listed under Problem~\ref{prob_1_in}.\\

\begin{problem} Main problem\\
\label{prob_1_in}
Consider  $C^\infty(\kappa,v_0)$ defined by (\ref{inter}), and $C^{\infty,nfb}(\kappa,v_0)$ defined by (\ref{inter_nfb}),   of the AGN channel driven by a time-invariant stable and unstable, AR$(c)$ noise, i.e., $c\in (-\infty,\infty)$: \\
(a) What are  necessary and/or sufficient conditions for  {\bf (C1)} to hold? \\
\CDC{(b) What are  necessary and/or sufficient conditions for asymptotic stationarity of the process $(X^n, Y^n), n=1,2,\ldots$ or of only the marginal process $X^n$, that achieve $C^\infty(\kappa,v_0)$, and for $C^\infty(\kappa,v_0)=C^{\infty}(\kappa), \forall v_0$ i.e., to be  independent of initial data?}\\
\CDC{(c) What are the characterizations and  closed form formulas of  feedback capacity $C^\infty(\kappa,v_0)$?\\ 
(d) How do we extract simple lower bounds on non-feedback capacity, $C^{\infty,nfb}(\kappa,v_0)$ from the chararacterizations of feedback capacity?}
 \end{problem}

\CDC{To address Problem~\ref{prob_1_in}  we make  use of the  identities 
\begin{align}
\lim_{n \longrightarrow \infty} \frac{1}{n} H(Y^n|v_0)=&\lim_{n \longrightarrow \infty} \frac{1}{n} \sum_{t=1}^n H(Y_t|Y^{t-1},v_0)=\lim_{n \longrightarrow \infty} \frac{1}{n} \sum_{t=1}^n H(Y_t-{\bf E}\big\{Y_t|Y^{t-1}, v_0\big\}|Y^{t-1},v_0)  \label{lim_ex} \\
=&\lim_{n \longrightarrow \infty} \frac{1}{n} \sum_{t=1}^n H(I_t), \hso I_t\tri Y_t-{\bf E}\big\{Y_t|Y^{t-1}, v_0\big\} \hso \mbox{an indep. innovations process.}
\end{align}
Then we identify necessary and/or sufficient conditions for the  limits $\lim_{n \longrightarrow \infty} \frac{1}{n} \big\{H(Y^n|v_0)-H(V^n|v_0)\big\}$  and    $\lim_{n \longrightarrow \infty} \frac{1}{n} {\bf E}_{v_0}\big\{\sum_{t=1}^n \big(X_t\big)^2 \big\}$ to exist, i.e., in $[0,\infty)$. }  

\subsection{Methodology of the Paper} 
Our methodology is based on the following main steps.\\
{\it Step 1.} We characterize $C_n(\kappa,v_0)$ defined by (\ref{FTFIC_1in}), i.e., the $n-$FTFI capacity, of the  AGN channel driven by a time-varying AR$(c_t)$    noise. We also give  a lower bound on  the characterization of the  $n-$FTwFI capacity $C_n^{nfb}(\kappa,v_0)$, using  a Gaussian  channel input process, which  is realized by an AR$(\overline{\Lambda}_t)$ process,
\bea
X_t=\overline{\Lambda}_t X_{t-1}+Z_t,\hso  X_1=Z_1, \hso \overline{\Lambda}_t \in (-\infty,\infty), \hso  t=2, \ldots, n \label{nofeedback_input}
\eea
where 
$Z^n$ an independent Gaussian sequence,  independent of $(V^n,V_0)$. \\ 
{\it Step 2.} We characterize the feedback capacity $C^\infty(\kappa,v_0)=C^{\infty}(\kappa), \forall v_0$ defined by (\ref{inter}), and we give  a lower bound  on the characterization of $C^{\infty,nfb}(\kappa,v_0)$ defined by (\ref{inter_nfb}),  of the  AGN channel driven by a time-invariant stable or unstable noise, AR$(c), c\in (-\infty,\infty)$. Our analysis  identifies  necessary and/or sufficient conditions for condition {\bf (C1)} to hold, expressed  in terms of the convergence properties of {\it generalized difference Riccati equations (DREs) and algebraic Riccati equations (AREs)}, of estimating the channel state, that is, the noise $V^n$,   from the channel output process $Y^n$,  and the initial state $V_0=v_0$, for $n=1,2,\ldots$.  This step is analogous to  \cite[Theorem~4.1]{kourtellaris-charalambousIT2015_Part_1}, although the models considered in \cite{kourtellaris-charalambousIT2015_Part_1} involve a classical control DRE and ARE.\\
{\it Step 3.} We derive a closed form formula of   feedback capacity $C^\infty(\kappa,v_0)=C^\infty(\kappa), \forall v_0$, that shows there are {\it multiple regimes} of capacity, and these regimes depend on the parameters $(c,K_W, \kappa)$. 
Our feedback capacity formulae $C^\infty(\kappa)$ for AGN channels driven by stable noise AR$(c), c\in (-1,1)$ is    fundamentally different   from the  one  obtained using the characterization of feedback capacity in  \cite[Theorem~6.1]{kim2010}. We show this difference  is   mainly attributed to the appended  detectability and stabilizability conditions on the characterization of our feedback capacity, \CDC{to ensure the optimal channel input process $X^n, n=1, 2,\ldots$  is such that the limits,  $\lim_{n \longrightarrow \infty} \frac{1}{n} {\bf E}_{v_0}\big\{\sum_{t=1}^n \big(X_t\big)^2 \big\}\in [0,\infty)$,  $\lim_{n \longrightarrow \infty}\frac{1}{n}\big\{H(Y^n|v_0)-H(V^n|v_0)\big\}\in [0,\infty)$  exist, and  the joint  process $(X^n, Y^n), n=1,2,\ldots$  is asymptotically  stationary, which are not accounted for, in  \cite[Theorem~6.1]{kim2010}.}  \\
We also give an achievable lower bound on the non-feedback capacity  $C^{\infty,nfb}(\kappa,v_0)$, based on (\ref{nofeedback_input}), with $\overline{\Lambda}_t=0,  \forall t$, i.e., $X_t=Z_t$, $Z^n, n=1, \ldots$   an  independent and identically distributed (IID) sequence, and holds for stable and unstable AR$(c), c \in (-\infty,\infty)$ noise.\\
 {\it Step 4.} 
\CDC{We  show the characterization of feedback capacity given in  \cite[Theorem~6.1, $C_{FB}$]{kim2010} (i.e., the limiting expression of $C^\infty(\kappa,v_0)$,  without the stabilizability condition), and with  a zero variance of the innovations part of the channel input process (see \cite[Lemma~6.1]{kim2010}), gives an incorrect value of    $C_{FB}$. However, when the  stabilizability conditions is imposed (which is necessary for asymptotic limits to exists) then it follows the value  $C_{FB}=0$; clearly an incorrect limiting value of  (\ref{inter}). }

We structured the paper as follows.\\
In Section~\ref{sect:problem}, we derive  the characterization of the $n-$FTFI capacity,  and the lower bound on the   characterization of the $n-$FTwFI capacity, for AGN channels driven by the   AR$(c_t)$ noise (Section~\ref{sect:ar1_r}), and  present a preliminary   elaboration on technical issues that are integral part of  capacity definition (\ref{inter}).      \\ 
 In Section~\ref{sect:q1}, we present the  derivations of feedback capacity formulas of $C^{\infty}(\kappa,v_0)=C^{\infty}(\kappa), \forall v_0$, i.e.,  (\ref{inter}), and the  achievable lower bounds on the non-feedback capacity 
$C^{\infty,nfb}(\kappa,v_0)$, for stable and unstable noise, using    the asymptotic analysis of   generalized Kalman-filters \cite{kailath-sayed-hassibi,caines1988}.  
  \\
\CDC{In Section~\ref{lite} we collect and discuss some of the oversights in  \cite{kim2010}, which are also repeated in \cite{liu-han2017,liu-han2019,gattami2019,li-elia2019}.}
  
%



For the rest of this section, we follow a rather unconventional presentation; we  
present a summary of the main results of the paper, with informative simulations of achievable rates.

\subsection{Summary of the Main Results, Graphical Evaluations, and Discussion} 
\label{sect:main}
Below,  we summarize the main results of the paper, and present graphical evaluations of  achievable rates. 

(R1) {\it Multiple Regimes of Capacity.} There are multiple regimes of feedback capacity $C^{\infty}(\kappa,v_0)=C^{\infty}(\kappa), \forall v_0$ defined by (\ref{inter}), and these regimes  depend on the parameters $(c,K_W, \kappa)$,  as follows.
\begin{align}
&\mbox{Regime 1}:   \hso {\cal K}^\infty(c,K_W) \tri \Big\{ \kappa\in [0, \infty): 1 < c^2 < \infty, \hso \kappa > \frac{K_W\big(1+\sqrt{4c^2-3}\big)}{2\big(c^2-1\big)^2}  \Big\}, \label{reg_1} \\
&\mbox{Regime 2:} \hso  {\cal K}^{\infty,nfb}(c,K_W) \tri \Big\{ \kappa\in [0,\infty): 1 < c^2 <\infty,\hso  \kappa \leq \frac{K_W\big(1+\sqrt{4c^2-3}\big)}{2\big(c^2-1\big)^2}  \Big\}, \label{reg_3}\\
&\mbox{Regime 3:} \hso  {\cal K}^{\infty,nfb}(c) \tri \Big\{\kappa\in [0,\infty): 0 \leq c^2 \leq 1\Big\}. \label{reg_2}
\end{align}
{\it Regime 1.} This corresponds to the unstable AR$(c)$ noise, i.e., $c\in (-\infty, 1)\bigcup (1,\infty)$, there is a threshold effect on power, $\kappa> \kappa_{min}\tri \frac{K_W\big(1+\sqrt{4c^2-3}\big)}{2\big(c^2-1\big)^2}$, and feedback increases capacity.\\
{\it Regime 2.} This corresponds to the unstable AR$(c)$ noise, i.e., $c\in (-\infty, 1)\bigcup (1,\infty)$,  there is a threshold effect on power, $\kappa \leq  \kappa_{max}\tri \frac{K_W\big(1+\sqrt{4c^2-3}\big)}{2\big(c^2-1\big)^2}$,  feedback does not increase  capacity, $C^\infty(\kappa,v_0)= C^{\infty,nfb}(\kappa), \forall v_0$.\\
{\it Regime 3.} This  corresponds to the marginally stable AR$(c)$ noise, i.e., $c\in [-1, 1]$,   there is no threshold effect on power, feedback does not increase capacity, $C^\infty(\kappa,v_0)= C^{\infty,nfb}(\kappa), \forall v_0$. 
%
%
%

(R2) {\it Feedback capacity $C^\infty(\kappa,v_0)=C^{\infty}(\kappa), \forall v_0$}.  For Regime 1 the following hold.\\
The  optimal time-invariant, channel input distribution for $C^\infty(\kappa,v_0)$ defined by (\ref{inter}), is induced by a  jointly Gaussian channel input process $X^{n,o}$ (not necessarily stationary), with a representation 
\begin{align}
&X_t^o= \overline{\Lambda}^\infty \Big(X_{t-1}^o - {\bf E}_{v_0}\Big\{X_{t-1}^o\Big|Y^{o,t-1}\Big\}\Big) +Z_t^o, \hso t=2, \ldots, n, \label{in_1}   \\
&\hst  =\Lambda^\infty \Big(V_{t-1} - {\bf E}_{v_0}\Big\{V_{t-1}\Big|Y^{o,t-1}\Big\}\Big) +Z_t^o, \hst \overline{\Lambda}^\infty=-\Lambda^\infty, \label{in_2} \\
&X_1^o =\Lambda^\infty \Big(v_{0} - {\bf E}_{v_0}\Big\{V_{0}\Big\}\Big) +Z_1^o=Z_1^o, \\
&\frac{1}{n}  {\bf E}_{v_0} \Big\{\sum_{t=1}^{n} \big(X_t^o\big)^2\Big\}=\frac{1}{n} \sum_{t=1}^n \Big\{ \big(\Lambda^\infty\big)^2 K_{t-1}^o+ K_{Z}^\infty\Big\}  \leq \kappa, \label{cp_e_ar2_s1_new_in}\\
&Z_t^o  \hso  \mbox{indep. of $(X^{o,t-1}, V^{t-1}, Y^{o,{t-1}}, V_0)$},   \hso \mbox{ $Z_t^o \in N(0, K_{Z}^\infty)$},\hso \mbox{ for $t=1,\ldots,n$,}\\
&Z^{o,n} \hso \mbox{indep. of $(V^n,V_0)$,}\\
&K_t^o \tri {\bf E}_{v_0} \Big(E_t^o\Big)^2, \hso     \hso  E_t^o\tri V_{t} - {\bf E}_{v_0}\Big\{V_{t}^o\Big|Y^{o,t}\Big\},    \hso t=1, \ldots, n, \hso K_0^o=0,  \label{error_in}  
\end{align}
\begin{align}
&\widehat{V}_t^o\tri {\bf E}_{v_0}\Big\{V_{t}\Big|Y^{o,t}\Big\}, \hso \mbox{satisfy the  generalized Kalman-filter equations  (\ref{Q_1_8_s1_new})-(\ref{Q_1_10_s1_new})}. 
\end{align}
The feedback capacity is characterized by 
\begin{align}
&C^\infty(\kappa,v_0)
\tri
 \sup_{\big(\Lambda^\infty, K_{Z}^\infty\big): \hso \lim_{n \longrightarrow \infty} \frac{1}{n} \sum_{t=1}^n  \big(\Lambda^\infty\big)^2 K_{t-1}^o+ K_{Z}^\infty   \leq \kappa} \lim_{n \longrightarrow \infty}\frac{1}{2n}   \sum_{t=1}^n   \log\Big( \frac{\big(\Lambda^\infty+c\big)^2 K_{t-1}^o  + K_{Z}^\infty +K_{W}}{K_{W}}\Big),\label{int_fb} \\
  & \hst \hst \hst =C^\infty(\kappa),\hso  \forall v_0, \hso  \mbox{for values of power $\kappa \in {\cal K}^\infty(c,K_W)$},  \label{reg_1_in}\\
&\mbox{subject to: $K_Z^\infty\geq 0$ and    $K_t^o, t=1, \ldots,n$ satisfies the generalized DRE,}\nonumber \\
&K_{t}^o=  c^2 K_{t-1}^o  + K_{W} -\frac{ \Big( K_{W} + c K_{t-1}^o\Big(\Lambda^\infty + c \Big)\Big)^2}{ \Big(K_{Z}^\infty+ K_{W} + \Big(\Lambda^\infty + c \Big)^2 K_{t-1}^o\Big)}, \hst  K_t^o \geq0,  \hso K_{0}^o=0, \hst t=1, \ldots, n. \label{DRE_in} 
 \end{align}
The feedback capacity is computed from the time-invariant optimization problem\footnote{For $\kappa \in {\cal K}^\infty(c,K_W)$  the limit in (\ref{int_fb}) converges to a finite number and the supremum exists and it is finite.}
\begin{align}
&C^{\infty}(\kappa,v_0)=C^{\infty}(\kappa)\tri \sup_{(K^\infty, K_Z^\infty): \: \big(\Lambda^\infty\big)^2 K^\infty+ K_{Z}^\infty  \leq \kappa} \frac{1}{2} \log\Big( \frac{\big(\Lambda^{\infty}+c\big)^2 K^{\infty}  +K_Z^{\infty} +K_{W}}{K_{W}}\Big), \;  \kappa \in {\cal K}^\infty(c,K_W),  \label{int_nf_8_a}\\
&\mbox{subject to: $K_Z^\infty\geq 0$ and     $K^\infty$ is the unique $K^\infty\geq 0$ stabilizing solution of the  generalized ARE}\nonumber \\
&K^\infty=  c^2 K^\infty  + K_{W} -\frac{ \Big( K_{W} + c K^\infty\Big(\Lambda^\infty + c \Big)\Big)^2}{ \Big(K_{Z}^\infty+ K_{W} + \Big(\Lambda^\infty + c \Big)^2 K^\infty\Big)}, \hst K^\infty \geq 0. \label{ARE_in}
\end{align}
Here ``stabilizing'' means $K^\infty\geq 0$ is such that the eigenvalue of the linear recursion of the estimation error $E_n^o\tri V_n- \widehat{V}_n^o, n=1,2,\ldots$, lies inside $(-1,1)$, hence it converges in mean-square-sense. \\
For Regime 1 (see  (\ref{reg_1}) which means $c^2 \in (1,\infty)$) the closed form capacity formula    is
\begin{align}
C^\infty(\kappa)= & \frac{1}{2}\log\Big( \frac{c^2\Big( \big(c^2-1\big)\kappa +K_W\Big)}{\Big(c^2-1\Big)K_W}\Big), \hso \mbox{for all} \hso  \kappa \in {\cal K}^\infty(c,K_W) \label{fb_1} \\
=& \log |c| +\frac{1}{2}\log\Big(\frac{1}{ c^2-1}+\gamma\Big),  \hso \gamma \tri \frac{\kappa}{K_W} \label{gamma} \\
\geq &  \log |c| \in (0,\infty) \label{thr}
\end{align}
and it is achieved by  the  unique values $(K^{\infty,*}, \Lambda^{\infty,*}, K_Z^{\infty,*})$:
\begin{align}
&K^{\infty,*} = \frac{\kappa\big(c^2-1\big)^2-K_W}{c^2\big(c^2-1\big)} \in (0,\infty), \\
&\Lambda^{\infty,*}= \frac{c K_W}{\kappa\big(c^2-1\big)^2-K_W}\in (-\infty,\infty),\\
&K_Z^{\infty,*}= \frac{\kappa\big(c^2-1\big)\Big(\kappa \big(c^2-1\big)^2-K_W\Big)-K_W^2}{\big(c^2-1\big)\Big(\kappa \big(c^2-1\big)^2-K_W\Big)} \in (0,\infty).
\end{align}
Since for Regime 1, $\kappa \in {\cal K}^\infty(c,K_W)$, and  $c^2 \in (1,\infty)$, inequality (\ref{thr}) indicates a threshold effect on the achievable rate for the unstable noise. From (\ref{gamma}), the behavior of $C^\infty(\kappa)$ for $\gamma$ large is 
\begin{align}
 C^\infty(\kappa) \simeq \log |c| + \frac{1}{2} \log \gamma, \hst \mbox{for large   $\gamma$.}
\end{align}
(R3) {\it Capacity without feedback   $C^{\infty,nfb}(\kappa,v_0)=C^{\infty,nfb}(\kappa), \forall v_0$}. For Regime 2 and Regime 3,   feedback does not increase capacity. That is, if $\kappa \in {\cal K}^{\infty,nfb}(c,K_W)$ or $\kappa \in {\cal K}^{\infty,nfb}(c)$, there does not exist  channel input process such that the limit in (\ref{int_fb}) converges to a nonzero value.  On the other hand, for $\kappa \in {\cal K}^{\infty,nfb}(c,K_W)$ or $\kappa \in {\cal K}^{\infty,nfb}(c)$,  there exists a channel input without feedback, in particular, with  $\overline{\Lambda}^\infty=\Lambda^\infty=0$ in (\ref{in_1}) and (\ref{in_2}), such  that the limit in (\ref{inter_nfb}) converges to a nonzero value, for the AGN channel driven by a stable and unstable AR$(c),c\in (-\infty, \infty)$  noise.  

(R4) {\it Lower bound  on capacity without feedback $C^{\infty,nfb}(\kappa,v_0)$.} An achievable lower bound  on the characterization of nonfeedback capacity, is obtained,  for the AGN channel driven by a stable and unstable AR$(c),c\in (-\infty, \infty)$  noise, which corresponds to a unit memory channel input, $X_t^o=\overline{\Lambda}^\infty X_{t-1}^o+Z_t^o, X_1^o=Z_1^o,  t=2, \ldots,$ and holds   for any  $\kappa \in [0,\infty)$.   Further, another  achievable  lower bound on the non-feedback capacity  for the AGN channel driven by a stable and unstable AR$(c),c\in (-\infty, \infty)$  noise,    is obtained from (R2),   by letting $\overline{\Lambda}^\infty=\Lambda^\infty=0$, which holds for any  $\kappa \in [0,\infty)$, with corresponding IID channel input, $X_t^o=Z_t^o, t=1, \ldots,$ $K_Z^\infty=\kappa$.
  The  lower bound based on the IID channel input  is given by 
\begin{align}
 C^{\infty,nfb}(\kappa,v_0) \geq & \;C_{LB}^{\infty,nfb}(\kappa) \tri \frac{1}{2} \log\Big( \frac{c^2 K^{\infty,*}  + \kappa +K_{W}}{K_{W}}\Big), \hso \forall \kappa \in [0,\infty), \label{int_nf_8}\\
=&  \frac{1}{2} \log\Big(\frac{ \kappa \Big(1+c^2\Big)+K_W + \sqrt{\Big(\kappa\big(1-c^2\big) +K_W\Big)^2+4c^2 K_W \kappa}}{2K_W}      \Big),& c \in (-\infty, \infty)
\label{int_nf_8_aa} 
\end{align}
where $K^{\infty,*} \tri \lim_{n \longrightarrow \infty} K_n^{o} \in [0,\infty)$ is the unique nonnegative stabilizing solution of (\ref{ARE_in}), that corresponds to  the optimal strategy $(\Lambda^{\infty,*}, K_Z^{\infty,*})=(0,\kappa)$, and    given by  
\begin{align}
&K^{\infty,*}= \left\{ \begin{array}{lll}  \frac{- \kappa \Big(1-c^2\Big)-K_W + \sqrt{\Big(\kappa\big(1-c^2\big) +K_W\Big)^2+4c^2 K_W \kappa}}{2c^2}, & c\neq 0, &\kappa \in [0,\infty),  \\
\frac{\kappa K_W}{\kappa+K_W}, & c=0, & \kappa \in [0,\infty),  \end{array} \right. 
 \label{int_nf_9}\\
& \overline{\Lambda}^\infty=\Lambda^{\infty,*}=0, \hso K_Z^{\infty,*}=\kappa. \label{int_nf_9_a}
\end{align} 

(R5) Numerical evaluations  of the feedback capacity $C^\infty(\kappa)$ for $\kappa \in {\cal K}^\infty(c,K_W)$, i.e. regime 1,  based on (\ref{fb_1}), and the lower bound of non-feedback capacity $C_{LB}^{\infty,nfb}(\kappa)$ for $\kappa \in [0,\infty)$ based on (\ref{int_nf_8_aa}),  are shown in Figure~\ref{fig:ar1cap}. These   illustrate that feedback capacity $C^\infty(\kappa)$ for $\kappa \in {\cal K}^\infty(c,K_W)$ is an increasing function of the parameter, $|c| \in (1, \infty)$, that is, the more unstable the AR$(c)$ noise the   higher the value of  capacity $C^\infty(\kappa)$. Further, the lower bound on non-feedback capacity $C_{LB}^{\infty,nfb}(\kappa)$ is achievable  for all $\kappa \in [0,\infty)$, for stable and unstable  AR$(c), c\in (-\infty,\infty)$ noise.
As illustrated in  Figure~\ref{fig:ar1cap}, for values of $|c|> 1$, a discontinuity occurs at $\kappa=\frac{K_W\big(1+\sqrt{4c^2-3}\big)}{2\big(c^2-1\big)^2}$. \CDC{Note, that this occurs since  i) for $|c| \in (1, \infty)$ and $\kappa\leq\frac{K_W\big(1+\sqrt{4c^2-3}\big)}{2\big(c^2-1\big)^2}$, and ii) for $|c| \in (0, 1]$ and $\kappa \in [0,\infty)$, the depicted curves correspond to  lower bounds on non-feedback capacity $C_{LB}^{\infty,nfb}(\kappa)$,  and not the actual non-feedback capacity, whereas for  $|c| \in (1, \infty)$ and $\kappa>\frac{K_W\big(1+\sqrt{4c^2-3}\big)}{2\big(c^2-1\big)^2}$ the depicted curves correspond to the feedback capacity.} 

\begin{figure}
\centering
  \includegraphics[width=0.9\linewidth]{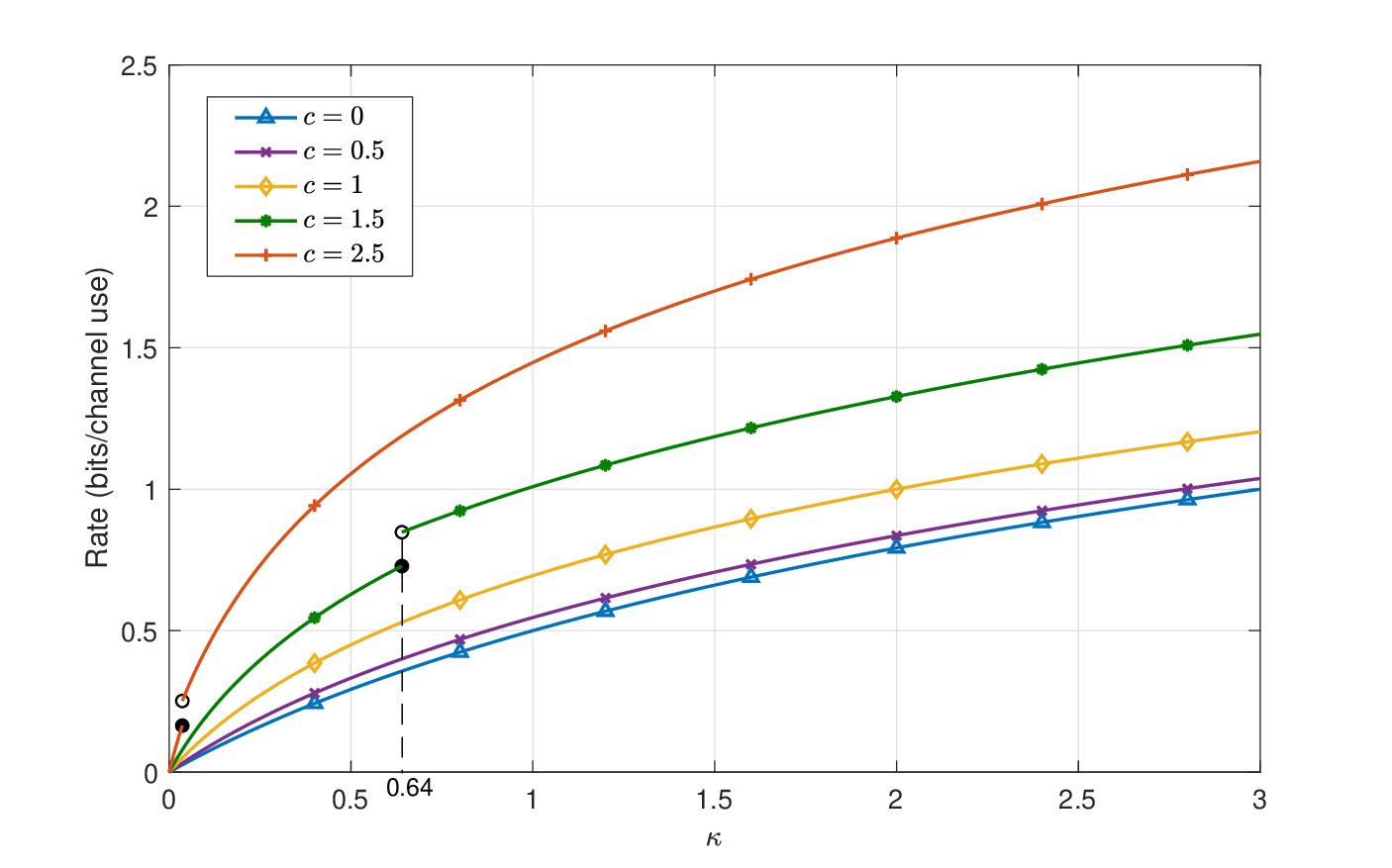}
  \caption{Feedback capacity $C^\infty(\kappa)$ for $\kappa \in {\cal K}^{\infty}(c,K_W)$ based on (\ref{fb_1}) and lower bound on non-feedback capacity $C_{LB}^{\infty,nfb}(\kappa)$ for  $\kappa \in [0,\infty)$ based on (\ref{int_nf_8_aa}),    of the AGN channel driven by AR$(c)$ noise, for various values of $c\in (-\infty, \infty)$ and $K_W=1$.}
  \label{fig:ar1cap}
\end{figure}

\begin{figure}
\centering
  \includegraphics[width=0.8\linewidth]{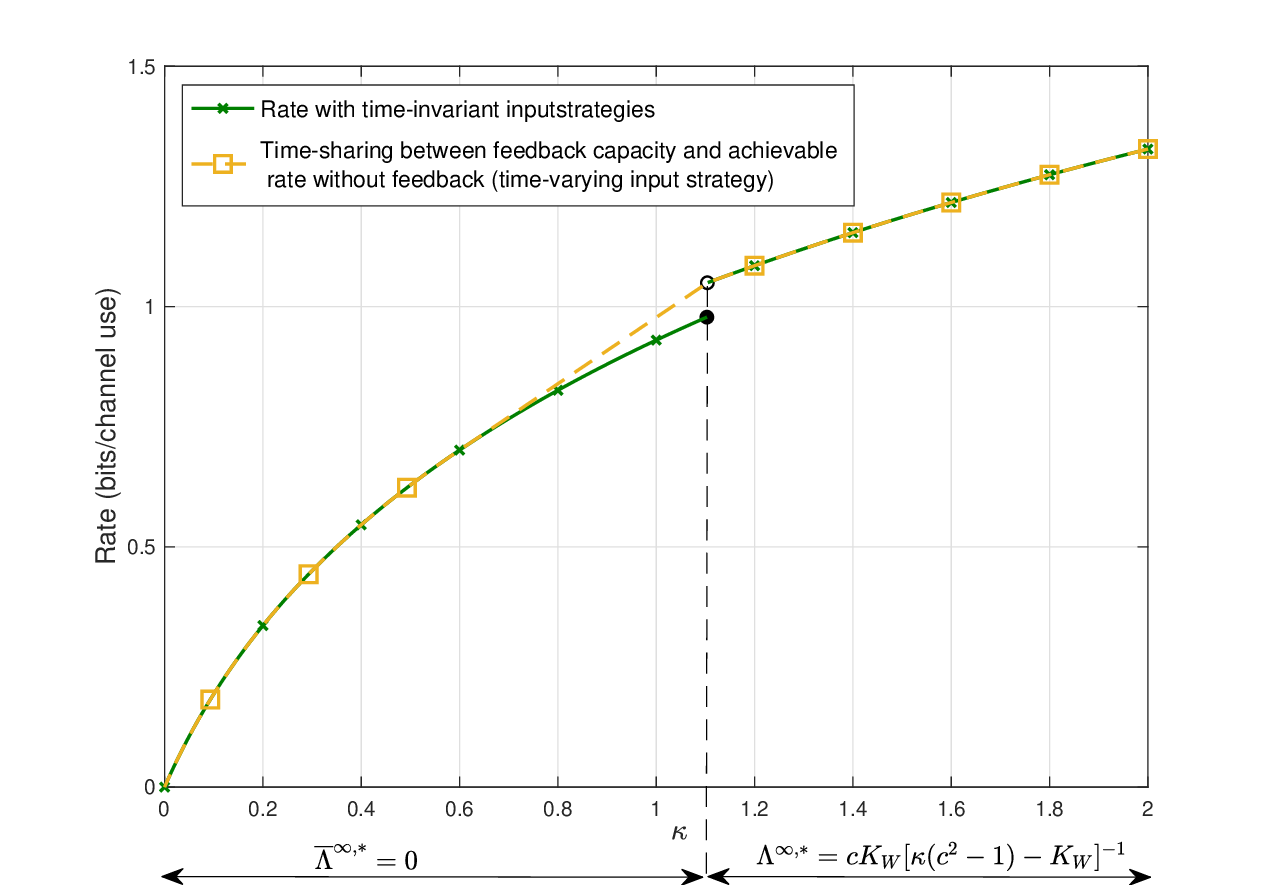}
  \caption{Time-sharing between time-invariant channel input feedback strategy and time-invariant channel input non-feedback strategy (of the lower bound)  for an AGN channel driven by AR$(c)$ noise, with $c=1.5$ and $K_W=1$.}
  \label{fig:timesharing}
\end{figure}

\begin{figure}
\centering
  \includegraphics[width=0.85\linewidth]{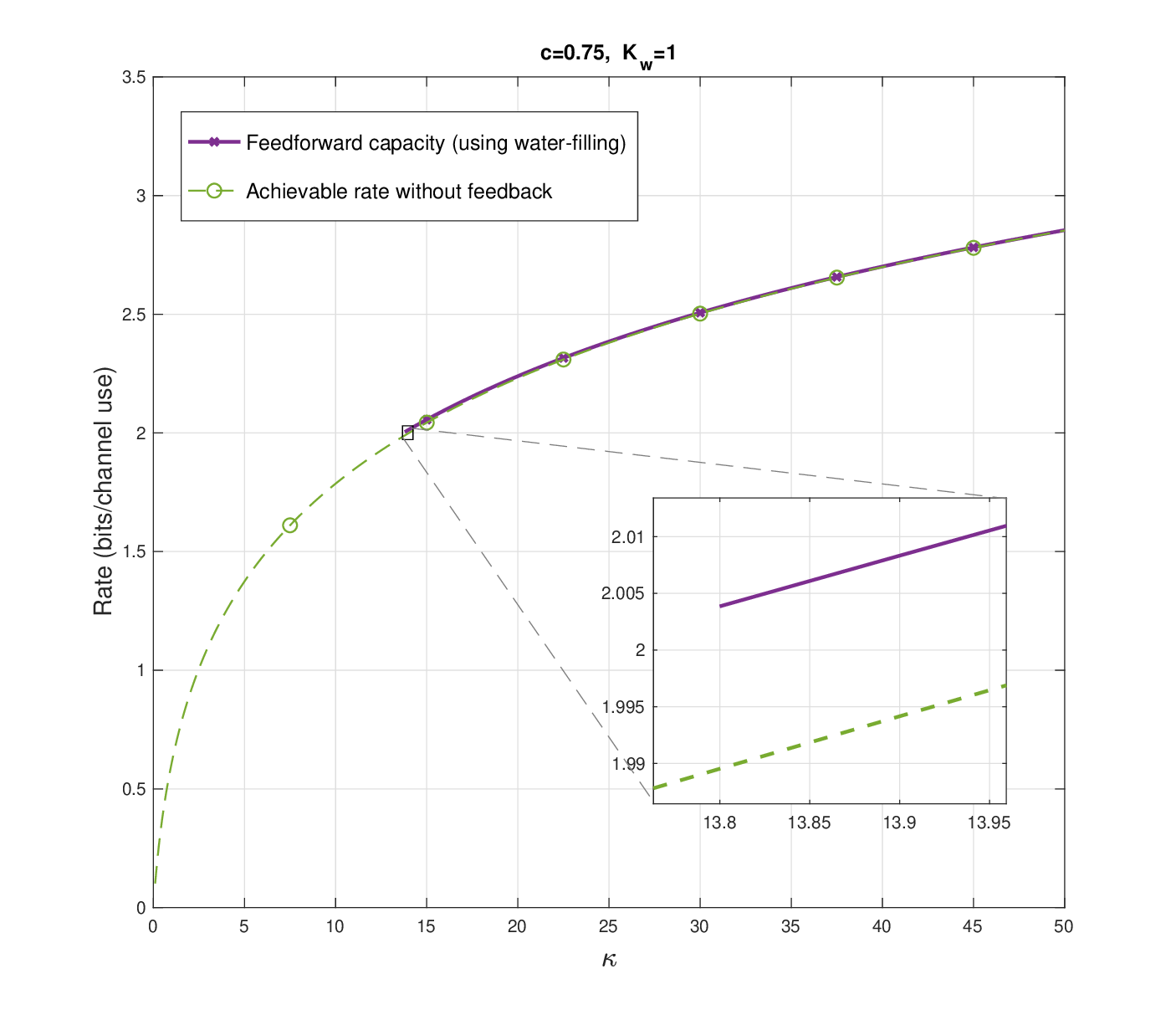}
  \caption{Comparison of non-feedback capacity $C^{nfb}(\kappa)$  based on water-filling formulae (\ref{nfb_wf}), of   \cite[eqn(5.5.14)]{ihara1993} and \cite[eqn(6)]{butman1976},   and lower bound $C_{LB}^{\infty,nfb}(\kappa)$ on  achievable rates without feedback based on formulae  \eqref{int_nf_8_aa}, of transmitting an IID channel input $Z_t^o\in N(0,\kappa)$, for an AGN channel driven by AR$(c)$, noise,  $c=0.75$ and $K_W=1$ (the values correspond to the maximum difference). } 
  \label{fig:compnofeed}
\end{figure}

The rate of a time-sharing scheme between feedback capacity of regime 1, $C^\infty(\kappa),  \kappa \in {\cal K}^\infty(c,K_W)$,  and the lower bound on capacity without feedback $C_{LB}^{\infty,nfb}(\kappa), \kappa \in [0,\infty)$,   is illustrated in Figure~\ref{fig:timesharing}. This scheme  results in higher rates compared to feedback capacity, since it employs a time-varying  channel input strategy, that is, two different strategies, one without feedback and one with feedback   are applied, whereas the feedback capacity of regime 1,   is only  defined over time-invariant channel input strategies.

Now, we digress to recall the closed form  non-feedback capacity formulae, that is derived using water-filling in \cite[eqn(5.5.14)]{ihara1993} (see also (and \cite[eqn(6)]{butman1976}),  for the AGN channel driven a stable AR$(c)$ noise. 
\begin{align}
C^{nfb}(\kappa) =\frac{1}{2}\log\Big(1 +\kappa +   \frac{c^2}{1-c^2}\Big), \hso \kappa>\big(\frac{1}{(1-|c|)^2}-\frac{1}{(1-c^2)}\big), \hso c\in (-1,1),  \hso K_W=1. \label{nfb_wf}
\end{align}
 Figure~\ref{fig:compnofeed} compares the non-feedback capacity $C^{nfb}(\kappa)$ based on (\ref{nfb_wf}) to the achievable lower bound on the non-feedback $C_{LB}^{\infty,nfb}(\kappa)$ based on  \eqref{int_nf_8_aa}, which corresponds to transmitting an IID channel input $Z_t^o\in N(0,\kappa)$, i.e., $K_Z^{\infty,*}=\kappa$,   for $c=0.75$ and $K_W=1$. Surprisingly, 
contrary to the non-feedback lower bound formulae  \eqref{int_nf_8_aa},  which holds for all stable or unstable AR$(c), c\in (-\infty, \infty)$ noise and $\kappa\in[0,\infty)$, the  closed form non-feedback capacity formulae (\ref{nfb_wf}), based on water-filling, is restricted to $\kappa>\big(\frac{1}{(1-|c|)^2}-\frac{1}{(1-c^2)}\big)$, and to the stable AR$(c), c\in (-1,1)$. The  maximum difference $C^{nfb}(\kappa)-C_{LB}^{\infty,nfb}(\kappa)$, when $K_W=1$ occurs at $c=0.75$, and is less than $1.5\times 10^{-2}$ bits per channel use. This difference is expected to be reduced further if a unit memory optimal input $X_t^o=\overline{\Lambda}^\infty X_{t-1}^o+Z_t^o, X_1^o=Z_1^o,  t=2, \ldots,$ is used, as stated in (R4).

\section{Characterizations of $n-$FTFI and $n-$FTwFI Capacity   of  AGN Channels}
\label{sect:problem}
In this section we present the following main results. \\ 
(1) Theorem~\ref{thm_FTFI} (Section~\ref{sect:ar1_r}), which gives the characterization of $n-$FTFI capacity for time-varying  feedback codes of Definition~\ref{def_code}.(a), \\
(2)  high level discussion (Section~\ref{prel_ric}) on the implications of generalized Kalman-filter equations on the characterizations  of $n-$FTFI capacity, and \\
(3) Corollary~\ref{thm_FTWFI} (Section~\ref{sect:wf_ar1_r}), which gives a lower bound   on the  $n-$FTwFI capacity for time-varying  non-feedback codes of Definition~\ref{def_code}.(b), based on a Markov channel input process without feedback, and follows directly from Theorem~\ref{thm_FTFI}.

\subsection{Characterization of $n-$FTFI   Capacity for AGN Channels Driven by  AR$(c_t)$ Noise}
\label{sect:ar1_r}
Below, we  introduce  the characterization of the  $n-$FTFI capacity, for an AGN channel, driven by the  time-varying   AR$(c_t)$ noise, 
for  the  feedback code of Definition~\ref{def_code}.(a).  
 Our presentation, of the next theorem, is based on the degenerate case of the general characterization of the $n-$FTFT capacity of AGN channels, derived in \cite{charalambous-kourtellaris-loykaCDC2018}. We should mention that although,  \cite{yang-kavcic-tatikonda2007},  treats AGN channels driven by stable noise, some parts of  the representation given below can be    extracted from the analysis of  \cite[Section~II-V]{yang-kavcic-tatikonda2007}.\\

\begin{theorem}  Characterization of $n-$FTFI Capacity for AGN Channels Driven by  AR$(c_t)$ Noise\\
\label{thm_FTFI}
Consider the AGN channel (\ref{AGN_in}) driven by a time-varying AR$(c_t)$ noise, i.e., (\ref{ykt_ar_1}), and the  code of Definition~\ref{def_code}.(a). 
Then the following hold.\\
(a) The optimal time-varying channel input distribution with feedback, for the optimization problem $C_n(\kappa,v_0)$ defined by  (\ref{FTFIC_1in}), is  conditionally Gaussian, of  the form
\begin{align}
{\bf P}_{X_t|X_{t-1},Y^{t-1},V_0}={\bf P}_{X_t|V_{t-1},Y^{t-1},V_0}, \hst t=1, \ldots, n \label{tv_is}
\end{align}
and it  is induced by the time-varying jointly Gaussian channel input process $X^n$, with a representation\footnote{\CDC{The fact that $X_1=Z_1, K_1=0, \widehat{V}_0=v_0$ is due to the code definition, i.e., $V_0=v_0$ is known to the encoder.}} 
\begin{align}
&X_t = \overline{\Lambda}_t \Big({X}_{t-1} - \widehat{X}_{t-1}\Big) + Z_t,      \hst t=2, \ldots,n,  \label{Q_1_3_s1}  \\
&\hso \hso =\Lambda_t \Big({V}_{t-1} - \widehat{V}_{t-1}\Big) + Z_t,  \hso \overline{\Lambda}_t=-\Lambda_t, \label{Q_1_3_s1_a} \\
&X_1 =  Z_1, \label{Q_1_3_a_s1}   \\
&Z_t\in  N(0, K_{Z_t}), \hso t=1, \ldots, n \hso \mbox{a  Gaussian sequence,}\label{Q_1_5_s1} \\
&Z_t \hso  \mbox{independent of}  \hso  (V^{t-1},X^{t-1},Y^{t-1}, {V}_0),  \hso t=1, \ldots, n,\label{Q_1_6_s1}\\
&Z^n \hso  \mbox{independent of}  \hso  (V^{n},{V}_0),\\
&V_t=c_t V_{t-1} +W_t, \hso V_0=v_0,\hst c_t \in (-\infty,\infty), \hso t=1, \ldots, n,\label{Q_1_3_a_s1_a} \\
&Y_t= X_t + V_t=  \overline{\Lambda}_t \Big({X}_{t-1} - \widehat{X}_{t-1}\Big)+c_t \Big(Y_{t-1}-X_{t-1}\Big) +W_t + Z_t,  \hso t=2, \ldots, n \label{Q_1_4_s1_a}\\
&\hso \hso =\Lambda_t \Big({V}_{t-1} - \widehat{V}_{t-1}\Big)+c_t V_{t-1} +W_t + Z_t,  \label{Q_1_4_s1}\\
&Y_1=Z_1+c_1 V_0 +W_1, \hso V_0=v_0, \label{Q_1_4_s1_a_b}\\
&\frac{1}{n}  {\bf E}_{v_0} \Big\{\sum_{t=1}^{n} \big(X_t\big)^2\Big\}=\frac{1}{n}   \sum_{t=1}^n \Big\{\Big(\Lambda_t\Big)^2  K_{t-1}  + K_{Z_t} \Big\}   \leq \kappa,     \label{cp_e_ar2_s1}\\
&(\Lambda_t, K_{Z_t})\in (-\infty, \infty) \times [0,\infty) \hst \mbox{scalar-valued,  non-random,}\\
&\widehat{X}_{t} \tri {\bf E}_{v_0}\Big\{X_{t} \Big| Y^{t}\Big\}, \hso \widehat{V}_{t} \tri {\bf E}_{v_0}\Big\{V_{t} \Big| Y^{t}\Big\},  \\
& K_{t}\tri {\bf E}_{v_0}\left\{\Big(X_{t} - \widehat{X}_{t}\Big)^2  \right\}={\bf E}_{v_0}\left\{\Big(V_{t} - \widehat{V}_{t}\Big)^2  \right\},  \hst  t=1, \ldots, n.\label{s1_24}
\end{align}  
Further, $H(Y^n|v_0)-H(V^n|v_0)$,     $(\widehat{V}_t, K_t), t=1, \ldots, n$ are determined  by the generalized\footnote{Unlike \cite{yang-kavcic-tatikonda2007}, we use    the term generalized, because,  the conditions for the asymptotic analysis to hold,  are fundamentally different from those of asymptotic analysis of  classical Kalman-filter equations.} time-varying Kalman-filter  and generalized time-varying difference Riccati  equation (DRE), of estimating $V^n$ from $Y^n$, given below.

{\it  Generalized   Kalman-filter Recursion for (\ref{Q_1_3_a_s1_a})-(\ref{Q_1_4_s1_a}) \cite{caines1988,kailath-sayed-hassibi}:} 
\begin{align}
&\widehat{V}_{t} = c_t \widehat{V}_{t-1} + M_t(K_{t-1}, \Lambda_t, K_{Z_t}) I_t, \hso \widehat{V}_{0}=v_0, \label{Q_1_8_s1} \\
&\hso = F_t(K_{t-1},\Lambda_t,K_{Z_t}) \widehat{V}_{t-1} + M_t(K_{t-1}, \Lambda_t, K_{Z_t}) Y_t, \hso \widehat{V}_{0}=v_0, \label{Q_1_8_s1_a}  \\
&I_t \tri Y_t-{\bf E}_{v_0}\Big\{Y_t\Big|Y^{t-1}\Big\}= Y_t - c_t \widehat{V}_{t-1}, \hso I_1= Z_1+W_1,\hso    t=2, \ldots, n, \label{Q_1_9_s1}\\
& \hso = \Big(\Lambda_t +c_t\Big)\Big(V_{t-1}- \widehat{V}_{t-1}\Big) + Z_t +W_t,\label{Q_1_9_s1_n} \\
&M_t(K_{t-1}, \Lambda_t, K_{Z_t})  \tri  \Big( K_{W_t} + c_t  K_{t-1}\Big(\Lambda_t + c_t \Big)\Big)\Big(K_{Z_t}+ K_{W_t} + \Big(\Lambda_t + c_t \Big)^2 K_{t-1}\Big)^{-1}, \\
& F_t(K_{t-1},\Lambda_t, K_{Z_t}) \tri c_t -M_t(K_{t-1}, \Lambda_t, K_{Z_t}) \Big(\Lambda_t +c_t\Big)\\
& I_t, \hso t=1, \ldots, n, \hst \mbox{an orthogonal innovations process.}\label{Q_1_11_s1}
 \end{align}
 {\it Generalized Time-Varying Difference Riccati  Equation:}
\begin{align}
K_{t}= & c_t^2 K_{t-1}  + K_{W_t} -\frac{ \Big( K_{W_t} + c_t K_{t-1}\Big(\Lambda_t + c_t \Big)\Big)^2}{ \Big(K_{Z_t}+ K_{W_t} + \Big(\Lambda_t + c_t \Big)^2 K_{t-1}\Big)}, \hst  K_t \geq0,  \hso K_{0}=0, \hst t=1, \ldots, n, 
 \label{Q_1_10_s1}
 \end{align}
{\it Error Recursion of the Generalized Kalman-filter, $E_t\tri V_t-\widehat{V}_t, t=1, \ldots, n$ :}
 \begin{align} 
E_t=&F_t(K_{t-1},\Lambda_t,K_{Z_t}) E_{t-1}-M_t(K_{t-1}, \Lambda_t, K_{Z_t})  \Big(Z_t+ W_t\Big)+W_t,  \hso E_0=v_0-\widehat{V}_0, \ \ t=1, \ldots, n. \label{i_error_s1}
\end{align} 
Entropy of Channel Output Process:
\begin{align}
H(Y^n|v_0)= \sum_{t=1}^n H(Y_t|Y^{t-1},v_0)= \sum_{t=1}^n H(Y_t-{\bf E}\big\{Y_t|Y^{t-1}, v_0\big\}|Y^{t-1},v_0)
=\sum_{t=1}^n H(I_t). \label{Q_1_4_s1_aaa}
\end{align}
(b) The characterization of the $n-$FTFI capacity $C_n(\kappa,v_0)$ defined by  (\ref{FTFIC_1in}) is
\begin{align}
&C_{n}(\kappa,v_0) 
\tri  \sup_{\big(\Lambda_t, K_{Z_t} \big), t=1,\ldots, n: \hso \frac{1}{n} \sum_{t=1}^n \Big\{\big(\Lambda_t\big)^2 K_{t-1}+ K_{Z_t}\Big\}   \leq \kappa}  \frac{1}{2}  \sum_{t=1}^n   \log\Big( \frac{\big(\Lambda_t+c_t\big)^2 K_{t-1}  + K_{Z_t} +K_{W_t}}{K_{W_t}}\Big) \label{tv_stra}   \\
&\mbox{subject to: $K_t, t=1, \ldots,n$ satisfies recursion (\ref{Q_1_10_s1}) and $K_{Z_t}\geq 0 , t=1, \ldots, n$.}
\end{align} 
\end{theorem}

\begin{proof} (a) Representation (\ref{tv_is}) follows, from a degenerate case of  \cite{charalambous-kourtellaris-loykaCDC2018}. The representation of the jointly Gaussian process $X^n$, defined by (\ref{Q_1_3_s1_a}), such that $Z^n$ satisfies  (\ref{Q_1_5_s1}) and (\ref{Q_1_6_s1}), is also a degenerate case of  \cite{charalambous-kourtellaris-loykaCDC2018}, where the channel is more general, of the form $Y_t= C_{t-1} Y_{t-1}+ D_t X_t +D_{t,t-1} X_{t-1} +V_t$, $V_t=F_tV_{t-1} +W_t$, where $(C_{t,t-1}, D_t, D_{t,t-1}, F_t)$ are nonrandom,  i.e., with  past dependence on channel inputs and outputs.   
The representation (\ref{Q_1_3_s1}) follows from (\ref{Q_1_3_s1_a}), by substituting $V_{t-1}=Y_{t-1}-X_{t-1}$.  Expressions (\ref{Q_1_3_a_s1_a})-(\ref{s1_24}) follow directly from (\ref{Q_1_3_s1}) and  (\ref{Q_1_3_s1_a}), and the channel definition.  The generalized Kalman-filter equations follow from standard texbooks, i.e., \cite{caines1988}. (\ref{Q_1_4_s1_aaa}) follows from the independent property of the innovations process.  (b) Follows from (\ref{FTFIC_1in}), (\ref{Q_1_4_s1_aaa}),  $H(V^n|v_0)=\sum_{t=1}H(W_t)$, and part (a).  
\end{proof}

\ \

\begin{remark} By the definition of the innovations process and entropy, (\ref{Q_1_9_s1_n}) and (\ref{Q_1_4_s1_aaa}), it follows that whether the limit exists, $\lim_{n \longrightarrow \infty}\frac{1}{n} \big\{H(Y^n|v_0)-H(V^n|v_0)\big\}= \lim_{n \longrightarrow \infty}\frac{1}{n}\sum_{t=1}^n \big\{H(I_t)-H(W_t)\big\}\in [0,\infty)$ can be  determined from the limiting  covariance of the innovations process $I^n$ and noise $W^n$. Similarly, for $\lim_{n \longrightarrow \infty} \frac{1}{n}  {\bf E}_{v_0} \Big\{\sum_{t=1}^{n} \big(X_t\big)^2\Big\}\in [0,\infty)$ by     (\ref{cp_e_ar2_s1}). 
\end{remark}

\subsection{Some Preliminary Facts on Generalized Difference and Algebraic Riccati Equations}
\label{prel_ric}
At this point we pause to discuss several facts,  to provide insight into  the questions of asymptotic analysis,  of  $C(\kappa,s_0)$,  defined by (\ref{FTFI_2}), and $C^\infty(\kappa,v_0)$ defined by (\ref{inter}),  based on the properties of generalized DREs and AREs. These   are   well-developed in the systems and control theory literature, for control and estimation problems  \cite{kailath-sayed-hassibi,caines1988}. The reader  finds a more detailed  presentation in Section~\ref{sect:q1},  Section~\ref{sect:pre}, and  Theorem~\ref{thm_ric}. 


{\bf Fact 1.} {\it The generalized time-varying Kalman-filter with correlated state and observation noise.} \\
The definition of the innovations process $I^n$ of $Y^n$, i.e.,   (\ref{Q_1_9_s1}),  (\ref{Q_1_9_s1_n}), implies    the the generalized Kalman-filter recursions (\ref{Q_1_8_s1})-(\ref{Q_1_10_s1}), correspond to the problem of estimating $X^n$ from $\overline{Y}^n$, of the  system: 
\begin{align}
&V_t=A_t V_{t-1} +W_t, \hso V_0=v_0, \hso t=1, \ldots, n,\label{state_1}  \\
&\overline{Y}_t= C_t {V}_{t-1}  +W_t + Z_t,  \hso \overline{Y}_t\tri Y_t+ \Lambda_t  \widehat{V}_{t-1},  \label{obse_1}\\
&A_t\tri c_t, \hso   C_t \tri  \Lambda_t +c_t.
\end{align}
The terminology,   ``generalized DREs'' is often used in the literature, because in system (\ref{state_1}),  (\ref{obse_1}),  the noises entering $V^n$ and $\overline{Y}^n$ are correlated.  The  noise covariance,  for system (\ref{Q_1_3_a_s1_a}), (\ref{Q_1_4_s1}), or system 
(\ref{state_1}),  (\ref{obse_1}) is thus,  defined by 
\begin{align}
&cov( \left[ \begin{array}{c} W_t  \\ W_t +Z_t\end{array} \right], \left[ \begin{array}{c} W_t  \\ W_t +Z_t\end{array} \right]^T) \tri    \left[ \begin{array}{cc} Q_t & S_t  \\ S_t & R_t\end{array} \right]= \left[ \begin{array}{cc} K_{W_t} & K_{W_t}  \\ K_{W_t} &K_{W_t} +K_{Z_t}\end{array} \right]\ \succeq 0,\hso   t=1,\ldots, n, \label{cp_e_ar2_s1_a}\\
&Q_t\tri K_{W_t}, \hso S_t\tri K_{W_t}, \hso R_t \tri K_{W_t}+K_{Z_t} >0.
\end{align}
With  the above notation,  then the  
generalized time-varying DRE (\ref{Q_1_10_s1}),  is given by 
\begin{align}
K_{t}= & A_t^2 K_{t-1}  + K_{W_t} -\frac{ \Big( K_{W_t} + A_t K_{t-1} C_t \Big)^2}{ \Big(R_t + \Big(C_t \Big)^2 K_{t-1}\Big)}, \hst  K_t \geq0,  \hso K_{0}=0, \hst t=1, \ldots, n, 
 \label{Q_1_10_s1_new1}
 \end{align}
By \cite{kailath-sayed-hassibi},  it  follows that the   
generalized time-varying DRE (\ref{Q_1_10_s1}), equivalently (\ref{Q_1_10_s1_new1}),  is the time-varying version of  the one studied in \cite[Section~17.7, i.e., eqn(14.7.1), and throughout the book, for $S\neq 0$]{kailath-sayed-hassibi}. That is, the Kalman-filter (\ref{Q_1_8_s1})-(\ref{Q_1_10_s1}) or equivalently, the Kalman-filter   for  (\ref{state_1}),  (\ref{obse_1}),  is not the classical Kalman-filter. The generalized Kalman-filter (\ref{Q_1_8_s1})-(\ref{Q_1_10_s1}) reduces to the classical Kalman-filter if and only if the noise $W^n$ does not enter $Y^n$ or $\overline{Y}^n$, that is,  $S=0$ (see  the notation in \cite[$S=0$]{kailath-sayed-hassibi}). 


{\bf Fact 2.} {\it On the convergence of solutions of time-invariant generalized DREs to solutions of  AREs.}\\
For the time-invariant  AR$(c)$ noise, i.e., $c_t=c, K_{W_t}=K_W, t=1, \ldots, n$, consider  the time-invariant channel input distributions or  strategies,  i.e., $ \Lambda_t=\Lambda^\infty, K_{Z_t}=K_{Z}^\infty, t=1, \ldots, n$, which do not imply the corresponding generated  process $(X^o,Y^o)$ is stationary.  Let  $K_t^o, t=0,1, \ldots, n$, denote the sequence  generated by the  strategies $(\Lambda^\infty,K_{Z}^\infty)$. Then from Theorem~\ref{thm_FTFI}, we obtain $C^\infty(\kappa,v_0)$ defined by (\ref{int_fb}), (\ref{DRE_in}),   that is,   the time-invariant version of the generalized DRE (\ref{Q_1_10_s1}),  equivalently (\ref{Q_1_10_s1_new1}):
\begin{align}
&K_{t}^o=  A^2 K_{t-1}^o  + K_{W} -\frac{ \Big( K_{W} + A K_{t-1}^o C^\infty \Big)^2}{ \Big(R^\infty + \Big(C^\infty \Big)^2 K_{t-1}^o\Big)}, \hst  K_t^o \geq0,  \hso K_{0}^o=0, \hst t=1, \ldots, n, 
 \label{Q_1_10_s1_new1_a}\\
 & C^\infty\tri C+ \Lambda^\infty, \hso R^\infty\tri K_W + K_Z^\infty. 
 \end{align} 
 Hence, to determine whether the limit in $C^\infty(\kappa,v_0)$ defined by (\ref{int_fb}),  exists, it is necessary to understand the convergence  properties of  $K_{t}^o, t=0,1,\ldots, n$, as $n\longrightarrow \infty$, and these properties depend on   the values of parameters $(c,K_W,\kappa)$. 
It is well-known, and easily  verified from \cite{kailath-sayed-hassibi,caines1988},  that  the  convergence properties of the mean-square error $K_{t}^o, t=0,1,\ldots, n$, as $n\longrightarrow \infty$, of  the generalized Kalman-filter, are fundamentally different from those  of the classical Kalman-filter.   In particular, even for the special case of  a stable   AR$(c)$ noise, i.e., $c\in (-1,1)$,   by   Section~\ref{sect:pre} (i.e.,  Theorem~\ref{thm_ric}.(1)), the conditions, known as detectability and stabilizability, are sufficient and/or necessary  conditions, for the  convergence of the  mean-square error $K_{t}^o, t=0,1,\ldots,n$, as $n\longrightarrow \infty$,     to a finite nonnegative, unique limit  $K^\infty\geq 0$, such that  $K^\infty$ satisfies the generalized algebraic Riccati equation (ARE), i.e., the steady state version of (\ref{Q_1_10_s1_new1_a}): 
\begin{align}
K^\infty =A^2 K^\infty  + K_{W} -\frac{ \Big( K_{W} + A K^\infty C^\infty \Big)^2}{ \Big(R^\infty + \Big(C^\infty \Big)^2 K^\infty\Big)}, \hst  K^\infty \geq 0. 
 \label{Q_1_10_s1_new1_ain}
 \end{align}
On the other hand, the following is a well-known property of classical Kalman-filters, which is   easily verified from the properties of DREs and AREs presented in    \cite[with $S=0$]{kailath-sayed-hassibi} and \cite{caines1988}: for classical time-invariant DREs and classical AREs, i.e., that  correspond to a stable state process,  to be estimated, driven by a Gaussian noise, and such that,  the noise is independent of  the Gaussian noise entering the observations,   then the detectability and stabilizability conditions are automatically satisfied. This implies   the mean-square estimation error of  Kalman-filter  converges,  to a finite nonnegative, unique  limit, which satisfies a classical ARE. The unique solution of the classical ARE is also  stabilizing.

{\bf Fact 3.} {\it On the zero variance of the innovations process of the channel input.}\\
Suppose  $K_Z^\infty$ in  (\ref{Q_1_10_s1_new1_a}) and (\ref{Q_1_10_s1_new1_ain}),  is replaced by $K_Z^\infty=0$, i.e., $R^\infty=K_W$. Then the resulting generalized ARE  (\ref{Q_1_10_s1_new1_ain}), with $K_Z^\infty$ is precisely the Riccati equation in the feedback characterization in \cite[Theorem~6.1, $\Sigma=K^\infty$]{kim2010} for the AR$(c)$ noise,  and this equation is  a quadratic polynomial in $K^\infty$, with two solutions:
\begin{align}
K^\infty=0,  \hst K^\infty\equiv K^\infty(\Lambda^\infty) =\frac{K_W\Big((\Lambda^\infty)^2-1\Big)}{\Big(\Lambda^\infty+ c\Big)^2}, \hst \Lambda^\infty \neq -c. \label{gre_1_a}
\end{align} 
The second   solution $K^\infty(\Lambda^\infty)$     is a  functional of  $\Lambda^\infty$, and gives rize to solutions:
\begin{align}
&K^\infty(\Lambda^\infty) =\frac{K_W\Big((\Lambda^\infty)^2-1\Big)}{\Big(\Lambda^\infty+ c\Big)^2}\geq 0, \hst \mbox{if and only if $|\Lambda^\infty|\geq 1$,} \label{gre_1}\\
&K^\infty(\Lambda^\infty) =\frac{K_W\Big((\Lambda^\infty)^2-1\Big)}{\Big(\Lambda^\infty+ c\Big)^2}\leq  0, \hst \mbox{if and only if $|\Lambda^\infty|\leq  1$.} \label{gre_2}
\end{align}
{\it The main question is then:} which one of the solutions $K^\infty$ is the unique limit of the sequence $K_{t}^o, t=0,1,\ldots,n$,  as $n\longrightarrow \infty$, which then defines the  unique limit of the entropy rate $\frac{1}{n}H(Y^{n,o}|v_0)=\frac{1}{n}\sum_{t=1}^nH(I_t^o)$ in (\ref{Q_1_4_s1_aaa}) and average power $\frac{1}{n}  {\bf E}_{v_0} \Big\{\sum_{t=1}^{n} \big(X_t\big)^2\Big\}$ in  (\ref{cp_e_ar2_s1}),  as $n\longrightarrow \infty$? \\
{\it The answer to this question is:} the solution that corresponds to $|\Lambda^\infty|<1$; this  follows from   the properties of generalized DREs and AREs, stated in Theorem~\ref{thm_ric}.(1),  (see also Lemma~\ref{lem_pr_are}.(3)),  because     $|\Lambda^\infty|<1$ is a necessary and sufficient condition for convergence of  $K_{t}^o, t=0,1,\ldots,n$ for all $K_0^o \geq 0$,  as $n\longrightarrow \infty$,   to a unique limit $K^\infty\geq 0$ which satisfies the generalized ARE (\ref{Q_1_10_s1_new1_ain}). Consequently, the unique nonnegative limit of the $K_{t}^o, t=0,1,\ldots,n \; \forall K_0^o\geq 0$,  as $n\longrightarrow \infty$, is   $K^\infty=0$, since both solutions (\ref{gre_1}) and (\ref{gre_2}) are ruled out, by the condition $|\Lambda^\infty|<1$. By the expression of the entropies $H(Y^{n,o}|v_0)-H(V^n|v_0)=\sum_{t=1}^n \big\{H(I_t^o)-H(W_t)\big\}$, inside the limit in  (\ref{int_fb}), if $K_Z^\infty=0$,   then   $C^\infty(\kappa,v_0)=0, \forall \kappa \in [0,\infty), \forall v_0$. This means the feedback capacity characterization in \cite[Theorem~6.1, $C_{FB}$]{kim2010} is zero. An alternative illustration is given in Counterexample~\ref{counterex}.

\CDC{
\begin{remark} 
{\bf Facts 1-3} state  that certain  fundamental technical issues need to be accounted for in the asymptotic analysis of entropy rates and of the average power in \cite{liu-han2017,liu-han2019,gattami2019,li-elia2019},  even if the noise $V^n$ is stable, stationary, etc.  to ensure optimal time-invariant channel input strategies  induce asymptotic stationarity of the channel  input  process, and  asymptotic stationarity of the output process (for stable noise). 
\end{remark}
}

\subsection{Converse Coding Theorem  for  AGN Channels}
By Theorem~\ref{thm_FTFI}, the characterization of $n-$FTFI capacity, $C_{n}(\kappa,v_0)$,  is expressed in terms of the mean-square error $K_t, t=1, \ldots,n$, that satisfies the time-varying generalized RDE  (\ref{Q_1_10_s1}).  We recall  the error recursion of the generalized Kalman-filter given by  (\ref{i_error_s1}), which satisfies a  linear time-varying recursion, and hence its convergence properties, in mean-square sense, i.e., $\lim_{n\longrightarrow \infty} K_n = \lim_{n\longrightarrow \infty}{\bf E}_{v_0}\big\{\big(E_n\big)^2 \big\}$  is determined by the properties of $F_t(K_{t-1},\Lambda_t,K_{Z_t})$ and  $M_t(K_{t-1}, \Lambda_t, K_{Z_t}),\Lambda_t, K_{Z_t}, t=1,2,  \ldots$. In general, $\lim_{n \longrightarrow \infty } K_n =\lim_{n \longrightarrow \infty } {\bf E}_{v_0}\big\{\big(E_n\big)^2 \big\}$ does not exist, for arbitrary  $F_t(K_{t-1},\Lambda_t,K_{Z_t})$, $M_t(K_{t-1}, \Lambda_t, K_{Z_t}), \Lambda_t, K_{Z_t}, t=1,2,  \ldots$. In view of the error recursion (\ref{i_error_s1}), we have the following theorem. \\

\begin{theorem} Converse coding theorem\\
Consider the feedback code ${\cal C}_{{\mathbb Z}^+}^{fb}$ of Definition~\ref{def_code}.(a). \\
{\it Converse Coding Theorem.}  If there exists a feedback code ${\cal C}_{{\mathbb Z}^+}^{fb}$, i.e., with $\epsilon_n \longrightarrow 0$, as $n \longrightarrow \infty$, then the code rate $R(v_0)$ satisfies:
\bea
R(v_0) \leq C(\kappa,v_0)\tri  \lim_{ n \longrightarrow \infty} \frac{1}{n} C_{n}(\kappa,v_0),    \hst C_{n}(\kappa,v_0) \hso \mbox{defined in Theorem~\ref{thm_FTFI}.(b)}   \label{fc}
\eea
provided the following conditions hold:\\
(C1) the maximizing element, denoted by  $(\Lambda_t^*, K_{Z_t}^*), t=1, \ldots, n$ which satisfies the average power constraint  exists, and \\
(C2)  the limit exists in $[0,\infty)$.
\end{theorem}
\begin{proof}  (C1) and (C2) follow from the above discussion; the converse coding theorem is similar to \cite{cover-pombra1989}.
\end{proof}

\begin{remark} By the average power  (\ref{cp_e_ar2_s1}) and  optimization problem  (\ref{tv_stra}),   it is necessary to identify sufficient and/or necessary conditions such that  the maximizing element, $(\Lambda_t^*, K_{Z_t}^*), t=1, \ldots, n$, exists in the set, and to  ensure  convergence of  $K_n = {\bf E}_{v_0}\big\{\big(E_n\big)^2 \big\}$ (that satisfies the time-varying DRE (\ref{Q_1_10_s1})),  as $n\longrightarrow \infty$,   such that the limit in (\ref{fc}) exists in $[0,\infty)$. However,  to ensure   $C(\kappa,v_0)$ is independent of $v_0$, it is necessary that the   limit is also  independent of $v_0$. On the other hand, if the limit  $C(\kappa,v_0)$ depends on $V_0=v_0$, then one needs to consider a formulation based on  compound capacity, by taking infimum over all initial states $V_0=v_0$, as done,  for example, in \cite{permuter-weissman-goldsmith2009}, for finite state feedback channels,  otherwise different $v_0$ give rise to different rates.  
\end{remark}

\subsection{Lower Bound on Characterization of $n-$FTwFI  Capacity for AGN Channels Driven by  AR$(c_t)$ Noise}
\label{sect:wf_ar1_r}
Next, we give a lower bound on the characterization of $n-$FTwFI  Capacity, for the  non-feedback code  of Definition~\ref{def_code}.(b),  which follows directly from  Theorem~\ref{thm_FTFI}. \\

\begin{corollary}  Lower bound on characterization of $n-$FTwFI Capacity for AGN Channels Driven by  AR$(c_t)$ Noise\\
\label{thm_FTWFI}
Consider the AGN channel (\ref{AGN_in}) driven by a time-varying AR$(c_t)$ noise, i.e., (\ref{ykt_ar_1}), and the code without feedback, of Definition~\ref{def_code}.(b). Define the information theoretic optimization problem of  capacity without feedback, i.e., the analog of (\ref{FTFIC_1in}),  by
\begin{align}
C_n^{nfb}(\kappa,v_0) \tri   \sup_{ {\bf P}_{X_t|X^{t-1}, V_0}, t=1, \ldots, n: \hso \frac{1}{n}  {\bf E}\big\{\sum_{t=1}^n \big(X_t\big)^2 \big\}\leq \kappa}  H(Y^n|v_0)-H(V^n|v_0) \label{FTFIC_1}
\end{align}
provided the supremum exists. Then the following hold.\\
(a) A lower bound on $C_n^{nfb}(\kappa,v_0)$ is obtained by the  conditionally Gaussian, time-varying  channel input distribution without feedback,  given by  
\begin{align}
{\bf P}_{X_t|X^{t-1},V_0}={\bf P}_{X_t|X_{t-1},V_0}, \hst t=1, \ldots, n \label{is_nfb}
\end{align}
which is induced by the time-varying jointly Gaussian channel input process $X^n$, with a representation 
\begin{align}
&X_t = \overline{\Lambda}_t {X}_{t-1} + Z_t,      \hst t=2, \ldots,n,  \label{Q_1_3_s1_w}  \\
&X_1 = Z_1, \label{Q_1_3_a_s1_w}   \\
&Z_t\in  N(0, K_{Z_t}), \hso t=1, \ldots, n \hso \mbox{a  Gaussian sequence,}\label{Q_1_5_s1_w} \\
&Z_t \hso  \mbox{independent of}  \hso  (V^{t-1},X^{t-1},Y^{t-1}, {V}_0),  \hso t=1, \ldots, n,\label{Q_1_6_s1_w}\\
&Z^n \hso  \mbox{independent of}  \hso  (V^{n},{V}_0),\\
&V_t=c_t V_{t-1} +W_t, \hso V_0=v_0,\hst c_t \in (-\infty,\infty), \hso t=1, \ldots, n, \label{Q_1_3_a_s1_a_w}  \\
&Y_t= X_t + V_t=  \Big(\overline{\Lambda}_t-c_t\Big) {X}_{t-1}+c_t Y_{t-1} +W_t + Z_t, \hso t=2, \ldots, n, \label{Q_1_4_s1_a_w}\\
&Y_1=Z_1+c V_0 +W_1,\hso V_0=v_0, \label{Q_1_4_s1_a_w_a}\\
&\frac{1}{n}  {\bf E}_{v_0} \Big\{\sum_{t=1}^{n} \big(X_t\big)^2\Big\}=\frac{1}{n}   \sum_{t=1}^n \Big\{\Big(\overline{\Lambda}_t\Big)^2  K_{X_{t-1}}  + K_{Z_t} \Big\}   \leq \kappa,     \label{cp_e_ar2_s1_w}\\
&(\overline{\Lambda}_t, K_{Z_t})\in (-\infty, \infty) \times [0,\infty) \hst \mbox{scalar-valued,  non-random,}\\
& K_{X_t}\tri {\bf E}_{v_0}\Big(X_{t}\Big)^2,\label{s1_24_aw}\\
&\widehat{X}_{t} \tri {\bf E}_{v_0}\Big\{X_{t} \Big| Y^{t}\Big\}, \hso \widehat{V}_{t} \tri {\bf E}_{v_0}\Big\{V_{t} \Big| Y^{t}\Big\}, \hso E_t^{nfb} \tri V_{t} - \widehat{V}_{t} \\
& K_{t}\tri {\bf E}_{v_0}\left\{\Big(X_{t} - \widehat{X}_{t}\Big)^2  \right\}={\bf E}_{v_0}\left\{\Big(V_{t} - \widehat{V}_{t}\Big)^2  \right\},  \hst  t=1, \ldots, n.\label{s1_24_w}
\end{align}  
Further,  $(\widehat{X}_t, K_t), t=1, \ldots, n$ are determined  by the generalized time-varying Kalman-filter  and generalized time-varying difference Riccati  equation (DRE), of estimating $X^n$ from $Y^n$, and $K_{X_t}, t=1, \ldots, n$ is determined by the time-varying Lyapunov difference equation, given below.

{\it  Generalized   Kalman-filter Recursion for
 (\ref{Q_1_3_s1_w})-(\ref{Q_1_4_s1_a_w_a}) :} 
\begin{align}
&\widehat{X}_{t} = \overline{\Lambda}_t \widehat{X}_{t-1} + M_t^{nfb}(K_{t-1}, \overline{\Lambda}_t, K_{Z_t}) I_t, \hso \widehat{X}_{1}=\widehat{x}_1, \hso t=2, \ldots, n \label{Q_1_8_s1_w} \\
&\hso = F_t^{nfb}(K_{t-1},\overline{\Lambda}_t,K_{Z_t}) \widehat{X}_{t-1} + M_t^{nfb}(K_{t-1}, \overline{\Lambda}_t, K_{Z_t})\Big( Y_t-c_tY_{t-1}\Big), \label{Q_1_8_s1_a_w}  \\
&I_t \tri Y_t - \Big(\overline{\Lambda}_t-c_t\Big) \widehat{X}_{t-1}-c_t Y_{t-1}, \hso I_1=Z_1+W_1,  \hso t=2, \ldots, n, \label{Q_1_9_s1_w}\\
& \hso = \Big(\overline{\Lambda}_t -c_t\Big)\Big(X_{t-1}- \widehat{X}_{t-1}\Big) + Z_t +W_t,\label{Q_1_9_s1_n_w} \\
&M_t^{nfb}(K_{t-1}, \overline{\Lambda}_t, K_{Z_t})  \tri  \Big( K_{Z_t} + \overline{\Lambda}_t  K_{t-1}\Big(\overline{\Lambda}_t - c_t \Big)\Big)\Big(K_{Z_t}+ K_{W_t} + \Big(\overline{\Lambda}_t - c \Big)^2 K_{t-1}\Big)^{-1}, \\
& F_t^{nfb}(K_{t-1},\overline{\Lambda}_t, K_{Z_t}) \tri \overline{\Lambda}_t -M_t^{nfb}(K_{t-1}, \overline{\Lambda}_t, K_{Z_t}) \Big(\overline{\Lambda}_t -c_t\Big)\\
& I_t, \hso t=1, \ldots, n, \hst \mbox{an orthogonal innovations process.}\label{Q_1_11_s1_w}
 \end{align}
 {\it Generalized Time-Varying Difference Riccati  Equation:}
\begin{align}
K_{t}= & \overline{\Lambda}_t^2 K_{t-1}  + K_{Z_t} -\frac{ \Big( K_{Z_t} + \overline{\Lambda}_t K_{t-1}\Big(\overline{\Lambda}_t - c_t \Big)\Big)^2}{ \Big(K_{Z_t}+ K_{W_t} + \Big(\overline{\Lambda}_t - c_t \Big)^2 K_{t-1}\Big)}, \hst  K_t \geq0,  \hso K_{0}=0, \hst t=1, \ldots, n, 
 \label{Q_1_10_s1_w}
 \end{align}
  {\it Time-Varying Difference Lyapunov   Equation:}
\begin{align}
K_{X_t}= \overline{\Lambda}_t^2 K_{X_{t-1}}  + K_{Z_t}, \hst  K_{X_t} \geq0,  \hso K_{X_0}=0, \hst t=1, \ldots, n, 
 \label{Q_1_10_s1_lw}
 \end{align}
 {\it Error Recursion of the Generalized Kalman-filter, $E_t^{nfb}\tri X_t-\widehat{X}_t, t=1, \ldots, n$ :}
 \begin{align} 
E_t^{nfb}=&F_t^{nfb}(K_{t-1},\overline{\Lambda}_t,K_{Z_t}) E_{t-1}-M_t^{nfb}(K_{t-1}, \overline{\Lambda}_t, K_{Z_t})  \Big(Z_t+ W_t\Big)+Z_t,  \; E_0^{nfb}=\mbox{given}, \; t=1, \ldots, n. \label{i_error_s1_nfb}
\end{align} 
(b) The lower bound characterization of the $n-$FTwFI capacity $C_n^{nfb}(\kappa,v_0)$, defined by   (\ref{FTFIC_1}),  is
\begin{align}
&C_{n}^{nfb}(\kappa,v_0) \geq C_{n,LB}^{nfb}(\kappa,v_0)
\tri  \sup_{\big(\overline{\Lambda}_t, K_{Z_t} \big), t=1,\ldots, n: \hso \frac{1}{n} \sum_{t=1}^n \big\{\big(\overline{\Lambda}_t\big)^2 K_{X_{t-1}}+ K_{Z_t}\big\}   \leq \kappa}  \frac{1}{2}  \sum_{t=1}^n   \log\Big( \frac{\big(\overline{\Lambda}_t-c_t\big)^2 K_{t-1}  + K_{Z_t} +K_{W_t}}{K_{W_t}}\Big)\\
&\mbox{subject to: $K_t, K_{X_t},  t=1, \ldots,n$ satisfy recursions (\ref{Q_1_10_s1_w}), (\ref{Q_1_10_s1_lw}),  and $K_{Z_t}\geq 0 , t=1, \ldots, n$.}
\end{align} 
\end{corollary}
\begin{proof} (a) Similar to the feedback capacity of Theorem~\ref{thm_FTWFI}, by the maximum entropy of Gaussian distributions, the maximizing distributions ${\bf P}_{X_t|X^{t-1}, V_0}, t=1, \ldots, n$ for the optimization problem (\ref{FTFIC_1}) are conditionally Gaussian, such that $(X^n,Y^n)$ for $V_0=v_0$,  is jointly Gaussian,  the average power constraint is satisfied, and condition (\ref{g_cp_3}) is respected. Clearly, the restriction to distributions  that satisfy (\ref{is_nfb}) result in a lower bound on  $C_n^{nfb}(\kappa,v_0)$ defined by  (\ref{FTFIC_1}). Note that the restriction to (\ref{is_nfb}) is precisely the restriction of feedback distributions  (\ref{tv_is}) to non-feedback distributions. 
The rest of the equations follow, similarly to  Theorem~\ref{thm_FTFI}.(a),  and in particular,  (\ref{Q_1_3_s1}), if the channel is used without feedback, i.e.,  $X_t=\overline{\Lambda}_t X_{t-1} +Z_t$.  
The rest of the expression of part (a) are obtained as in Theorem~\ref{thm_FTFI}.(a), and the generalized Kalman-filter recursions follow from \cite{caines1988,kailath-sayed-hassibi}. (b) Due to  the expressions of part (a). 
\end{proof}

\ \

\begin{remark}
Corollary~\ref{thm_FTWFI} is useful, because  the  lower bound is much easier to compute, compared to $C_n^{nfb}(\kappa,v_0)$, defined by  (\ref{FTFIC_1}), where  the supremum is taken over all jointly Gaussian   channel input processes $X^n,n=1,2,\ldots,$ without feedback  or distributions without feedback, ${\bf P}_{X_t|X^{t-1}, V_0}, t=1,2, \ldots$.
\end{remark}

\section{New Formulas of Capacity of AGN Channels Driven by Stable and Unstable AR$(c)$ Noise and  Generalized Riccati Equations}
\label{sect:q1}
In this section we derive a closed form  formula for feedback capacity $C^{\infty}(\kappa,v_0)$,  defined by  (\ref{inter}),   and lower bounds on capacity without feedback $C^{\infty,nfb}(\kappa,v_0)$, defined by  (\ref{inter_nfb}), of AGN channels driven by AR$(c)$, stable and unstable noise, when  channel input strategies or distributions are  time-invariant.
This section includes material on basic  properties of generalized DREs,   AREs, and  definitions and implications of the notions of detectability and stabilizability, which  are  discussed in Section~\ref{prel_ric}.

\subsection{Characterization of Feedback Capacity for Time-Invariant Channel Input Distributions}
\label{sect:char_fb} 
We  restrict the class of  channel input distributions of Theorem~\ref{thm_FTFI}  to the class of time-invariant distributions. We note that our  restriction is weaker  than the analysis in \cite{kim2010}, which  presupposes  stationarity or asymptotic joint stationarity of the joint Gaussian process $(X^n, Y^n), n=1, 2, \ldots$ (the author also considers a double sided joint process). However, unlike  \cite{yang-kavcic-tatikonda2007,kim2010}, we do not assume the AR$(c)$ noise is stable. 
 \\
By Theorem~\ref{thm_FTFI}, and  restricting the channel input strategies to the time-invariant channel input strategies, $(\Lambda_t,K_{Z_t})=(\Lambda^\infty, K_Z^\infty), t=1, \ldots, n$, (not necessarily stationary)    then we have the following  representation\footnote{The variation of notation is judged necessary to distinguish it from the time-varying channel input strategies $(\Lambda_t, K_{Z_t})$ and corresponding   distributions ${\bf P}_{X_t|V_{t-1}, Y^{t-1}}={\bf P}_t(dx_t|v_{t-1},y^{t-1}),  t=1, \ldots, n$.}.
\begin{align}
&X_t^o = \Lambda^\infty \Big({V}_{t-1} - \widehat{V}_{t-1}^o\Big) + Z_t^o,   \hso X_1^o=Z_1^o,    \hst t=2, \ldots,n,  \label{Q_1_3_s1_new}  \\
&V_t=c V_{t-1} +W_t, \hso V_0=v_0, \hso t=1,\ldots, n, \label{Q_1_3_a_s1_a_new}  \\
&Y_t^o= X_t^o + V_t=  \Lambda^\infty \Big({V}_{t-1} - \widehat{V}_{t-1}^o\Big)+c V_{t-1} +W_t + Z_t^o, \hso t=2, \ldots,n, \label{Q_1_4_s1_new}\\
&Y_1^o=Z_1^o+c V_0 +W_1, \hso V_0=v_0,\label{Q_1_4_s1_a_new}\\
&Z_t^o\sim N(0, K_{Z}^\infty), \hso t=1, \ldots, n \hso \mbox{is a  Gaussian sequence,}\label{Q_1_5_s1_new} \\
&Z_t^o \hso  \mbox{is independent of}  \hso  (V^{t-1},X^{o,t-1},Y^{o,t-1}, {V}_0),  \hso t=1, \ldots, n,\label{Q_1_6_s1_new}\\
&Z^{o,n} \hso  \mbox{is independent of $(V^n,V_0)$}, \\
&cov( \left[ \begin{array}{c} W_t  \\ W_t +Z_t^o\end{array} \right], \left[ \begin{array}{c} W_t  \\ W_t +Z_t^o\end{array} \right]^T) = \left[ \begin{array}{cc} K_{W} & K_{W}  \\ K_{W} &K_{W} +K_{Z}^\infty\end{array} \right], \label{cp_e_ar2_s1_a_new}\\
&\frac{1}{n}  {\bf E}_{v_0} \Big\{\sum_{t=1}^{n} \big(X_t^o\big)^2\Big\}=\frac{1}{n}   \sum_{t=1}^n \big(\Lambda^\infty\big)^2  K_{t-1}^o  + K_{Z}^\infty  \leq \kappa,     \label{cp_e_ar2_s1_new}\\
&(\Lambda^\infty, K_{Z}^\infty)\in (-\infty, \infty) \times [0,\infty) \hst \mbox{are   non-random}, \label{Q_1_3_s1_e}\\
&{\bf P}_{X_t^o|V_{t-1}, Y^{o,t-1}}={\bf P}^\infty(dx_t|v_{t-1},y^{t-1}), \hso t=1, \ldots, n, \hso \mbox{that is, the distribution is time-invariant}
\end{align}  
where  $(\widehat{V}_t^o, K_t^o), t=1, \ldots, n$ satisfy  the generalized Kalman-filter  and time-invariant DRE, given below.

{\it  Generalized  Kalman-filter Recursion:} 
\begin{align}
&\widehat{V}_{t}^o = c \widehat{V}_{t-1}^o + M(K_{t-1}^o, \Lambda^\infty, K_{Z}^\infty) I_t^o, \hso \widehat{V}_{0}^o=v_0, \label{Q_1_8_s1_new} \\
&\hso = F(K_{t-1},\Lambda^\infty,K_{Z}^\infty) \widehat{V}_{t-1}^o + M(K_{t-1}^o, \Lambda^\infty, K_{Z}^\infty) Y_t^o, \hso \widehat{V}_{0}^o=v_0, \label{Q_1_8_s1_a_new}  \\
&I_t^o \tri Y_t^o - c \widehat{V}_{t-1}^o, \hso I_1^o=Z_1^o+W_1,  \hso t=2, \ldots, n, \label{Q_1_9_s1_new}\\
& \hso = \Big(\Lambda^\infty +c\Big)\Big(V_{t-1}- \widehat{V}_{t-1}^o\Big) + Z_t^o +W_t, \\
&M(K_{t-1}^o, \Lambda^\infty, K_{Z}^\infty)  \tri  \Big( K_{W} + c  K_{t-1}^o\Big(\Lambda^\infty + c \Big)\Big)\Big(K_{Z}^\infty+ K_{W} + \Big(\Lambda^\infty + c \Big)^2 K_{t-1}^o\Big)^{-1}, \label{Q_1_9_s1_new_g}  \\
& F(K_{t-1}^o,\Lambda^\infty, K_{Z}^\infty) \tri c -M(K_{t-1}^o, \Lambda^\infty, K_{Z}^\infty) \Big(\Lambda^\infty +c\Big), \label{Q_1_9_s1_new_gg}\\
& I_t^o, \hso t=1, \ldots, n, \hst \mbox{an orthogonal innovations process.}\label{Q_1_11_s1_new}
 \end{align}
 {\it Generalized Time-Invariant Difference Riccati  Equation:}
\begin{align}
K_{t}^o= & c^2 K_{t-1}^o  + K_{W} -\frac{ \Big( K_{W} + c K_{t-1}^o\Big(\Lambda^\infty + c \Big)\Big)^2}{ \Big(K_{Z}^\infty+ K_{W} + \Big(\Lambda^\infty + c \Big)^2 K_{t-1}^o\Big)}, \hst  K_t^o \geq0,  \hso K_{0}^o=0, \hst t=1, \ldots, n, 
 \label{Q_1_10_s1_new}
 \end{align}
We note that the Kalman-filter recursion (\ref{Q_1_8_s1_new}) is time-varying, but the DRE (\ref{Q_1_10_s1_new}) is time-invariant. \\ 
The analog of the error recursion (\ref{i_error_s1}),  for  time-invariant strategies, is the following.  

{\it Error Recursion of the Generalized Kalman-filter, $E_t^o\tri V_t-\widehat{V}_t^o, t=1, \ldots, n$:}
 \begin{align} 
&E_t^o=F(K_{t-1}^o,\Lambda^\infty,K_{Z}^\infty) E_{t-1}^o-M(K_{t-1}^o, \Lambda^\infty, K_Z^\infty)  \Big(Z_t^o+ W_t\Big)+W_t,  \; E_0^o=0, \; t=1, \ldots, n, \label{i_error_s1_ti}\\
&Z_t^o \in N(0, K_{Z}^\infty), \; t=1, 2, \ldots, n.
\end{align}
We note that recursion (\ref{i_error_s1_ti}) is linear time-varying.  Hence,   $\lim_{n \longrightarrow \infty } K_n^o =\lim_{n \longrightarrow \infty } {\bf E}_{s_0}\big\{\big(E_n^o\big)^2 \big\}$ is not expected to exist,  for arbitrary  $(F(K_{t-1}^o,\Lambda^\infty,K_{Z}^\infty),M(K_{t-1}^o, \Lambda^\infty, K_{Z}^\infty) ), t=1,2,  \ldots$. Indeed, the convergence properties of the  sequence $K_0^o, K_1^o, \ldots, K_n^o$ generated  by (\ref{Q_1_10_s1_new}), 
 as $n \longrightarrow \infty$,  are characterized by the   detectability and  stabilizability conditions  \cite{caines1988,kailath-sayed-hassibi} (which we introduce shortly). These conditions ensure existence of a unique nonnegative  limit, 
 $\lim_{n \longrightarrow \infty } K_n^o =K^\infty$, such that $K^\infty\geq 0$ is the unique solution of a generalized ARE and satisfies the  stability property:  
$\lim_{n \longrightarrow\infty} F(K_{n-1}^o,\Lambda^\infty,K_{Z}^\infty)=F(K^\infty,\Lambda^\infty,K_{Z}^\infty)\in (-1,1)$.

Next, we define the  characterization of the $n-$FTFI capacity, its per unit time limit, and the alternative definition, with the per unit time limit and maximization interchanged. \\

\begin{definition} Characterizations of asymptotic limits \\
\label{def_cc}
Consider the  characterization of the $n-$FTFI capacity of Theorem~\ref{thm_FTFI},    restricted to  the time-invariant strategies $(\Lambda_t=\Lambda^\infty, K_{Z_t}=K^\infty), t=1, \ldots, n$, as defined by (\ref{Q_1_3_s1_new})-(\ref{Q_1_10_s1_new}). \\
(a) The characterization of the $n-$FTFI capacity for time-invariant strategies is defined by   
\begin{align}
&C_{n}^o(\kappa,v_0) 
\tri  \sup_{\big(\Lambda^\infty, K_{Z}^\infty\big): \hso \frac{1}{n} \sum_{t=1}^n\big(\Lambda^\infty\big)^2 K_{t-1} ^o+ K_{Z}^\infty   \leq \kappa}    \sum_{t=1}^n   \frac{1}{2}\log\Big( \frac{\big(\Lambda^\infty+c\big)^2 K_{t-1}^o  + K_{Z}^\infty +K_{W}}{K_{W}}\Big)\label{Q_1_2_s1_new}\\
&\mbox{subject to: $K_t^o, t=1, \ldots,n$ satisfies recursion (\ref{Q_1_10_s1_new}) and $K_{Z}^\infty\geq 0 , t=1, \ldots, n$}\label{Q_1_2_s1_a_new}
\end{align}
provided the supremum exists in the set. The per unit time-limit is then defined by 
\begin{align}
&C^o(\kappa,v_0) 
\tri \lim_{n \longrightarrow\infty} \frac{1}{n} C_{n}^o(\kappa,v_0) .\label{Q_1_2_s1_new_a}
\end{align}
provided the supremum exists and the limit exists in $[0,\infty)$.  \\
(b) The characterization of the $n-$FTFI capacity for time-invariant strategies, with limit and maximization interchanged is defined by 
\begin{align}
C^\infty(\kappa,v_0) 
\tri & \sup_{\big(\Lambda^\infty, K_{Z}^\infty\big): \hso \lim_{n \longrightarrow \infty} \frac{1}{n} \sum_{t=1}^n  \big(\Lambda^\infty\big)^2 K_{t-1} ^o+ K_{Z}^\infty   \leq \kappa} \lim_{n \longrightarrow \infty}\frac{1}{n}   \sum_{t=1}^n   \frac{1}{2}\log\Big( \frac{\big(\Lambda^\infty+c\big)^2 K_{t-1}^o  + K_{Z}^\infty +K_{W}}{K_{W}}\Big) \label{Q_1_2_s1_new_b}
\end{align}
provide the limit exists in $[0,\infty)$ and the supremum also exists in the set.
\end{definition}

To ensure $C^\infty(\kappa,v_0)$ defined by  (\ref{Q_1_2_s1_new_b}) is well defined,   i.e.,  that the  optimal time-invariant channel input strategy or distribution ensures the limit exits and   $C^\infty(\kappa,v_0)$ is independent of $v_0$, we shall impose condition {\bf (C1)}.    We shall express  condition {\bf (C1)} in terms of   properties of generalized time-invariant  DREs and AREs, introduced  in the next section, from which answers to questions of Problem~\ref{prob_1_in} are obtained.

\subsection{Convergence Properties of Time-Invariant Generalized RDEs}
\label{sect:pre}
We recall that in  the study of mean-square estimation, and in particular, the  filtering theory, of time-invariant jointly Gaussian processes described  by linear recursions, driven by Gaussian noise processes, and of  jointly stationary Gaussian processes, the concepts of detectability and stabilizability, have been very effective \cite{kailath-sayed-hassibi,caines1988}. In this section, we summarize these concepts in relation to  the properties of generalized DREs and AREs. 

Let   $\{K_t, t=1, 2, \ldots, n\}$ denote a sequence that satisfies the time-invariant  generalized DRE with arbitrary initial condition 
\begin{align}
&K_{t}= c^2 K_{t-1}  + K_{W} -\frac{ \Big( K_{W} + c  K_{t-1}\Big(\Lambda + c \Big)\Big)^2}{ \Big(K_{Z}+ K_{W} + \Big(\Lambda + c \Big)^2 K_{t-1}\Big)}, \hst K_{0}=\mbox{given}, \hst t=1, \ldots, n.  \label{DRE_TI_g} 
\end{align}
We note that a solution  of (\ref{DRE_TI_g}) is a functional of the parameters of the right hand side, that is, $K_t\equiv K_t(c, K_W, \Lambda, K_Z, K_0), t=1, \ldots, n$. 
To discuss  the properties of the generalized DRE (\ref{DRE_TI_g}),  we introduce, as often done in the analysis of generalized DREs \cite{caines1988} and \cite[Section~14.7, page 540]{kailath-sayed-hassibi}, the following definitions. 
\begin{align}
&A\tri c,\hst C\tri \Lambda + c, \hst A^* \tri  c-K_W R^{-1}C, \hst B^{*,\frac{1}{2}} \tri   K_W^{\frac{1}{2}} B^{\frac{1}{2}}\label{ma_1}\\
&R \tri K_{Z} + K_W,  \hso  B \tri  1-K_W\big(K_Z+K_W\big)^{-1}.\label{m_2}
\end{align}
By (\ref{Q_1_9_s1_new_g}) and (\ref{Q_1_9_s1_new_gg}), we also have 
\begin{align} 
&M(K, \Lambda, K_Z) \tri \Big( K_W + A  K C\Big)\Big(R + \big(C\big)^2 K\Big)^{-1}, \label{ma_3a} \\
&F(K,\Lambda,K_Z) =A -M(K, \Lambda, K_Z)C.\label{ma_3}
\end{align}
The generalized algebraic Riccati equation (ARE) corresponding to  (\ref{DRE_TI_g}) is 
 \begin{align}
&K= c^2 K  + K_{W} -\frac{ \Big( K_{W} + c  K\Big(\Lambda + c \Big)\Big)^2}{ \Big(K_{Z}+ K_{W} + \Big(\Lambda + c \Big)^2 K\Big)}, \hst K\geq 0. \label{DRE_TI_gae} 
\end{align}

Next, we introduce the definition of asymptotic stability of the error recursion  (\ref{i_error_s1_ti}).\\

\begin{definition} Asymptotic stability\\
A solution $K\geq 0$ to the generalized ARE (\ref{DRE_TI_gae}), assuming it exists, is called stabilizing if $|F(K,\Lambda,K_Z)|<1$. In this case, we  say $F(K,\Lambda,K_Z)$ is asymptotically stable, that is, $|F(K,\Lambda,K_Z)|<1$. 
\end{definition}

With respect to any of the above  generalized DRE and ARE, we  define    the important notions of detectability, unit circle controllability, and stabilizability.  \\

\begin{definition} Detectability, Stabilizability, Unit Circle controllability \\
\label{def:det-stab}
(a) The pair  $\big\{A,C\big\}$ is called detectable if there exists a $G \in {\mathbb R}$ such that  $|A- G C|<1$ (stable).\\
(b) The pair $ \big\{A^*, B^{*,\frac{1}{2}}\big\}$ is called unit circle controllable if  there exists a $G \in {\mathbb R}$ such that   $|A^*- B^{*,\frac{1}{2}}G|\neq 1$.\\
(c) The pair $\big\{A^*, B^{*,\frac{1}{2}}\big\}$ is called stabilizable if  there exists a $G \in {\mathbb R}$ such that   $|A^*- B^{*,\frac{1}{2}}G|< 1$.
\end{definition}

The next theorem characterizes, detectability,   unit circle controllability, and stabilizability \cite{kailath-sayed-hassibi,vanschuppen2010}.\\
 
\begin{lemma}\cite{kailath-sayed-hassibi,vanschuppen2010}  Necessary and sufficient conditions for detectability, unit circle controllability, stabilizability\\
(a) The pair $\big\{A,C\big\}$ is detectable  if and only if there exists no eigenvalue, eigenvector $\{\lambda, x\}$, of $A$, i.e.,  $Ax=\lambda x$ such that $|\lambda|\geq 1$, and such that $Cx =0$\\
(b) The pair $\big\{A^*, B^{*,\frac{1}{2}}\big\}$ is unit circle controllable if and only if  there exists no  eigenvalue, eigenvector $\{\lambda, x\}$, $x A^*=x\lambda$,  such that   $|\lambda|=1$, and   such that  that $x B^{*,\frac{1}{2}} =0$.  \\
(c) The pair $\big\{A^*, B^{*,\frac{1}{2}}\big\}$ is stabilizable  if and only if there exists no eigenvalue, eigenvector $\{\lambda, x\}$, $x A^*=x \lambda$ such that $|\lambda|\geq 1$, and such that $x B^{*,\frac{1}{2}} =0$.
\end{lemma}

In the next theorem we summarize known results on sufficient and/or necessary conditions for convergence of solutions $\{K_t, t=1, 2, \ldots, n\}$ of the generalized time-invariant DRE, as $n \longrightarrow \infty$,  to a nonnegative  $K$,  which is the unique stabilizing solution of a corresponding generalized ARE.    \\

\begin{theorem}\cite{kailath-sayed-hassibi,caines1988} Convergence of time-invariant generalized DRE\\
\label{thm_ric}
Let  $\{K_t, t=1, 2, \ldots, n\}$ denote a sequence that satisfies the time-invariant  generalized DRE (\ref{DRE_TI_g}) with arbitrary initial condition, and $(A, C, A^*, B^{*, \frac{1}{2}})$ defined by (\ref{ma_1}), (\ref{m_2}). Then the following hold.\\
(1) Consider the generalized RDE (\ref{DRE_TI_g})  with zero initial condition, i.e., $K_{0}=0$, and assume,  the pair $\big\{A,C\big\}$ is detectable, and  the pair $\big\{A^*, B^{*,\frac{1}{2}}\big\}$ is unit circle controllable.\\
Then the  sequence $\{K_{t}: t=1, 2, \ldots, n\}$ that satisfies the generalized DRE (\ref{DRE_TI_g}),  with zero initial condition $K_{0}=0$,  converges to $K$, i.e., $\lim_{n \longrightarrow \infty} K_{n} =K$, where  $K$ satisfies   the ARE 
\bea
K= c^2 K + K_{W} -\frac{ \Big( K_W + c  K\big(\Lambda + c \big)\Big)^2}{ \Big(K_Z+ K_W + \big(\Lambda + c \big)^2 K\Big)} \label{ric_22}
\eea
 if and only if the pair $\big\{A^*, B^{*,\frac{1}{2}}\big\}$ is stabilizable.\\
(2) Assume,  the pair $\big\{A,C\big\}$ is detectable, and  the pair $\big\{A^*, B^{*,\frac{1}{2}}\big\}$ is unit circle controllable.  Then there exists a unique stabilizing solution $K\geq 0$ to the generalized ARE (\ref{DRE_TI_g}), i.e.,  such that,  $|F(K,\Lambda,K_Z)|<1$, if and only if  $\{A^*, B^{*,\frac{1}{2}}\}$ is stabilizable.\\
(3) If $\{A, C\}$ is detectable and $\{A^*, B^{*,\frac{1}{2}}\}$ is stabilizable,  then any solution $K_{t}, t=1, 2, \ldots,n$ to the generalized RDE (\ref{DRE_TI_g})  with arbitrary  initial condition, $K_{0}$ is such that $\lim_{n \longrightarrow \infty} K_{n} =K$, where $K\geq 0$ is the  unique solution of  the generalized ARE (\ref{DRE_TI_g}) with  $|F(K, \Lambda,K_Z)|<1$, i.e., it is stabilizing.
\end{theorem}

Theorem~\ref{thm_ric}.(1) follows by combining \cite[Lemma~14.2.1, page 507]{kailath-sayed-hassibi} of classical DREs and AREs  with  \cite[Section~14.7]{kailath-sayed-hassibi} of generalized DREs and AREs.  Theorem~\ref{thm_ric}.(2) is given in  \cite[Theorem~E.6.1, page 784]{kailath-sayed-hassibi}.  Theorem~\ref{thm_ric}.(3) is obtained from \cite[Theorem~4.2, page 164]{caines1988}, and also   \cite{kailath-sayed-hassibi}.

From  Theorem~\ref{thm_ric}, we can easily re-confirm {\bf Facts 2, 3} of Section~\ref{prel_ric}, as  shown  in the next lemma. \\

\begin{lemma} Properties of Solutions of DREs and AREs for different cases\\
\label{lem_pr_are}
Let $(A, C, A^*, B^{*, \frac{1}{2}})$ be  defined by (\ref{ma_1}), (\ref{m_2}). \\
(1) Suppose  $c \in (-1,1)$.  Then the pair $\{A, C\}$ is  detectable. \\
(2) Suppose $K_Z=0$.  Then the pair $\{A^*,B^{*,\frac{1}{2}}\}$ is unit circle controllable if and only if $|\Lambda| \neq 1$. \\
(3) Suppose $K_Z=0$. Then the pair $\{A^*,B^{*,\frac{1}{2}}\}$ is stabilizable  if and only if $|\Lambda| < 1$.\\
(4) Suppose $c\in (-1,1), K_Z=0$. The sequence    $\{K_t, t=1, 2, \ldots, n\}$ that satisfies  the generalized DRE with zero initial condition, i.e., 
\begin{align}
&K_{t}= c^2 K_{t-1}  + K_{W} -\frac{ \Big( K_{W} + c  K_{t-1}\big(\Lambda + c \big)\Big)^2}{ \Big( K_{W} + \big(\Lambda + c \big)^2 K_{t-1}\Big)}, \hst K_{0}=0, \hst t=1, \ldots, n \label{i_Q_1_10} 
\end{align}
converges to $K\geq 0$, i.e., $\lim_{n \longrightarrow \infty} K_n=K$, where  $K$ satisfies the generalized ARE (\ref{DRE_TI_gae}) if and only if the $\{A^*, B^{*,\frac{1}{2}}\}$ is stabilizable, equivalently, $|\Lambda|<1$.\\
(5) Suppose   $K_Z=0$, and $|\Lambda|\neq 1$,   with  the  corresponding ARE,     
\begin{align}
K= c^2 K + K_W -\frac{ \Big( K_W + c  K\big(\Lambda + c \big)\Big)^2}{ \Big( K_W + \big(\Lambda + c \big)^2 K\Big)}. \label{i_ric_3_n}
\end{align}
Then the  two solution, without the restriction $K\geq 0$, are  given by    
\bea
K=0, \hst K=\frac{K_W\Big(\big(\Lambda\big)^2-1\Big)}{\Big(\Lambda+ c\Big)^2}, \hso c\neq -\Lambda \label{i_rae_nu}
\eea
Moreover,  $K=0$ is the  unique and stabilizing solution $K\geq 0$ to (\ref{i_ric_3_n}), i.e., such that $|F(K,\Lambda,K_Z)|<1$, if and only if $|\Lambda|<1$.
\end{lemma}
\begin{proof} See Appendix~\ref{lem_pr_are_AP}. 
\end{proof} 

In the next remark we  make some comments on \cite[Theorem~6.1, see also Lemma~6.1]{kim2010}, i.e.,  that a zero variance of the innovations process of the channel input process is not optimal.\\

\begin{remark} Asymptotic stationarity of optimal process of \cite{kim2010}  \\
\label{rem_g-kalman_1}
 Consider the characterization of feedback capacity given in \cite[Theorem~6.1, $\Sigma$ satisfying eqn(61)]{kim2010}, in which the variance of the innovations process is replaced by a zero value (see comment below \cite[Theorem~6.1]{kim2010}). Then $\Sigma=0$ is one solution of  \cite[$\Sigma$ satisfying eqn(61)]{kim2010}. \\
We ask: what are necessary and/or sufficient conditions for  convergence $\lim_{n\longleftrightarrow} \Sigma_n= \Sigma$, where $\Sigma \geq 0$ is the unique limit that stabilizes the estimation error of the noise?\\
By  the multidimensional version of   Theorem~\ref{thm_ric}.(1), and Lemma~\ref{lem_pr_are}.(3),   then the  limit  $\lim_{n\longrightarrow \infty}\Sigma_n$ converges if and only if the stabilizability condition holds. For the AR$(c)$ noise model,  since the characterization of feedback capacity given   \cite[Theorem~6.1]{kim2010}, presupposes a zero variance of the innovations process, i.e., $K_Z^\infty=0$, then the value of feedback capacity \cite[Theorem~6.1, $C_{FB}=0,\forall \kappa \in [0,\infty)$]{kim2010} (see also  (\ref{i_kim_1a})-(\ref{i_ric_3_na}) with  $K^\infty=0$, which implies  $C^K(\kappa)=0, \forall \kappa \in [0,\infty)$).
\end{remark}

\subsection{Feedback Capacity  of  AGN Channels Driven by Time-Invariant Stable/Unstable AR$(c)$ Noise}
\label{sect:cor-solu}
In this section we analyze the asymptotic per unit time limit of the $n-$FTFI capacity of Definition~\ref{def_cc}, by making use of  the properties of generalized DREs and AREs of Section~\ref{sect:pre} to identify sufficient and necessary conditions,  such that  condition {\bf (C1)} holds.    Then we derive closed form expressions for $C^\infty(\kappa,v_0)=C^\infty(\kappa), \forall v_0$ defined by (\ref{Q_1_2_s1_new_b}), for Regime 1 given by (\ref{reg_1}), and  we show that feedback does not increase  $C^\infty(\kappa)$, for Regimes 2 and 3 given by (\ref{reg_2}) and (\ref{reg_3}). 

First,  we define the  main problem of  asymptotic analysis. \\

\begin{problem} Problem of feedback capacity $C^\infty(\kappa,s_0)$ for stable/unstable time-invariant AR$(c)$ noise \\
\label{prob_1}
Consider the  characterization of the $n-$FTFI capacity of Theorem~\ref{thm_FTFI}, 
and restrict the admissible strategies or distributions  to  the time-invariant  strategies or distributions, defined by (\ref{Q_1_3_s1_new})-(\ref{Q_1_3_s1_e}), which generate $(X^{o,n}, Y^{o,n})$.\\
Define the per unit time limit and maximum by 
\begin{align}
C^\infty(\kappa,v_0) \tri & \max_{{\cal P}_{[0,\infty]}^{\infty}(\kappa) }\lim_{n \longrightarrow\infty}\frac{1}{2n}     \sum_{t=1}^n  \log\Big( \frac{\big(\Lambda^\infty+c\big)^2 K_{t-1}  + K_{Z}^\infty +K_{W}}{K_{W}}\Big) \label{i_ll_2}
\end{align}
where the average power constraint is defined by 
\begin{align}
&{\cal P}_{[0,\infty]}^{\infty}(\kappa)\tri   \Big\{(\Lambda^\infty, K_Z^\infty): X_t^o=\Lambda^\infty \big({V}_{t-1} - \widehat{V}_{t-1}\big) + Z_t^o, \; X_1 =  Z_t^o, \hso t=2, \ldots,n, \nonumber \\
& Z_t^o \in N(0, K_{Z}^\infty), \hso K_Z^\infty\geq 0,  \hso \lim_{n \longrightarrow\infty} \frac{1}{n} {\bf E}_{v_0}\Big( \sum_{t=1}^n \big(X_t^o\big)^2\Big)  =\lim_{n \longrightarrow\infty}\frac{1}{n} \sum_{t=1}^n (\Lambda^\infty)^2 K_{t-1} + K_{Z}^\infty  \leq \kappa \Big\}.  \label{Q_1_9_nn}
\end{align}
Determine sufficient and/or necessary conditions such that \\
(a) the per unit time limit    exists, i.e., condition {\bf (C1)} holds,  and \\
(b) the maximum over $(\Lambda^\infty, K_Z^\infty)$ exists, and  the optimal strategy is such that  $C^\infty(\kappa,v_0)=C^\infty(\kappa)$ is  independent of  the initial state $v_0$.  
\end{problem}

In the next theorem we provide the answer to Problem~\ref{prob_1}, by invoking  Theorem~\ref{thm_ric}. \\

 \begin{theorem} Feedback capacity $C^\infty(\kappa,s_0)$ \\ 
 \label{lem_cov}
Consider the Problem~\ref{prob_1}, defined by (\ref{i_ll_2}), (\ref{Q_1_9_nn}). \\
Define the set 
\begin{align}
{\cal P}^\infty \tri & \Big\{(\Lambda^\infty, K_Z^\infty)\in (-\infty, \infty)\times [0,\infty): \nonumber \\
& \mbox{(i) the pair  $\{A, C\}\equiv \{A, C(\Lambda^\infty)\}$ is detectable,} \nonumber \\
& \mbox{(ii) the pair $\{A^*, B^{*,\frac{1}{2}}\}\equiv \{A^*(K_Z^\infty), B^{*,\frac{1}{2}}(K_Z^\infty)\}$ is stabilizable}\Big\}. \label{adm_set}
\end{align}
Then 
\begin{align}
& C^\infty(\kappa,v_0)=C^\infty(\kappa)\tri    \max_{\big(\Lambda^\infty, K_{Z}^\infty\big) \in {\cal P}^\infty(\kappa)} \frac{1}{2} \log\Big( \frac{\big(\Lambda^\infty+c\big)^2 K^\infty  + K_{Z}^\infty +K_{W}}{K_{W}}\Big) \label{ll_3}
\end{align}
that is, $C^\infty(\kappa,v_0)$ is independent of $v_0$, where  
\begin{align}
{\cal P}^\infty(\kappa)\tri& \Big\{(\Lambda^\infty, K_Z^\infty)\in {\cal P}^\infty: K_Z^\infty\geq 0, \hso \big(\Lambda^\infty\big)^2 K^\infty + K_{Z}^\infty \leq \kappa, \nonumber \\
& K^\infty= c^2 K^\infty + K_W -\frac{ \Big( K_W + c  K^\infty\big(\Lambda^\infty + c \big)\Big)^2}{ \Big(K_Z^\infty+ K_W + \big(\Lambda^\infty + c \big)^2 K^\infty\Big)}\nonumber \\
&\mbox{$K^\infty\geq 0$ is unique and stabilizable, i.e.,  $|F(K^\infty,\Lambda^\infty,K_Z^\infty)|<1$} \Big\} \label{ll_4}, \\
 F(K^\infty,\Lambda^\infty, K_{Z}^\infty) \tri & c -M(K^\infty, \Lambda^\infty, K_{Z}^\infty) \Big(\Lambda^\infty +c\Big), \label{ll_4_a} \\
M(K^\infty, \Lambda^\infty, K_{Z}^\infty)  \tri  & \Big( K_{W} + c  K^\infty\Big(\Lambda^\infty + c \Big)\Big)\Big(K_{Z}^\infty+ K_{W} + \Big(\Lambda^\infty + c \Big)^2 K^\infty\Big)^{-1} \label{ll_4_b}
\end{align}
provided there exists $\kappa\in [0,\infty)$ such that the set ${\cal P}^\infty(\kappa)$ is non-empty.\\
\CDC{Moreover, the maximum element $(\Lambda^\infty, K_Z^\infty) \in {\cal P}^\infty(\kappa)$, is such that, \\
 (i) if the noise is stable, i.e., $c\in (-1,1)$ then the input and the output processes $(X_t^o, Y_t^o), t=1, \ldots$ are asymptotic stationary,  and \\
 (ii) if the noise is unstable i.e., $c\notin (-1,1)$ then the input and the innovations  processes $(X_t^o, I_t^o), t=1, \ldots$ are asymptotic stationary. }
\end{theorem} 
 \begin{proof} The   sequence $\{K_{t}^o: t=1, 2, \ldots, n\}$  satisfies the time-invariant generalized DRE (\ref{Q_1_10_s1_new}),  with zero initial condition, $K_{0}^o=0$. Then for elements in the set ${\cal P}^\infty$,   an application of 
 Theorem~\ref{thm_ric}.(1), (2), states that the sequence generated by  (\ref{Q_1_10_s1_new})  converges, i.e.,  $\lim_{n \longrightarrow \infty} K_{n}^o =K^\infty$, where  $K^\infty=K^\infty(\Lambda^\infty, K_Z^\infty)\geq  0$ is the unique stabilizing solution of the  generalized ARE given in  (\ref{ll_4}). Hence, the following summands converge, and so the limits exist in $[0,\infty)$.
 \begin{align}
&\lim_{n \longrightarrow\infty}\frac{1}{n} \sum_{t=1}^n \Big(\big(\Lambda^\infty\big)^2 K_{t-1}^o + K_{Z}^\infty\Big)=\big(\Lambda^\infty\big)^2 K^\infty + K_{Z}^\infty, \label{conve_1} \\
 & \lim_{n \longrightarrow\infty}\frac{1}{2n}     \sum_{t=1}^n  \log\Big( \frac{\big(\Lambda^\infty+c\big)^2 K_{t-1}^o  + K_{Z}^\infty +K_{W}}{K_{W}}\Big)=\frac{1}{2} \log\Big( \frac{\big(\Lambda^\infty+c\big)^2 K^\infty  + K_{Z}^\infty +K_{W}}{K_{W}}\Big).\label{conve_2}
 \end{align} 
 This establishes the characterization of  the right hand side of   (\ref{ll_3}), and its   independence on $v_0$. \CDC{The last part of the theorem  follows from the asymptotic properties of the Kalman-filter, as follows.  For (i).  $E_t^o, t=1, \ldots$ is asymptotically stationary, which implies $X_t^o=\Lambda^\infty E_{t-1}^o + Z_t^o, t=1,\ldots$, the innovations process $I_t^o, t=1, \ldots$ and $Y_t^o=X_t^o+ V_t, t=1, \ldots$ are   asymptotically stationary.  Similarly for (ii), with the exception that $Y_t^o=X_t^o+ V_t, t=1, \ldots$ is not asymptotically stationary, because $V_t, t=1, \ldots $ is unstable.} 
 \end{proof}

Clearly,  the set ${\cal P}^\infty$, defined in  Theorem~\ref{lem_cov} characterizes  condition {\bf (C1)}, and  (\ref{ll_3}) characterizes the asymptotic limit of feedback capacity defined by (\ref{inter}).\\
In the next remark, we discuss some aspects of  Theorem~\ref{lem_cov}, and we show that ${\cal P}^\infty(\kappa) \subseteq {\cal P}^\infty$ is non-empty for some values of $\kappa \in [0,\infty)$.\\

\begin{remark} Comments on Theorem~\ref{lem_cov}\\
\label{rem_mt}
(1)  Theorem~\ref{lem_cov} characterizes the feedback capacity  $C^\infty(\kappa,s_0)=C^\infty(\kappa)$, 
 independently of $v_0$,  for AGN channels, driven by  stable or unstable AR$(c)$ noise, i.e., $c \in (-\infty,\infty)$. \\
(2) Let $(\Lambda^{\infty,*} K_Z^{\infty,*})\in {\cal P}^\infty(\kappa)$ denote the optimal pair for the optimization problem $C^\infty(\kappa)$. Then we need to characterize the set of all $\kappa \in [0,\infty)$ such that  $(\Lambda^{\infty,*} K_Z^{\infty,*}) \in {\cal P}^\infty(\kappa)$. \\
\CDC{
{\bf Case 1-Stable.} Suppose $c\in (-1,1)$.\\ (i)  Suppose $K_{Z}^\infty=0$.  By Lemma~\ref{lem_pr_are},   $\{A, C\}$ is detectable and $\{A^*, B^{*\frac{1}{2}}\}$ is stabilizable  if and only if $|\Lambda^\infty|<1$. For such a choice of  $(\Lambda^\infty, K_Z^\infty)\in {\cal P}^\infty$, defined by (\ref{adm_set}) then $K^\infty=0$, and   ${\cal P}^\infty(\kappa)$ in non-empty for all $\kappa \in [0,\infty)$.    \\
(ii) Suppose $\Lambda^\infty=0$. Then $K_{Z}^\infty=\kappa$, and  similar to  Lemma~\ref{lem_pr_are},  $\{A, C\}$ is detectable and $\{A^*, B^{*\frac{1}{2}}\}$ is stabilizable for all $\kappa \in (0,\infty)$, and the non-feedback channel input $X_t^o=Z_t^o, t=1,2, \ldots$  induces a strictly positive achievable rate.         \\
{\bf Case 2-Unstable.} Suppose $|c|\geq 1$. \\ 
(i) Suppose $K_{Z}^\infty=0$. By Lemma~\ref{lem_pr_are},  $\{A^*, B^{*\frac{1}{2}}\}$ is stabilizable  if 
and only if $|\Lambda^\infty|<1$, and 
there exists a $G \in (-\infty, \infty)$ such that $|A- GC|=|c- G(\Lambda^\infty+c)|<1$, i.e., for $G=1$, then 
$\{A,C\}$ detectable if   $|\Lambda^\infty|<1$. For such a choice of  $(\Lambda^\infty, K_Z^\infty)\in {\cal P}^\infty$, then   $K^\infty=0$, and hence   ${\cal P}^\infty(\kappa)$ in non-empty for all $\kappa \in [0,\infty)$.  \\
(ii) Suppose $\Lambda^\infty=0$. Then $K_{Z}^\infty=\kappa$, and   $\{A, C\}$ is detectable and $\{A^*, B^{*\frac{1}{2}}\}$ is stabilizable for all $\kappa \in (0,\infty)$, and the non-feedback channel input $X_t^o=Z_t^o, t=1,2, \ldots$  induces a strictly positive achievable rate.        \\
However, from  Case 1.(i) and Case 2.(i), since  $K_Z^{\infty,*}=0$ then $C^\infty(\kappa)=0, \forall \kappa \in [0,\infty)$. On the other hand, Case 1.(ii) and Case 2.(ii), with $\Lambda^\infty=0$,  $K_{Z}^\infty=\kappa \in (0,\infty)$ gives  a strictly positive non-feedback achievable rate. This is re-visited in Theorem~\ref{thm_nfb}.} 
\end{remark}

In the next lemma we give  necessary conditions for the optimization problem $C^\infty(\kappa)$ defined by (\ref{ll_3}).\\

 \begin{lemma} Necessary conditions for the optimization problem of Theorem~\ref{lem_cov}\\
\label{lemma_nc}
Suppose there exists a policy   $(\Lambda^{\infty,*}, K_Z^{\infty,*})\in {\cal P}^{\infty}(\kappa)$ for the optimization problem  $C^\infty(\kappa)$ in (\ref{ll_3}).  \\ 
 Define the Lagrangian by 
\begin{align}
&{\cal L}(\Lambda^{\infty}, K_Z^{\infty},K^\infty,\lambda)\tri \big(\Lambda^\infty+c\big)^2 K^\infty  + K_{Z}^\infty +K_{W} \nonumber \\
&\hst -\lambda_1 \Big\{\Big(K^\infty-c^2 K^\infty - K_W\Big)\Big(K_Z^\infty+ K_W + \big(\Lambda^\infty + c \big)^2 K^\infty\Big) +\Big( K_W + c  K^\infty\big(\Lambda^\infty + c \big)\Big)^2\Big\} \nonumber \\
& \hst    -\lambda_2\Big((\Lambda^\infty)^2 K^\infty + K_{Z}^\infty-\kappa\Big)-\lambda_3\Big(- K^\infty\Big) -\lambda_4\Big(-K_Z^\infty\Big), 
\label{nc_1}  \\
&\lambda\tri (\lambda_1,\lambda_2,\lambda_3 \lambda_4) \in {\mathbb R}^4.
\end{align} 
Then the following hold.\\
(i) Stationarity:
\begin{align}
&\frac{\partial}{\partial K_Z^\infty}{\cal L}(\Lambda^{\infty}, K_Z^{\infty},K^\infty,\lambda)\Big|_{\Lambda^{\infty}=\Lambda^{\infty,*}, K_Z^{\infty}=K_Z^{\infty,*},K^\infty=K^{\infty,*},\lambda=\lambda^{*}}=0,\label{nc_3}  \\
&\frac{\partial}{\partial \Lambda^\infty}{\cal L}(\Lambda^{\infty}, K_Z^{\infty},K^\infty,\lambda)\Big|_{\Lambda^{\infty}=\Lambda^{\infty,*}, K_Z^{\infty}=K_Z^{\infty,*},K^\infty=K^{\infty,*},\lambda=\lambda^{*}}=0,\label{nc_4}  \\
&\frac{\partial}{\partial K^\infty}{\cal L}(\Lambda^{\infty}, K_Z^{\infty},K^\infty,\lambda)\Big|_{\Lambda^{\infty}=\Lambda^{\infty,*}, K_Z^{\infty}=K_Z^{\infty,*},K^\infty=K^{\infty,*}, \lambda=\lambda^{*}}=0.   
 \label{nc_5} 
\end{align} 
Complementary Slackness:
\begin{align} 
&\lambda_2^*\Big(\Lambda^{\infty,*})^2 K^{\infty,*} + K_{Z}^{\infty,*} - \kappa\Big)=0, \hso \lambda_3^* K^{\infty,*}=0, \hso \lambda_4^* K_Z^{\infty,*}=0,\label{nc_6} \\
&\lambda_1^* \Big\{\Big(K^{\infty,*}-c^2 K^{\infty,*} - K_W\Big)\Big(K_Z^{\infty,*}+ K_W + \big(\Lambda^{\infty,*} + c \big)^2 K^{\infty,*}\Big) +\Big( K_W + c  K^{\infty,*}\big(\Lambda^{\infty,*} + c \big)\Big)^2\Big\}=0.
\end{align}
Primal Feasibility:
\begin{align}
&(\Lambda^{\infty,*})^2 K^{\infty,*} + K_{Z}^{\infty,*} \leq \kappa, \hso K_Z^{\infty,*}\geq 0, \hso K^{\infty,*}\geq 0,\label{nc_8}  \\
&\Big(K^{\infty,*}-c^2 K^{\infty,*} - K_W\Big)\Big(K_Z^{\infty,*}+ K_W + \big(\Lambda^{\infty,*} + c \big)^2 K^{\infty,*}\Big) +\Big( K_W + c  K^{\infty,*}\big(\Lambda^{\infty,*} + c \big)\Big)^2 \leq 0. \label{nc_9}  
\end{align}
Dual Feasibility:
\begin{align}
\lambda_1^* \geq 0, \hso \lambda_2^* \geq 0, \hso \lambda_3^*\geq 0, \hso \lambda_4^*\geq 0.
\end{align}
(ii) If  $K_Z^*=0$ then  $K^{\infty,*}=0$,  if $K^{\infty,*}=0$ then $K_Z^*=0$, and  if either $K_Z^*=0$ or $K^{\infty,*}=0$  then $C^\infty(\kappa)=0, \forall \kappa \in [0,\infty)$.\\
(iii) A necessary condition for existence of $\kappa \in (0,\infty)$ such that $C^\infty(\kappa)>0$ is  $\lambda_1^* > 0, \lambda_2^* > 0, \lambda_3^*= 0, \lambda_4^*= 0$, and (\ref{nc_9}) holds with equality.
 \end{lemma}
 \begin{proof} See Section~\ref{pr_lemma_nc}. 
 \end{proof}

%
%



We shall derive the main  Theorem~\ref{thm_sol}, after we introduce the lower bound on feedback and non-feedback capacity, of Section~\ref{sect:nfb} (below).

\subsubsection{Achievable Rates Without Feedback for Stable and Unstable AR$(c)$ Noise}
\label{sect:nfb}
By Remark~\ref{rem_mt}, for   $\Lambda^\infty=0$,  the optimization problem of Theorem~\ref{lem_cov} reduces to $C^\infty(\kappa)\Big|_{\Lambda^\infty=0}$, which is an  achievable rate, because the pair $\{A,C\}$ is detectable  and the pair $\{A^*, B^{*,\frac{1}{2}}\}$ is stabilizable, when  $\Lambda^\infty=0$.  For $\Lambda^\infty=0$,  by (\ref{Q_1_3_s1_new}), the channel input is an independent innovations process $X_t^o=Z_t^0, t=1, \ldots,n$, and hence the code does not use  feedback.  In the next theorem we calculate  $C^\infty(\kappa)\Big|_{\Lambda^\infty=0}$.\\

  \begin{theorem} Achievable rates without feedback for stable and unstable AR$(c)$ noise  \\ 
 \label{thm_nfb}
For $\Lambda^\infty=0$, define the set 
\begin{align}
{\cal P}_0^{\infty,nfb} \tri & \Big\{K_Z^\infty\in [0,\infty): \nonumber \\
& \mbox{(i) the pair  $\{A, C\}\Big|_{\Lambda^\infty=0}\equiv \{A, C(\Lambda^\infty)\}\Big|_{\Lambda^\infty=0}$ is detectable,} \label{nf_1}  \\
& \mbox{(ii) the pair $\{A^*, B^{*,\frac{1}{2}}\}\Big|_{\Lambda^\infty=0}\equiv \{A^*(K_Z^\infty), B^{*,\frac{1}{2}}(K_Z^\infty)\}\Big|_{\Lambda^\infty=0}$ is stabilizable}\Big\}. \label{nf_2} 
\end{align}
For $\Lambda^\infty=0$, define the   channel input and output processes  by
\begin{align}
&X_t^o =Z_t^o,     \hst t=1, \ldots,n,  \label{nf_1_a}  \\
&V_t=c V_{t-1} +W_t, \hso V_0=v_0, \label{nf_1_b}  \\
&Y_t^o= X_t^o + V_t=c {V}_{t-1}^o+W_t + Z_t^o,  \label{nf_1_d}\\
&Y_1^o=Z_1^o+c V_0 +W_1. \label{nf_1_e}
\end{align}

(1) A lower bound on non-feedback capacity $C^{\infty,nfb}(\kappa,v_0)$ is $C_{LB}^{\infty,nfb}(\kappa)$ given by
\begin{align}
& C^{\infty,nfb}(\kappa,v_0)\geq C^\infty(\kappa)\Big|_{\Lambda^\infty=0}= C_{LB}^{\infty,nfb}(\kappa)\tri     \max_{K_{Z}^\infty \in {\cal P}_0^{\infty,nfb}(\kappa)} \frac{1}{2} \log\Big( \frac{c^2 K^\infty  + K_{Z}^\infty +K_{W}}{K_{W}}\Big) \label{nf_4}
\end{align}
where  
\begin{align}
{\cal P}_0^{\infty,nfb}(\kappa)\tri& \Big\{ K_Z^\infty\in {\cal P}_0^{\infty,nfb}: K_Z^\infty\geq 0, \hso K_{Z}^\infty \leq \kappa, \nonumber \\
& K^\infty= c^2 K^\infty + K_W -\frac{ \Big( K_W + c^2  K^\infty\Big)^2}{ \Big(K_Z^\infty+ K_W +   c^2 K^\infty\Big)},\nonumber \\
&\mbox{$K^\infty\geq 0$ is unique and stabilizing, i.e.,  $|F^{nfb}(K^\infty,K_Z^\infty)|<1$} \Big\} \label{nf_5}, \\
 F^{nfb}(K^\infty, K_{Z}^\infty) \tri & c -M^{nfb}(K^\infty, K_{Z}^\infty)c,\label{nf_6} \\
M^{nfb}(K^\infty, K_{Z}^\infty)  \tri  & \Big( K_{W} + c^2  K^\infty\Big)\Big(K_{Z}^\infty+ K_{W} + c^2 K^\infty\Big)^{-1}\label{nf_7}
\end{align}
provided there exists $\kappa\in [0,\infty)$ such that the set ${\cal P}_0^{\infty,nfb}(\kappa)$ is non-empty.\\
Moreover,  $C_{LB}^{\infty,nfb}(\kappa)$ is an achievable   rate without feedback,  independent of the initial state $V_0=v_0$, and\\
\CDC{
 (i) if the noise is stable, i.e., $c\in (-1,1)$ then the input and the output processes $(X_t^o, Y_t^o), t=1, \ldots$ are asymptotically stationary,  and \\
 (ii) if the noise is unstable i.e., $c\notin (-1,1)$ then the  input and the innovations  processes $(X_t^o, I_t^o), t=1, \ldots$ are asymptotic stationary. }\\
%
%
%
%
(2) The lower bound on non-feedback capacity of  (1) is given by  
\begin{align}
& C_{LB}^{\infty,nfb}(\kappa)= \frac{1}{2} \log\Big( \frac{c^2 K^{\infty,*}  + \kappa +K_{W}}{K_{W}}\Big), \hst \forall \kappa \in {\cal K}^{\infty, nfb}(c,K_w) \label{nf_8}
\end{align}
where $K^{\infty,*}\geq 0$ is unique and stabilizing, and $K_Z^{\infty,*}$,  are given by  
\begin{align}
&K^{\infty,*}= \left\{ \begin{array}{lll}  \frac{- \kappa \Big(1-c^2\Big)-K_W + \sqrt{\Big(\kappa\big(1-c^2\big) +K_W\Big)^2+4c^2 K_W \kappa}}{2c^2}, & \forall  c\neq 0, & \forall \kappa \in {\cal K}^{\infty, nfb}(c, K_{W}),  \\
\frac{\kappa K_W}{\kappa+K_W}, & c=0, & \forall \kappa \in [0,\infty),  \end{array} \right. 
 \label{nf_9}\\
 &K_Z^{\infty,*}=\kappa, \\
&{\cal K}^{\infty, nfb}(c, K_{W})\tri \Big\{\kappa \in [0,\infty): \hso K^\infty\geq 0 \Big\}=[0,\infty).\label{nf_10}
\end{align}
\end{theorem} 
 \begin{proof} (1) Note that by setting  $\Lambda^\infty=0$, then the representation of channel input $X^{o,n}$ defined by (\ref{Q_1_3_s1_new})-(\ref{Q_1_10_s1_new}) is used without feedback, and this is a lower bound on the non-feedback capacity. Hence, the statements follow from   Theorem~\ref{lem_cov}, as a special case.\\
 (2) This follows from (1), since the optimal $K_Z^\infty$ occurs on the boundary, i.e.,  $K_Z^{\infty,*}=\kappa$. Then by substituting $(\Lambda^{\infty,*},K_Z^{\infty,*})=(0,\kappa)$ into the generalized ARE of the constraint (\ref{nf_5}), we obtain 
 \bea
 c^2 \Big(K^{\infty,*}\Big)^2 +K^\infty \Big\{\kappa\Big(1-c^2\Big) +K_W \Big\} -K_W \kappa=0, \hso c\neq 0. \label{gare_nf}
 \eea
 Hence, the unique and  non-negative stabilizing   solution of the generalized ARE (\ref{gare_nf})  is given by  (\ref{nf_9}), for $c\neq 0$. For $c=0$, then the generalized ARE reduces to $K^{\infty,*} \big\{\kappa  +K_W \big\} -K_W \kappa=0$, and hence all equations under (2) are obtained.   The validity of the stability condition, i.e.,  $|F^{nfb}(K^\infty, K_{Z}^\infty)|<1$, although it is ensured by the conditions, it is shown   in Appendix~\ref{pr_stab_Lam0}.  
 \end{proof}

\begin{remark} On the achievable  rate without feedback of Theorem~\ref{thm_nfb}\\
By Theorem~\ref{thm_nfb}.(2) we deduce that the lower bound $C_{LB}^{\infty,nfb}(\kappa), \kappa \in [0,\infty)$, holds for stable and unstable AR(1) noise.    We should explicitly mention the two special cases $c=0$ and $|c|=1$. \\
(1) For $c=0$, we recover, as expected, the capacity of the AGN channel with memoryless noise. \\
(2) For $|c|=1$, we obtain the lower bound $C_{LB}^{\infty,nfb}(\kappa)$ on the non-feedback capacity  $C^{\infty,nfb}(\kappa)$, given by   
\begin{align}
& C_{LB}^{\infty,nfb}(\kappa)= \frac{1}{2} \log\Big( \frac{c^2 K^{\infty,*}  + \kappa +K_{W}}{K_{W}}\Big), \\
&K^{\infty,*}=   \frac{-K_W + \sqrt{K_W^2+4 K_W \kappa}}{2}, \hso K_Z^{\infty,*}=\kappa \in [0,\infty).
\end{align}
The above choice of a channel input strategy ensures the pair  $\{A, C\}\Big|_{\Lambda^{\infty,*}=0}= \{c, c\}$ is detectable, and the pair $\{A^*, B^{*,\frac{1}{2}}\}\Big|_{K_Z^{\infty,*}=\kappa}= \{c-\frac{K_W}{\kappa+K_W}, K_W^{\frac{1}{2}}\sqrt{\big(1-K_W\big(\kappa+K_W\big)^{-1}}\}$ is stabilizable, for any $\kappa \in (0,\infty)$.
\end{remark}

\subsubsection{Feedback Capacity for Stable and Unstable AR$(c)$ Noise}
\label{sect:fb}
%
%
 Next  we derive closed form expressions for  the feedback capacity, by solving  the optimization problem of Theorem~\ref{lem_cov}, i.e., (\ref{ll_3}). \\
 First, by the definition of the sets ${\cal P}^\infty$ and ${\cal P}_0^{\infty,nfb}$ of Theorem~\ref{lem_cov} and Theorem~\ref{thm_nfb},   we have 
\bea
 {\cal P}^\infty = {\cal P}_0^{\infty,nfb}\bigcup {\cal P}^{\infty,fb},  \hst  {\cal P}^{\infty,fb} \tri \Big\{(\Lambda^\infty, K_Z^\infty)\in{\cal P}^\infty: \Lambda^{\infty} \neq 0\Big\}.
 \eea
 Thus, if $ (\Lambda^\infty, K_Z^\infty)\in {\cal P}^{\infty,fb}$ then the channel input process applies feedback, and if $K_Z^\infty\in {\cal P}_0^{\infty,nfb}$ then the channel input process does not apply feedback.\\

 \begin{theorem} Feedback capacity-solution of  optimization problem of Theorem~\ref{lem_cov} \\ 
 \label{thm_sol}
(1)  The   non-zero feedback capacity $C^\infty(\kappa)$ defined by (\ref{ll_3}), for a  stable and unstable AR$(c)$ noise, i.e, $c \in (-\infty,\infty)$,  with $c\neq 0, c\neq 1$, occurs in the set ${\cal P}^{\infty,fb}$, and is  given, as follows.
\begin{align}
&C^\infty(\kappa)= \frac{1}{2}  \log\Big( \frac{\big(\Lambda^{\infty,*}+c\big)^2 K^{\infty,*}  + K_{Z}^{\infty,*} +K_{W}}{K_{W}}\Big), \hst \forall \kappa \in {\cal K}^{\infty}(c, K_{W})  \label{sol_1}    \\
&\hst \hst = \frac{1}{2}\log\Big( \frac{c^2\Big( \big(c^2-1\big)\kappa +K_W\Big)}{\Big(c^2-1\Big)K_W}\Big),\\
&\Lambda^{\infty,*}= \frac{\kappa\Big(1-c^2\Big)+K_W + c^2 K^{\infty,*}}{c\Big(c^2-2\Big)K^{\infty,*}}\in (-\infty,\infty),  \label{sol_2}\\
&K_Z^{\infty,*}+\Big(\Lambda^{\infty,*}\Big)^2K^{\infty,*}=\kappa, \label{sol_3}\\
&{\cal K}^{\infty}(c, K_{W})\tri \Big\{\kappa \in [0,\infty): \hso K^{\infty,*}> 0, \hso K_Z^{\infty,*}>0 \Big\}, \hst c\in (-\infty,\infty), \hso   c \neq 0, \; c\neq 1.\label{sol_5}
\end{align}
where $K^{\infty,*}$ is the unique positive and stabilizing solution, i.e., $|F(K^{\infty,*},\Lambda^{\infty,*},K_Z^{\infty.*})|<1$,  of the quadratic equation
\begin{align}
&c^4 \Big(c^2-1\Big)\Big(K^{\infty,*}\Big)^2+ c^4 \Big(\Big(1-c^2\Big)\kappa +K_W\Big)K^{\infty,*}+ \Big(\Big(1-c^2\Big)\kappa +K_W\Big)^2 + 4\Big(c^2-1\Big)K_W \kappa-c^4\kappa K_W=0. \label{sol_4}
\end{align}
Further, for any $\kappa \in {\cal K}^{\infty}(c, K_{W})$, then 
\begin{align}
&K^{\infty,*} = \frac{\kappa\big(c^2-1\big)^2-K_W}{c^2\big(c^2-1\big)} \in (0,\infty), \\
&\Lambda^{\infty,*}= \frac{c K_W}{\kappa\big(c^2-1\big)^2-K_W}\in (-\infty,\infty),\\
&K_Z^{\infty,*}= \frac{\kappa\big(c^2-1\big)\Big(\kappa \big(c^2-1\big)^2-K_W\Big)-K_W^2}{\big(c^2-1\big)\Big(\kappa \big(c^2-1\big)^2-K_W\Big)} \in (0,\infty).
\end{align}
(2) The  non-zero feedback capacity $C^\infty(\kappa)$, $\kappa\in {\cal K}^\infty(c,K_W)$ of part (1), is restricted  to the  region:\\
(a) ${\cal K}^\infty(c,K_W) \tri \Big\{ \kappa\in [0, \infty): 1 < c^2 < \infty, \kappa > \frac{K_W+K_W \sqrt{4c^2-3}}{2\big(c^2-1\big)^2}  \Big\}$.

(3) A non-zero feedback capacity   $C^\infty(\kappa)$, i.e., with $\Lambda^\infty \neq 0$,  does not exist   for the two regions:\\
(a) $\kappa \in {\cal K}^{\infty,nfb}(c) \tri \Big\{\kappa\in [0,\infty): 0 \leq c^2 \leq 1\Big\}$; 

(b) $\kappa \in {\cal K}^{\infty,nfb}(c,K_W) \tri \Big\{ \kappa\in [0,\infty): 1 < c^2 <\infty,\hso  \kappa \leq \frac{K_W+K_W \sqrt{4c^2-3}}{2\big(c^2-1\big)^2}  \Big\}$.
\end{theorem} 
 \begin{proof} See Appendix~\ref{pr_thm_sol}. 
  \end{proof}

From the previous theorem it then follows the next theorem, that states  feedback does not increase capacity $C^\infty(\kappa)$ defined by   (\ref{ll_3}),   for the two regions   ${\cal K}^{\infty,nfb}(c)$ and  ${\cal K}^{\infty,nfb}(c,K_W)$.\\

\begin{theorem} Feedback does not increase  capacity for certain regions\\
\label{them_nfb_a}
Feedback does not increase   capacity $C^\infty(\kappa)$ defined by (\ref{ll_3}),  for the two regions:\\
(a) $\kappa \in {\cal K}^{\infty,nfb}(c) \tri \Big\{\kappa\in [0,\infty): 0 \leq c^2 \leq 1\Big\}$; 

(b) $\kappa \in {\cal K}^{\infty,nfb}(c,K_W) \tri \Big\{ \kappa\in [0,\infty): 1 < c^2 <\infty,\hso  \kappa \leq \frac{K_W+K_W \sqrt{4c^2-3}}{2\big(c^2-1\big)^2}  \Big\}$.
\end{theorem}
\begin{proof}
  By Theorem~\ref{thm_sol}.(3) we deduce that, if $\Lambda^\infty \neq 0$, i.e., if feedback is used,  then  there does not exists  a non-zero value of   $C^\infty(\kappa)$, for $\kappa \in {\cal K}^{\infty,nfb}(c)\bigcup {\cal K}^{\infty,nfb}(c,K_W)$.  On the other hand, by Theorem~\ref{thm_nfb}, an achievable rate without feedback exists for  all $\kappa \in {\cal K}^{\infty,nfb}(c)\bigcup {\cal K}^{\infty,nfb}(c,K_W)$, by letting $\Lambda^\infty =0$, and $C_{LB}^{\infty,nfb}(\kappa)$ is a lower bound on capacity without feedback. 
\end{proof}

%

\ \

\begin{remark} Implications of Theorem~\ref{thm_sol} on operational non-feedback capacity for  unstable  noise. \\
 Theorem~\ref{thm_sol}.(3), when combined with Theorem~\ref{thm_nfb}, states  that, one does not need to presuppose, as usually done in information theory literature \cite{cover-thomas2006,blahut1987,gallager1968}, that  the joint  channel input and output process, and noise,  are  stationary to show existence of  achievable non-feedback rates, as done, in the  water-filling solution. We shall give a tighter lower bound on the   non-feedback capacity   in Corollary~\ref{thm_nfb_g}. 
\end{remark}

Another lower bound on feedback capacity is presented in Section~\ref{sect:nc} (below).

\subsubsection{Achievable Feedback Rates with Noise Cancellation for Stable AR$(c)$ Noise} 
\label{sect:nc} 
Let us recall the representation of channel input described by (\ref{Q_1_3_s1_new})-(\ref{Q_1_3_s1_e}). If we let $\Lambda^{\infty}=-c$, then (\ref{Q_1_4_s1_new}) and (\ref{Q_1_4_s1_a_new}), reduce to the equations $Y_t^o= c \widehat{V}_{t-1}^o +W_t + Z_t^o, t=2, \ldots, n, 
Y_1^o=Z_1^o+c V_0 +W_1$.  This means, for  the choice  of strategy $\Lambda^\infty=-c$, the channel output process at time
 $t$, $Y_t^{o}$   is driven by past observations,  and  $(Z_t^{o}, W_t)$, which mean  $V_t$ and $Y_t^o$ are correlated. Hence, if the resulting pair $\{A, C\}\Big|_{\Lambda^\infty=-c}\equiv \{A, C(\Lambda^\infty)\}\Big|_{\Lambda^\infty=-c}$ is detectable,  and  the pair $\{A^*, B^{*,\frac{1}{2}}\}\Big|_{\Lambda^\infty=-c}\equiv \{A^*(K_Z^\infty), B^{*,\frac{1}{2}}(K_Z^\infty)\}\Big|_{\Lambda^\infty=-c}$ is stabilizable, then we can have a corresponding rate which is achievable.

In the next theorem, we give an achievable rate with noise cancellation, that corresponds to   $C^\infty(\kappa)\Big|_{\Lambda^\infty=-c}$. \\

  \begin{theorem} Achievable rates with feedback and noise cancellation for  stable AR$(c)$ noise \\ 
 \label{thm_nc}
For $\Lambda^\infty=-c$, define the set 
\begin{align}
{\cal P}^{\infty,nc} \tri & \Big\{K_Z^\infty\in [0,\infty): \nonumber \\
& \mbox{(i) the pair  $\{A, C\}\Big|_{\Lambda^\infty=-c}\equiv \{A, C(\Lambda^\infty)\}\Big|_{\Lambda^\infty=-c}$ is detectable,} \label{nc_1}  \\
& \mbox{(ii) the pair $\{A^*, B^{*,\frac{1}{2}}\}\Big|_{\Lambda^\infty=-c}\equiv \{A^*(K_Z^\infty), B^{*,\frac{1}{2}}(K_Z^\infty)\}\Big|_{\Lambda^\infty=-c}$ is stabilizable}\Big\} \label{nc_2} 
\end{align}
and consider  the  channel input and output processes, given by (\ref{Q_1_3_s1_new})-(\ref{Q_1_10_s1_new}), with $\Lambda^\infty=-c$:
\begin{align}
&X_t^o = -c \Big({V}_{t-1} - \widehat{V}_{t-1}^o\Big) + Z_t^o,   \hso X_1^o=Z_1^o, \hso   t=2, \ldots,n,  \label{nc_3_a}  \\
&V_t=c V_{t-1} +W_t, \hso V_0=v_0,\label{nc_5_a}  \\
&Y_t^o= X_t^o + V_t=c \widehat{V}_{t-1}^o+W_t + Z_t^o,  \label{nc_5_a}
\end{align}
(1) A lower bound on the feedback capacity is  $C_{LB}^{\infty,nc}(\kappa)$ given by 
\begin{align}
&C^\infty(\kappa)\geq  C^\infty(\kappa)\Big|_{\Lambda^\infty=-c}= C_{LB}^{\infty,nc}(\kappa)\tri  \max_{K_{Z}^\infty \in {\cal P}^{\infty,nc}(\kappa)} \frac{1}{2} \log\Big( \frac{ K_{Z}^\infty +K_{W}}{K_{W}}\Big) \label{nf_4}
\end{align}
where  
\begin{align}
{\cal P}^{\infty,nc}(\kappa)\tri& \Big\{ K_Z^\infty\in {\cal P}^{\infty,nc}: K_Z^\infty\geq 0, \hso c^2 K^\infty+ K_{Z}^\infty \leq \kappa, \nonumber \\
& K^\infty= c^2 K^\infty + K_W -\frac{ K_W^2}{ K_Z^\infty+ K_W},\nonumber \\
&\mbox{$K^\infty\geq 0$ is unique and stabilizing, i.e.,  $|F^{nc}(K^\infty,K_Z^\infty)|<1$} \Big\} \label{nf_5_new}, \\
 F^{nc}(K^\infty, K_{Z}^\infty) &\tri  c  \label{nf_6_new}
\end{align}
provided there exists $\kappa\in [0,\infty)$ such that the set ${\cal P}^{\infty,nc}(\kappa)$ is nonempty.\\
Moreover,  $C_{LB}^{\infty,nc}(\kappa)$ is achievable and    is independent of the initial state $V_0=v_0$.\\ 
(2) The lower bound of (1) is given by  
\begin{align}
& C_{LB}^{\infty,nc}(\kappa)= \frac{1}{2} \log\Big( \frac{K_Z^{\infty,*}  +K_{W}}{K_{W}}\Big), \hso c\in (-1,1), \hst \forall \kappa \in [0,\infty)
  \label{nf_8}
\end{align}
where $K_Z^{\infty,*}$ and the unique and stabilizing  $K^{\infty,*}\geq 0$, are given by  
\begin{align}
&K^{\infty,*}=  \frac{K_Z^{\infty,*} K_W}{\Big(K_Z^{\infty,*}+K_W\Big)\Big(1-c^2\Big)}, \hso  c\in (-1,1), \hso \kappa \in [0,\infty), 
 \label{nf_9_new_a}\\
&K_Z^{\infty,*}=  \frac{-  \Big(K_W-\kappa \big(1-c^2\big)\Big) + \sqrt{ \Big(K_W-\kappa \big(1-c^2\big)\Big)^2      +4\kappa K_W \Big(1-c^2\Big)^2}}{2\Big(1-c^2\Big)}, \hso  c\in (-1,1), \hso \kappa \in [0,\infty).
 \label{nf_9_new_new}
\end{align}
\end{theorem} 
 \begin{proof} (1) By setting  $\Lambda^\infty=-c$, in  (\ref{Q_1_3_s1_new})-(\ref{Q_1_10_s1_new}), then the representation of channel input $X_t^{o}$ cancels the noise $V_{t-1}$ in the output $Y_t^o$, and (\ref{nc_3_a})-(\ref{nc_5_a}) are obtained. Clearly, the set ${\cal P}^{\infty,nc}$ is non-empty, since the pair  $\{A, C\}\Big|_{\Lambda^\infty=-c}= \{c,0\}$ is detectable, for $c\in (-1,1), c \neq 0$. Hence, the statements follow from    Theorem~\ref{lem_cov}, as a special case.\\
 (2) This follows from (1), by simple computations. In particular,  from the equation of $K^{\infty}$ that appears in constraint (\ref{nf_5_new}), we obtain
 \bea
 K^{\infty}= \frac{K_Z^\infty K_W}{\Big(K_Z^{\infty}+K_W\Big)\Big(1-c^2\Big)} \geq 0, \hso \mbox{if and only if} \hso |c|<1.
 \eea 
 Since the optimal $K_Z^\infty=K_Z^{\infty,*}$ is such that the average constraint holds with equality, then 
 $c^2 K^{\infty,*}+ K_{Z}^{\infty,*}= \kappa$ Substituting $K^{\infty,*}$, into the average constraint we arrive at the quadratic equation
\bea
\Big(1-c^2\Big)\Big( K_Z^{\infty,*} \Big)^2+\Big(K_W- \kappa\big(1-c^2\big)\Big)K_Z^{\infty,*} -K_W \kappa \Big(1-c^2\Big)=0, \hso |c|<1.
\eea
The rest follows from the above equation and the fact that $K_Z^{\infty,*}\geq 0$ belongs to the set ${\cal P}^{\infty,nc}$.
 \end{proof}
 
In the lemma  we show that the achievable rate without feedback of Theorem~\ref{thm_nfb},  is higher than the   achievable rate with feedback and  noise cancellation of Theorem~\ref{thm_nc}, when the noise is stable AR$(c)$. \\

\begin{lemma} On the comparison of lower bounds   of Theorem~\ref{thm_nfb} and Theorem~\ref{thm_nc}\\
\label{lemma_nc-nf}
For $c\in (-1,1), c \neq 0$ then the following inequality holds.
\bea
 C_{LB}^{\infty,nc}(\kappa) <  C_{LB}^{\infty,nfb}(\kappa), \hso \forall \kappa \in (0,\infty).
\eea
\end{lemma} 
 \begin{proof} See Appendix~\ref{pr_lemma_nc-nf}. \end{proof}

\subsection{Lower Bound on Characterization of Capacity without Feedback for Time-Invariant Channel Input Distributions} 
By recalling the lower bound of Corollary~\ref{thm_FTWFI},  and by restricting the non-feedback strategies to the time-invariant channel input strategies, $(\overline{\Lambda}_t,K_{Z_t})=(\overline{\Lambda}^\infty, K_Z^\infty), t=1, \ldots, n$, then we have a representation of $(X^{o,n}, Y^{o,n})$, as in Section~\ref{sect:char_fb}. The next corollary gives an achievable  lower bound on the characterization of non-feedback  capacity, as a time-invariant optimization problem, which is  
  analogous to Theorem~\ref{lem_cov}. \\

\begin{corollary} Achievable lower bound on non-feedback capacity  \\ 
 \label{thm_nfb_g}
Consider the lower bound on the  characterization of the $n-$FTFI capacity of Corollary~\ref{thm_FTWFI}.(b),  
and restrict the admissible strategies to  the time-invariant  strategies, $(\overline{\Lambda}_t, K_{Z_t})=(\overline{\Lambda}^\infty, K_Z), t=1, \ldots,n$,   which generate $(X^{o,n}, Y^{o,n})$. \\
Define the lower bound on non-feedback  capacity  by 
\begin{align}
&C^{\infty,nfb}(\kappa,v_0)\geq C_{LB}^{\infty,nfb}(\kappa,v_0) 
\tri
 \sup_{\big(\overline{\Lambda}^\infty, K_{Z}^\infty\big): \hso \lim_{n \longrightarrow \infty} \frac{1}{n} \sum_{t=1}^n  \big(\overline{\Lambda}^\infty\big)^2 K_{X_{t-1}}^o+ K_{Z}^\infty   \leq \kappa} \Big\{\nonumber \\
 & \hst \hst \lim_{n \longrightarrow \infty}\frac{1}{2n}   \sum_{t=1}^n  \log\Big( \frac{\big(\overline{\Lambda}^\infty-c\big)^2 K_{t-1}^o  + K_{Z}^\infty +K_{W}}{K_{W}}\Big)\Big\} \label{int_nfb_1}\\
&\mbox{subject to: $K_Z^\infty\geq 0$, $K_t^o$ and $K_{X_t}^o, t=1, \ldots,n$ satisfy the generalized DRE and Lyapunov equation}\nonumber \\ 
&K_{t}^o=  \big(\overline{\Lambda}^\infty\big)^2 K_{t-1}^o  + K_{Z}^\infty -\frac{ \Big( K_{Z}^\infty + \overline{\Lambda}^\infty K_{t-1}^o\Big(\overline{\Lambda}^\infty - c \Big)\Big)^2}{ \Big(K_{Z}^\infty+ K_{W} + \Big(\overline{\Lambda}^\infty - c \Big)^2 K_{t-1}^o\Big)}, \hst  K_t^o \geq0,  \hso K_{0}^o=0, \hst t=1, \ldots, n, \label{reg_23_in_a} \\
&K_{X_t}^o =\big(\overline{\Lambda}^\infty\big)^2 K_{X_{t-1}}^o +K_Z^\infty, \hso K_{X_0}=0.\label{reg_23_in}
\end{align}
Define the set\footnote{The definition of detectability and stabilizability follow from the DRE (\ref{reg_23_in_a}), while $\overline{\Lambda}^\infty\in (-1,1)$ is necessary for $K_{X_n}^o$ to converge as $n\longrightarrow \infty.$} 
\begin{align}
{\cal P}^{\infty,nfb} \tri & \Big\{(\overline{\Lambda}^\infty, K_Z^\infty)\in (-1, 1)\times [0,\infty): \nonumber \\
& \mbox{(i) the pair  $\{A, C\}\equiv \{A(\overline{\Lambda}^\infty), C(\overline{\Lambda}^\infty)\}$ is detectable,} \nonumber \\
& \mbox{(ii) the pair $\{A^*, B^{*,\frac{1}{2}}\}\equiv \{A^*(\overline{\Lambda}^\infty,K_Z^\infty), B^{*,\frac{1}{2}}(K_Z^\infty)\}$ is stabilizable}\Big\}. \label{adm_set_2}
\end{align}
where 
\begin{align}
&A\tri \overline{\Lambda}^\infty, \hso C\tri \overline{\Lambda}^\infty-c, \hso A^*=\overline{\Lambda}^\infty -K_Z^\infty R^{-1} C, \hso B^{*,\frac{1}{2}}\tri K_Z^{\infty,\frac{1}{2}} B^\frac{1}{2}, \\
& R\tri K_Z^\infty +K_W, \hso B\tri 1-K_Z^\infty \big(K_Z^\infty+K_W\big)^{-1}. \label{adm_set_2_2}
\end{align}
Then the lower bound on non-feedback capacity is given by the time-invariant optimization problem: 
\begin{align}
& C_{LB}^\infty(\kappa,v_0)=C_{LB}^\infty(\kappa)\tri    \max_{\big(\overline{\Lambda}^\infty, K_{Z}^\infty\big) \in {\cal P}^{\infty,nfb}(\kappa)} \frac{1}{2} \log\Big( \frac{\big(\overline{\Lambda}^\infty-c\big)^2 K^\infty  + K_{Z}^\infty +K_{W}}{K_{W}}\Big) \label{ll_3_3}
\end{align}
 where  
\begin{align}
{\cal P}^{\infty,nfb}(\kappa)\tri& \Big\{(\overline{\Lambda}^\infty, K_Z^\infty)\in {\cal P}^{\infty,nfb}: K_Z^\infty\geq 0, \hso \frac{K_Z^\infty}{1-\big(\overline{\Lambda}^\infty\big)^2} \leq \kappa, \nonumber \\
& K^\infty= \big(\Lambda\big)^2 K^\infty + K_Z^\infty -\frac{ \Big( K_Z^\infty + \overline{\Lambda}^\infty  K^\infty\big(\overline{\Lambda}^\infty - c \big)\Big)^2}{ \Big(K_Z^\infty+ K_W + \big(\overline{\Lambda}^\infty - c \big)^2 K^\infty\Big)}\nonumber \\
&\mbox{$K^\infty\geq 0$ is unique and stabilizable, i.e.,  $|F^{nfb}(K^\infty,\overline{\Lambda}^\infty,K_Z^\infty)|<1$} \Big\} \label{ll_4_4},
\end{align}
 \begin{align}
 F^{nfb}(K^\infty,\overline{\Lambda}^\infty, K_{Z}^\infty) \tri & \overline{\Lambda}^\infty -M^{nfb}(K^\infty, \overline{\Lambda}^\infty, K_{Z}^\infty) \Big(\overline{\Lambda}^\infty -c\Big), \label{ll_4_a_5} \\
M^{nfb}(K^\infty, \overline{\Lambda}^\infty, K_{Z}^\infty)  \tri  & \Big( K_{Z}^\infty + \overline{\Lambda}^\infty  K^\infty\Big(\overline{\Lambda}^\infty - c \Big)\Big)\Big(K_{Z}^\infty+ K_{W} + \Big(\overline{\Lambda}^\infty - c \Big)^2 K^\infty\Big)^{-1} \label{ll_4_b_6}
\end{align}
and the set ${\cal P}^{\infty,nfb}(\kappa)$ is nonempty for all  $\kappa\in [0,\infty)$.\\
Moreover, there exist  maximum element $(\overline{\Lambda}^\infty, K_Z^\infty) \in {\cal P}^{\infty,nfb}(\kappa)$, such that for all $\kappa \in [0,\infty)$, \\
\CDC{
 (i) if the noise is stable, i.e., $c\in (-1,1)$ then the input and the output processes $(X_t^o, Y_t^o), t=1, \ldots$ are asymptotically stationary,  and \\
 (ii) if the noise is unstable i.e., $c\notin (-1,1)$ then  input and the innovations  processes $(X_t^o, I_t^o), t=1, \ldots$ are asymptotically stationary. }
%
%
%
%
\end{corollary} 
\begin{proof} First, note that definition (\ref{int_nfb_1})-(\ref{adm_set_2_2}), follows from the restriction of time-varying strategies of Corollary~\ref{thm_FTWFI} to time-invariant, and by imposing the detectability and stabilizability conditions of the DRE (\ref{reg_23_in_a}) and convergence of the Lyapunov difference equation (\ref{reg_23_in}).  The lower bound  (\ref{ll_3_3}),  follows from $|\overline{\Lambda}^\infty|<1$, which implies   $\lim_{n \longrightarrow \infty} K_{X_n}^o=K_{X}^\infty$ satisfies $K_X^\infty =\big(\overline{\Lambda}^\infty\big)^2 K_X^\infty +K_Z^\infty$, and the convergence of the DRE (\ref{reg_23_in_a}) to the ARE, as stated in the set ${\cal P}^{\infty,nfb}(\kappa)$.  The fact that the set ${\cal P}^{\infty,nfb}(\kappa)$ is nonempty for all  $\kappa\in [0,\infty)$, follows from Theorem~\ref{thm_nfb}, since this is true for the special case   $\overline{\Lambda}^\infty=0$. The last statement follows, from the definition of the set ${\cal P}^{\infty,nfb}(\kappa)$, and the fact that for the special case of Theorem~\ref{thm_nfb}, such a channel  input always exists. 
\end{proof}

\ \

\begin{remark} On the lower bound on achievable rates without feedback\\
(1) We should note that the achievable lower bound on non-feedback capacity, given in   Corollary~\ref{thm_nfb_g}, i.e., $C_{LB}^\infty(\kappa,v_0)=C_{LB}^\infty(\kappa)$ defined by (\ref{ll_3_3}), on the non-feedback capacity $C^{nfb}(\kappa,v_0)$, defined by  (\ref{inter_nfb}),   holds for stable and unstable noise AR$(c), c \in (-\infty,\infty)$.  This is contrary to the well-known   water-filling solution, given by  (\ref{nfb_wf}), which presupposes the noise is stable, i.e., AR$(c), c \in (-1,1)$. \\
(2) We expect that the closed form expression of the lower bound $C_{LB}^\infty(\kappa,v_0)=C_{LB}^\infty(\kappa)$ defined by (\ref{ll_3_3}), can be found by carrying out the optimization problem, as in Theorem~\ref{thm_sol}. \\
(3) By Figure~\ref{fig:compnofeed}, the difference between the non-feedback capacity $C^{nfb}(\kappa)$  based on water-filling formulae (\ref{nfb_wf}) (see   \cite[eqn(5.5.14)]{ihara1993} and \cite[eqn(6)]{butman1976}),   and the lower bound $C_{LB}^{\infty,nfb}(\kappa)$ on  achievable rate without feedback based on the formulae \eqref{int_nf_8_aa},  of transmitting an IID channel input $Z_t^o\in N(0,\kappa)$,   for an AGN channel driven by AR$(c)$, noise,  $c=0.75$ and $K_W=1$. is less than $1.5 \times 10^{-2}$ bits per channel use. This difference is expected to be reduced further, if the lower bound $C_{LB}^\infty(\kappa,v_0)=C_{LB}^\infty(\kappa)$ defined by (\ref{ll_3_3}) is  used, because it employs  a Markov channel input, instead of an IID channel input.  
\end{remark}


\section{Discussion of Related  Literature and Comparison with  Main Results of the Paper}
\label{lite}
In this section we recall the formulation and some of the results of   \cite{cover-pombra1989} and compare them with  \cite{yang-kavcic-tatikonda2007,kim2010} and with recent results in \cite{liu-han2017,liu-han2019,gattami2019,li-elia2019}, with emphasis on the assumptions based on which  these are derived. Then we  specialize some of the results to AGN channels driven by the AR$(c)$ stable noise, to  point out  some of the oversights, which are overlooked in \cite{kim2010}, and are repeated in \cite{liu-han2017,liu-han2019,gattami2019,li-elia2019}.

\subsection{Cover and Pombra Characterization of Feedback Capacity} 
\label{sect:cp}
 Cover and Pombra \cite{cover-pombra1989} considered the AGN channel (\ref{AGN_in}), when the noise is time-varying, i.e., nonstationary Gaussian, with distribution ${\bf P}_{V_t|V^{t-1}}, t=2, \ldots, n, {\bf P}_{V_1|V_0}={\bf P}_{V_1}$,  and  derived converse and direct coding theorems,    
through the  characterization of the $n-$finite (block length) transmission of  feedback information ($n-$FTFI) capacity. The code   in 
\cite{cover-pombra1989}, is analogous to Definition~\ref{def_code}, with the fundamental difference that  the encoder  and decoder are replaced by  $X_1=e_1(W),X_2=e_2(W,X_1,Y_1),\ldots, X_n=e_n(W, X^{n-1}, Y^{n-1})$,  $y^n \longmapsto d_{n}(y^n)\in  {\cal M}^{(n)}$. The  average  error  probability is defined by (\ref{g_cp_4}), without knowledge of the initial state $v_0$, and    depends on the initial distribution ${\bf P}_{Y_1}$, induced by  $Y_1=X_1+V_1$, where  ${\bf P}_{V_1}$ is fixed.  \\
Cover and Pombra derived the characterization of $n-$FTFI capacity,  by recognizing that  entropy $H(Y^n)$ is  maximized, if the input $X^n$ is  jointly Gaussian, driven by a jointly Gaussian process $\overline{Z}^n$,  of the form:
\begin{align}
&X_t=\sum_{j=1}^{t-1} \Gamma_{t,j}V_j +\overline{Z}_t, \hso X_1=\overline{Z}_1, \hso t=2, \ldots, n, \label{cp_6}   \\
&\mbox{equivalently} \hso  
X^n =\Gamma^n V^n + \overline{Z}^n, \hso  \hso \mbox{$\Gamma^n$ is  lower diagonal non-random matrix}, \label{cp_10}\\
& \overline{Z}^n \hso \mbox{is jointly Gaussian, $N(0, K_{\overline{Z}^n})$}, \\
&  \overline{Z}^n \hso \mbox{is independent of} \hso V^n, \label{cp_8} \\
&\frac{1}{n}  {\bf E} \Big\{\sum_{t=1}^{n} (X_t)^2\Big\} =\frac{1}{n}Trace \; \left\{ {\bf E}\Big(X^n (X^{n})^T\Big) \right\} \leq \kappa. \label{cp_9} 
\end{align} 
The  $n$-FTFI  capacity is defined by maximizing $H(Y^n)-H(V^n)$ over all {\it time-varying} channel input strategies  $(\Gamma^n,K_{\overline{Z}^n})$, which induce {\it time-varying} feedback distributions,  as follows. 
\begin{align}
 C_n(\kappa) 
  \tri & \max_{  \big(\Gamma^n, K_{\overline{Z}^n}\big): \hso \frac{1}{n} Trace  \big\{ {\bf E} \big(X^n (X^n)^T\big)\big\}\leq \kappa}  H(Y^n) - H(V^n )  \label{cp_11}\\
=&  \max_{  \big(\Gamma^n, K_{\overline{Z}^n}\big):\hso \frac{1}{n} Trace\big\{\Gamma^n \; K_{V^n} \; (\Gamma^n)^T+ K_{\overline{Z}^n}\big\}  \leq \kappa } \frac{1}{2} \log \frac{|\big(\Gamma^n+I\big) K_{V^n}\big(\Gamma^{n}+I\big)^T+ K_{\overline{Z}^n}|}{|K_{V^n}|}. \label{cp_12}
  \end{align}
Coding theorems, are   stated in  \cite[Theorem~1]{cover-pombra1989}, based on $\frac{1}{n} C_n(\kappa)$ for sufficiently large ``$n$''.\\
The $n-$finite transmission without feedback information (FTwFI) capacity is characterized by the input process (\ref{cp_6}), with $\Gamma_{t,j}=0, \forall (t,j)$, i.e.,  $X_t=\overline{Z}_t,  t=1, \ldots n,$ or ${\Gamma}^n=0$, which imply  $K_{X^n}=K_{\overline{Z}^n}$, as follows  \cite[eqn(14)]{cover-pombra1989}.
\begin{align}
 C_n^{nfb}(\kappa) 
  \tri 
  \max_{ K_{X^n}:\hso \frac{1}{n} Trace\big( K_{X^n}\big)  \leq \kappa } \frac{1}{2} \log \frac{|K_{X^n}+ K_{V^n}|}{|K_{V^n}|}. \label{cp_nfb}
  \end{align}
 In \cite[Theorem~2,  Eqn(15), Eqn(16)]{cover-pombra1989},  it is proved that $C_n(\kappa)$ and  $C_n^{nfb}(\kappa)$,  satisfy
\begin{align}
\frac{1}{n} C_n^{nfb}(\kappa) \leq\frac{1}{n} C_n^{fb}(\kappa)  \leq \frac{1}{n} C_n^{nfb}(\kappa) +\frac{1}{2}.  \label{bounds}
\end{align} 
To this date  no closed formulas are available in the literature for  $C_n(\kappa)$ and  $C_n^{nfb}(\kappa)$ for unstable or unstable noise.   
 Bounds (\ref{bounds}) are  analyzed extensively in the literature \cite{dembo1989,ihara-yanaki1989,yanaki1992,yanaki1994,chen-yanaki1999,ozarow1990,ozarow1990a} (and references therein), 
 under stationarity or asymptotic stationarity. 

\CDC{
\subsection{Observations on the Feedback Codes and Optimal Channel Inpus  of Past Literature
    }
\label{app_sect_imp}
In this section, we make some observations on the  optimal channel inputs, that are used in \cite{yang-kavcic-tatikonda2007}, and applied in \cite{kim2010} to derive the main theorem \cite[Theorem~6.1]{kim2010}. From these observations then   follows that neither  \cite[Theorem~6.1]{kim2010} (and also \cite[Theorem~4.1]{kim2010}) nor \cite{liu-han2017,liu-han2019,gattami2019,li-elia2019},  characterize the feedback capacity for the Cover and Pombra \cite{cover-pombra1989} formulation of feedback code (with limit and supremum operations interchanged) of stationary channels. }


\begin{observation}\CDC{  Optimal channel inputs in \cite{yang-kavcic-tatikonda2007,kim2010,liu-han2017,liu-han2019,gattami2019,li-elia2019}\\
\label{obs_in_new}
(1) In \cite{yang-kavcic-tatikonda2007} the information structure of optimal channel inputs with feedback are derived 
for AGN channels driven by stationary autoregressive noise, described by a power spectral density (PSD) function $S_V(e^{j\omega}), \omega \in [-\pi,\pi]$ \cite[eqn(7)]{yang-kavcic-tatikonda2007}:
\begin{align}
S_V(e^{j\omega}) \tri& K_W\frac{\Big( 1-\sum_{k=1}^L a(k) e^{jk \omega}\Big)\Big(1- \sum_{k=1}^L a(k) e^{-j lk \omega}\Big)}{\Big( 1-\sum_{k=1}^L c(k) e^{jk \omega}\Big)\Big(1- \sum_{k=1}^L c(k) e^{-jk\omega}\Big)}, 
\hso |c(k)|<1,\hso |a(k)|\leq 1. \label{PSD_G}
\end{align}
In  \cite[Section~II.C, eqn(19), eqn(20)]{yang-kavcic-tatikonda2007},  the noise is represented by a particular state space realization, with state variables $S^n\tri (S_1, \ldots, S_n)$,   under the following {\it crucial assumption}\footnote{If (A1)  does not hold then the characterization of the $n-$FTFI  capacity \cite[Sections~II-V]{yang-kavcic-tatikonda2007} is not valid.}. }

\CDC{ (A1) \cite[page 933, I)-III)]{yang-kavcic-tatikonda2007}:   given the initial state of the noise $S_0=s_0$, which is known to the encoder and the decoder, the channel input $X^n\tri (X_1, \ldots, X_n)$ uniquely defines the state variables $S^n$ and vice-versa.}

\CDC{Then, under Assumotion (A1),  in \cite[Section~II-V]{yang-kavcic-tatikonda2007},  the optimal channel input $X^n$ is expressed in terms of the state $S^n$, its initial state $S_0=s_0$,  the outputs $Y^n$,  and a Gaussian innovations process (see \cite[Theorem~5, eqn(114)]{yang-kavcic-tatikonda2007}). The $n-$FTFI capacity is characterized by \cite[eqn(131)]{yang-kavcic-tatikonda2007}. \\
(2) In \cite[Theorem~6.1]{kim2010} the characterization of feedback capacity is derived by invoking an equivalent state space representation of the noise (\ref{PSD_G}), given by \cite[eqn(58)]{kim2010}, i.e.,  
\begin{align}
S_{t+1}=&FS_t +G W_t, \hso S_0=0, W_0=0,\hso  t=0, \ldots,  \label{ss_1}\\
 V_t=& HS_t+ G W_t,\hso  t=1, \ldots, \label{ss_2}
\end{align}
 where $W_t, t=1, \ldots, $ are  independent, Gaussian, independent of the Gaussian RV $S_0$, and  $F, G, H$ are matrices of dimensions $L\times L, L\times 1, 1\times L$, respectively.    $S_t$ is called the state of the noise.\\ 
  Although, not explicitly stated in \cite{kim2010}, from  \cite[Lemma~6.1]{kim2010} (which is used to derive  \cite[Theorem~6.1]{kim2010}), it follows  a crucial assumption, which is analogous to  (A1) (see also \cite[page 78, left column, first paragraph]{kim2010}):}
  
\CDC{   
(A2) \cite[Theorem~6.1 and Lemma~6.1]{kim2010}:  Given the initial state of the noise $S_0=s_0, W_0=w_0$, which is  known to the encoder and the decoder, the noise $V^n$ uniquely defines the state variables $S^n$ and vice-versa.}

\CDC{ 
(3) In Cover and Pombra \cite{cover-pombra1989} the code definition  and  optimal channel input process $X^n$, i.e.,  (\ref{cp_6}),   are not allowed to depend on the state variables $S^n$ and the initial states, i.e., assumptions (A1) or (A2) are not imposed, and hence $S^n$ needs to be estimated at the encoder.   Because  assumptions (A1) and (A2) are imposed in  the feedback capacity problem analyzed in \cite{yang-kavcic-tatikonda2007} and characterization of feedback capacity in \cite[Theorem~6.1]{kim2010} (see also \cite[Lemma~6.1]{kim2010}), respectively, then \cite{yang-kavcic-tatikonda2007} and  \cite[Theorem~6.1]{kim2010} treat a fundamentally different and much easier problem than  the Cover and Pombra \cite{cover-pombra1989}.   \\
Papers \cite{liu-han2017,liu-han2019,gattami2019,li-elia2019} should be read with caution, because the misinterpretations discussed above are repeated.\\
For example,  \cite{liu-han2017,liu-han2019,li-elia2019} developed their results based on  \cite[Theorem~4.1]{kim2010} which is equivalent to \cite[Theorem~6.1]{kim2010}, hence the discussion under (2) applies to  these references.\\
In  \cite{gattami2019} the related to  the presentation of \cite[Theorem~6.1]{kim2010}, such as,  in \cite[page 1949 second column]{gattami2019}, and \cite[Section~IV first paragraph]{gattami2019}, is misleading  because the discussion under (2) applies.  In  \cite[Theorem~2]{gattami2019}, the optimal channel input is a function of the state process $S^n$, and hence the discussion under (2) applies as well. 
}
\end{observation}

\CDC{
\subsection{Preliminary Observations on Convergence of Feedback Rates}
\label{app_comp}
In this section we make preliminary observation on conditions for  existence of the limits of the  feedback capacity definition (\ref{inter}), which are not accounted for in  \cite[Theorem~6.1]{kim2010} (equivalently in \cite[Theorem~4.1]{kim2010}). These are also overlooked in  \cite{liu-han2017,liu-han2019,gattami2019,li-elia2019}, when use is made of the results of \cite{kim2010}.  For the stationary AR$(c)$ noise,  the feedback capacity characterizations given in   \cite[Theorem~6.1]{kim2010},   and similarly in    \cite[Theorem~4.1]{kim2010}, correspond to  $K_Z^\infty=0$, (see\cite[Lemma~6.1]{kim2010} and  comments above it), which means the stabilizability condition (see Definition~\ref{def:det-stab}) is not satisfied. }

\ \

\begin{observation}
On the  feedback capacity characterizations given in \cite{kim2010}\\
\label{obs_in}
(1) Observe that the characterization of feedback capacity given in \cite[Theorem~6.1]{kim2010} for the  state space representation of the  noise  (\ref{ss_1}), (\ref{ss_2}),  with zero variance of the innovations  of the channel input process, implies one solution to the symmetric  matrix algebraic Riccati equation (ARE) \cite[Eqn(61) satisfied by  $\Sigma$]{kim2010}, is always the zero  matrix solution $\Sigma=0$. \\
 An important technical issue,  which is overlooked in the treatment of the asymptotic analysis in  \cite{kim2010}, is  the fact that, the    estimation error covariance\footnote{$\Sigma$ in the notation of \cite{kim2010}.} of the noise $V^n$  satisfies a generalized difference Riccati  equation\footnote{Generalized DREs and algebraic Riccati equations (ARE) are well-known in filtering  problems of Gaussian systems, where the state to be estimated is driven by a Gaussian noise,  which is correlated with the Gaussian noise that drives the observed process, i.e., the process that is used to estimate the state process \cite{kailath-sayed-hassibi,caines1988}.}  (DRE) \cite{kailath-sayed-hassibi}, and not a classical DRE (see Section~\ref{sect:q1}).  It then follows from the theory of generalized DREs and AREs \cite{kailath-sayed-hassibi,caines1988}, (see properties of Theorem~\ref{thm_ric}),  that even for the stable and stationary channel noise considered  in \cite{kim2010}, to ensure asymptotic stationarity   of the optimal  joint channel input and output process, it is necessary that  the  conditions, known as detectability and stabilizability conditions should be  appended  to the definition of characterization of 
 feedback capacity\footnote{It appears \cite{kim2010} did not recognize that stability or stationarity of the noise is not sufficient for detectability and stabilizability to hold, as in  classical differerence and algebraic equations. This is also not accounted for  in \cite{gattami2019}.}   in \cite[Theorem~6.1]{kim2010}.  When  these conditions are appended, then the  feedback capacity is fundamentally different from the one given in   \cite[Theorem~6.1]{kim2010}.\\
Two special cases are discussed below. \\
Case 1. Consider the  asymptotically stationary AR$(c)$ noise. From the  generalized ARE (\ref{ARE_in}) of $K^\infty$, with a zero variance of  the innovations of the channel input,  $K_Z^\infty=0$,  then it follows that   $(K^\infty, \Lambda^\infty)$ is the analog of $(\Sigma,X)$ in \cite[Theorem~6.1]{kim2010}.   The equation of $K^\infty$, with $K_Z^\infty=0$,  is  a quadratic polynomial, with  two roots, one at $K^\infty=0$,  and another at $K^\infty=K^\infty(\Lambda^\infty)$, given in (\ref{gre_1_a})-(\ref{gre_2}) (by simple algebra).  \\ 
It then follows from the general properties of generalized matrix difference RIccati equations \cite{kailath-sayed-hassibi} (see Theorem~\ref{thm_ric}, Lemma~\ref{lem_pr_are}  for its application to the AR$(c)$ noise), that the   asymptotic limit of the mean-square estimation error converges to     the zero solution $\Sigma\equiv K^\infty=0$,  and not the maximal solution.   This then implies, the feedback capacity based on \cite[Theorem~6.1, $C_{FB}$]{kim2010} is zero.\\
Case 2. Consider the PSD (\ref{PSD_G}), with $L=1$, denoted by AR$(a,c)$, given in \cite[eqn(43)]{kim2010}:  
\begin{align}
S_V(e^{j\omega}) \tri& K_W\frac{\Big( 1-a e^{i\omega}\Big)\Big(1-a e^{-i\omega}\Big)}{\Big( 1- c e^{i \omega}\Big)\Big(1- c e^{-i\omega}\Big)}, 
\hso |c|<1,\hso |a|\leq 1,\hso  K_W=1. \label{PSD_G_S}
\end{align}
 In \cite[above and below eqn(43)]{kim2010}, the following state space representation is introduced.
\begin{align}
S_{t+1}=&c S_t +W_t, \hso S_0=W_0=0, \hso t=0, \ldots, \label{ar_g_1}   \\
V_t=&\big(c-a\big)S_t +W_t,\hso t=1, \ldots, \label{ar_g_2}
\end{align}
Similarly to Case 1,  the ARE   \cite[Theorem~6.1, eqn(61) satisfied by  $\Sigma$]{kim2010}, is a quadratic polynomial in $\Sigma$, with two solutions, the zero solution $\Sigma=0$, and a nonzero solution $\Sigma\neq 0$ (one root is always zero  because $K_Z^\infty=0$).  It is easy to  verify that the choice of  the nonzero solution gives directly the  statement of  \cite[Theorem~5.3]{kim2010} (which is derived using a lengthy derivation  from  the frequency domain characterization of \cite[Theorem~4.1]{kim2010}).  \\
However, similarly to (1) it follows from the general properties of generalized matrix difference Riccati equations \cite{kailath-sayed-hassibi}, that the   asymptotic limit of the mean-square estimation error converges to     the zero solution $\Sigma\equiv K^\infty=0$,  and not the nonzero solution, which then implies the feedback capacity based on \cite[Theorem~6.1, $C_{FB}$]{kim2010} is zero.\\
 The above   technical oversights are also discussed in Remark~\ref{rem_g-kalman_1}, and  follow from Section~\ref{sect:ar1_r}, {\bf Facts 1-3}.  
We also give an alternative derivation in  Section~\ref{sect:kim} (see Counterexample~\ref{counterex}).\\
\CDC{Papers \cite{liu-han2017,liu-han2019,gattami2019,li-elia2019} should be read with caution, because the technical issues discussed above  are repeated. }
\end{observation}

\subsection{The Kim Characterization of Feedback Capacity}
\label{sect:kim}
  Let us now recall the characterization of feedback capacity given  
in  \cite[Theorem~6.1]{kim2010} for the stable AR$(c)$ noise, with  $V_0=v_0$, known to the encoder and the decoder, which presupposes the variance of the innovations of the channel input process is zero, i.e.,  $K_Z^\infty =0$ in the generalized ARE of Theorem~\ref{lem_cov}.

{\it The Feedback Capacity for the stable  AR$(c), c\in (-1,1)$ Noise \cite{kim2010}.} From   \cite[Theorem~6.1]{kim2010}, under the  stationarity or asymptotic stationarilty of the  input and  output process, the following is obtained.

{\bf (K1) } 
 The characterization of feedback capacity for the stable AR$(c), c\in (-1,1)$ noise,   is given by\footnote{The reader may  verify that characterization  (\ref{i_kim_1a})-(\ref{i_ric_3_na}), with $K_W=1$,  is a degenerate version of \cite[Theorem~6]{kim2010}.
 } \cite[Theorem~6.1, eqn(61)]{kim2010}
 \begin{align}
&C^{K}(\kappa)  \tri   \max_{\Lambda^\infty}  \log\Big( \frac{\big(\Lambda^\infty+c\big)^2 K^\infty   +K_{W}}{K_{W}}\Big) \label{i_kim_1a}\\
&\mbox{  $\Lambda^\infty \in {\mathbb R}\hso $ such that  $\hso |\Lambda^\infty|\neq 1$,}\label{i_kim_1} \\
&\mbox{$c^2 K^\infty \leq \kappa$}, \label{k_pc}  \\ 
&\mbox{$K^\infty$ is the maximal solution  of the ARE,} \\
& K^\infty= c^2 K^\infty + K_W -\frac{ \Big( K_W + c  K^\infty\big(\Lambda^\infty + c \big)\Big)^2}{ \Big( K_W + \big(\Lambda^\infty + c \big)^2 K^\infty\Big)}, \hst |c|<1.  \label{i_ric_3_na}   
\end{align}
Note that (\ref{i_ric_3_na}) corresponds to the generalized ARE of Theorem~\ref{lem_cov} with $K_Z^\infty =0$. 
Condition
 (\ref{i_kim_1}) is known as unit circle controllability \cite{kailath-sayed-hassibi}, (see Definition~\ref{def:det-stab}.(b)). {\bf Fact 3}, (\ref{gre_1_a})-(\ref{gre_2}) give all possible solutions of (\ref{i_ric_3_na}).  By solving (\ref{i_ric_3_na}) we obtain two  solutions, $K^\infty=0$ and   $K^\infty=\frac{K_W\big((\Lambda^\infty)^2-1\big)}{\big(\Lambda^\infty+ c\big)^2}, |\Lambda^\infty| \neq  1$. Note,  that  by Lemma~\ref{lem_pr_are}.(3),   the stabilizability condition is satisfied if and only if $|\Lambda^\infty|<1$, and hence stabilizability implies the only nonnegative solution is  $K^\infty=0$. Hence,  the solution  that  correspond to the per unit time limit (\ref{i_ll_2})     is $K^\infty=0$, i.e., $|\Lambda^\infty|<1$, which then implies $C^K(\kappa)=0, \forall \kappa \in [0,\infty)$.
 
{\bf (K2)} In  \cite[page 58, second column, last paragraph]{kim2010},  the solution of  the optimization Problem (\ref{i_kim_1a})-(\ref{i_ric_3_na}),   for the AR$(c)$ noise, described by the  power  spectral density,  $S_V(e^{j\omega})=\frac{1}{|1+c e^{j\omega}|^2}, -1 < c < 1$, i.e., for the an AR$(c)$ stationary noise with $K_W=1$, is given as follows. 
\begin{align}
&C^{K}(\kappa)=-\log x_0, \label{kim_ar1_a}  \\
&\mbox{where $x_0$ is the unique positive root of $\kappa x^2=\frac{\big(1-x^2\big)}{\big(1+|c|x\big)^2}$}.    \label{kim_ar1_b}
\end{align}  
 The (incorrect) simplified characterizations of the feedback capacity  i.e., with zero power spectral density and zero variance of the innovations process of the channel input,   of    \cite[Theorem~4.1]{kim2010} and  \cite[Theorem~6.1]{kim2010}, respectively, are utilized extensively in the literature, and    
 stimulated the interest into  additional investigations, by many authors, for example, \cite{liu-han2017,liu-han2019,gattami2019,li-elia2019}.

The next  proposition summarizes  the implications of the properties of generalized DREs and AREs of Theorem~\ref{thm_ric}, and  Lemma~\ref{lem_pr_are}, on the characterization of  feedback capacity given in  \cite[Theorem~4.1 and Theorem~6.1]{kim2010}.  \\

\begin{proposition} On  \cite[Theorem~6.1 and Lemma~6.1]{kim2010} for stable  AR$(c)$ noise  \\
\label{prop_kim}
Consider  the  characterization  of \cite[Theorem~6]{kim2010} for  stable  AR$(c)$ noise,     which   is equivalent to the optimization problem (\ref{i_kim_1a})-(\ref{i_ric_3_na}). \\
(a) There are two  solutions to the optimization problem (\ref{i_kim_1a})-(\ref{i_ric_3_na}), listed under Solutions \# 1,  \#2. \\
Solution \#1: 
\begin{align}
&K^\infty  =\frac{K_W\Big(\big(\Lambda^\infty\big)^2-1\Big)}{\Big(\Lambda^\infty+ c\Big)^2},  \hso   |\Lambda^\infty| >1, \hso \Lambda^\infty \neq -c, \label{i_soll_5} \\
& \big(\Lambda^\infty\big)^2 K^\infty=\kappa \hso \Longrightarrow \hso K_W \big(\Lambda^\infty\big)^4 -\big(K_W+\kappa\big) \big(\Lambda^\infty\big)^2 -2 c\kappa \Lambda^\infty -c^2 \kappa=0,\hso |\Lambda^\infty|> 1 \label{i_soll_7}\\
& C(\kappa)= \log |\Lambda^\infty|, \label{i_soll_8}\\
&\mbox{the pair $\{A, C\}$ is detectable},
 \label{soll_9}\\
&\mbox{the pair $\{A^*, B^{*,\frac{1}{2}}\}$ is unit circle controllable, which is equivalent to (\ref{i_kim_1}),} \\
&\mbox{the pair $\{A^*, B^{*,\frac{1}{2}}\}$ is  not stabilizable.} \label{soll_10}
\end{align}
In particular,  for $K_W=1$ then $x=\frac{1}{\Lambda^\infty}$ satisfies (\ref{kim_ar1_b}), and   (\ref{kim_ar1_a}) and  (\ref{kim_ar1_b}) are obtained  from  (\ref{i_soll_5})-(\ref{i_soll_8}), i.e., letting $K_W=1$.

Solution \#2:
\begin{align}
&K^\infty = 0,  \hso  |\Lambda^\infty|<1,  \label{i_soll_2}\\
& \big(\Lambda^\infty\big)^2 K^\infty=0 \leq \kappa, \hso C^{K}(\kappa)=0, \hso \forall \kappa \in [0,\infty),\label{i_soll_3} \\
&\mbox{the pair $\{A, C\}$ is detectable}, \\
&\mbox{the pair $\{A^*, B^{*,\frac{1}{2}}\}$ is stabilizable.}\label{i_soll_4}
\end{align}
(b) The asymptotic per unit time limit, as defined by (\ref{Q_1_2_s1_new_b}), converges to the characterization  (\ref{i_kim_1a})-(\ref{i_ric_3_na}) if and only if $|\Lambda^\infty|<1$,  and  Solution \#2 is the unique limit.   
\end{proposition}

\begin{proof} The statements follow from Theorem~\ref{thm_ric} and Lemma~\ref{lem_pr_are}.
\end{proof}

To gain additional insight,  we construct a  counterexample to  \cite[Lemma~6.1]{kim2010}, which states  that the variance of the innovations of the channel input process is optimal, i.e.,  $K_Z^\infty =0$, that  let to the incorrect characterization of feedback capacity  \cite[Theorem~6.1]{kim2010} (and also to \cite[Theorem~4.1]{kim2010}),  by calculating the transition map of   the error recursion of the generalized Kalman-filter. An alternative approach  to construct a counterexample, is to follow the statements prior to \cite[Lemma~14.2.1, page 507]{kailath-sayed-hassibi}, for the generalized DRE and ARE, using its  transform version, as discussed in \cite[Section~14.7, page 540]{kailath-sayed-hassibi}.  \\

\begin{counterexample} On \cite[Theorem~6.1 and Lemma~6.1]{kim2010}\\
\label{counterex}
Our objective is to evaluate $C^K(\kappa)$, defined by  (\ref{i_kim_1a})-(\ref{i_ric_3_na}), from (\ref{Q_1_2_s1_new_b}),  by first calculating the transition map of the error recursion of the generalized Kalman-filter.\\
Let $K_0^o, K_1^o, \ldots, K_n^o$ denote the sequence generated from the generalized DRE (\ref{Q_1_10_s1_new}), with $K_Z^\infty=0$, as stated in  \cite[Theorem~6]{kim2010}, i.e., 
\begin{align}
K_{t}^o= & c^2 K_{t-1}^o  + K_{W} -\frac{ \Big( K_{W} + c K_{t-1}^o\Big(\Lambda^\infty + c \Big)\Big)^2}{ \Big( K_{W} + \Big(\Lambda^\infty + c \Big)^2 K_{t-1}^o\Big)}, \hst  K_t^o \geq0,  \hso K_{0}^o=0, \hst t=1, \ldots, n, 
 \label{Q_1_10_s1_new_new}
 \end{align}
We shall compute the solution of (\ref{Q_1_10_s1_new_new}) recursively, to determine whether the transition map of the error recursion (\ref{i_error_s1_ti}) converges, i.e.,  whether it is stabilizing, by checking:
\begin{align}
\lim_{n \longrightarrow \infty} |F(K_{t-1}^o,\Lambda^\infty, K_{Z}^\infty=0)| =\lim_{n \longrightarrow \infty} |c - M(K_{t-1}^o, \Lambda^\infty, K_Z^\infty=0)\Big(\Lambda^\infty +c\Big)| < 1.
\end{align}
(i) By the generalized DRE (\ref{Q_1_10_s1_new_new}), 
at time $t=1$,    we have 
\begin{align}
&K_1^o=0 \hso \Longrightarrow \hso F(K_{0}^o,\Lambda^\infty, K_Z^\infty=0)= c -\Lambda^\infty -c=-\Lambda^\infty.  \label{rec_1_ce}
\end{align}
At time $t=2$,  we have 
  \begin{align}
&K_2^o=0  \hso \Longrightarrow \hso    F(K_{1},\Lambda^\infty, K_Z^\infty=0)=-\Lambda^\infty. \label{rec_2_ce}
\end{align}
For all $t\in\{3,4, \ldots, n\}$, by recursive calculations,   we have  
  \begin{align}
&K_t^o=0  \hso \Longrightarrow \hso  F(K_{t-1}^o,\Lambda^\infty, K_Z^\infty=0)    =-\Lambda^\infty, \hst t=2,3, \ldots, n. \label{rec_3_ce}
\end{align}
Clearly, a necessary condition for the transition map of the error recursion  to  convergence  is 
\bea
|  F(K_{n-1}^o,\Lambda^\infty, K_Z^\infty)|    =|-\Lambda^\infty|< 1, \hst n=0,1, \ldots, .
\eea
Thus, the necessity of $|\Lambda^\infty|<1$ implies Proposition~\ref{prop_kim}, Solution \#2 is the valid solution, re-confirming  our earlier  discussions. Hence, the unique, nonnegative solution of (\ref{Q_1_10_s1_new_new}) is $K_n^o=0, n=0,1, \ldots$. 
It then follows that  $\Lambda^\infty K_n^o=0, n=0,1,  \ldots$. Then   for any finite $n$, by (\ref{Q_1_2_s1_new_b}) with $K_Z^\infty=0$, we have
\begin{align}
&\frac{1}{n} {\bf E}_{v_0}\Big\{\sum_{t=1}^n \big(X_i^o\big)^2\Big\} =\frac{1}{n} \sum_{t=1}^n (\Lambda^\infty)^2 K_{t-1}^o=0, \label{i_limit_1_ce}\\
&  \frac{1}{2n} \sum_{t=1}^n  \log\Big( \frac{\Big(\Lambda^\infty+c\Big)^2 K_{t-1}^o +K_{W}}{K_W}\Big)=0.\label{i_limit_2_ce}
\end{align}   
The last two equations imply that the valid solution of  (\ref{i_kim_1a})-(\ref{i_ric_3_na}), i.e., that corresponds to (\ref{Q_1_2_s1_new_b}), when $K_Z^\infty=0$,   is given by $C^K(\kappa)=0, \forall \kappa \in [0,\infty)$.\\
(ii) It is easy to verify that the calculations in (i) are also predicted from the properties of the generalized DREs and AREs.
\end{counterexample}

\section{Conclusion}
The $n-$finite transmission feedback information  (FTFI) capacity for additive Gaussian noise (AGN) channels with  feedback,   is characterized, and lower bounds on the characterization of the  $n-$finite transmission without feedback information  (FTwFI) capacity are derived,  when the noise is described by {\it stable and unstable} autoregressive models. Closed form feedback capacity formulas are derived, when   {\it channel input strategies or distributions are time-invariant}, for  autoregressive memory one, stable and unstable noises, It is shown that feedback does not increase capacity, when the noise is stable and for certain unstable noise.  Lower bounds on the  non-feedback capacity are also derived, based on Markov channel input distributions, i.e., induced by a Gaussian Markov channel input process,  and also by an independent and identically distributed  channel input process. These  achievable lower bounds on non-feedback capacity hold for any autoregressive noise model, irrespective of whether it is stable or unstable. \\
An interesting    long standing unanswered question, is to derive a closed form expression for the $n-$FTFI capacity, to gain insight into the performance of achievable  time-varying rates.  
 

{\bf Acknowledgments.} This work was co-funded by the European Regional Development Fund and the Republic of Cyprus through the Research and Innovation Foundation, under Projects,  EXCELLENCE/1216/0365 and EXCELLENCE/1216/0296.

\section{Appendix}
\label{appendix}

\subsection{proof of  Lemma~\ref{lem_pr_are}}
\label{lem_pr_are_AP}
Recall (\ref{ma_1}) and (\ref{m_2}). 
 (1) By definition $A=c, C=\Lambda+c$. By $c \in (-1,1)$, then  there always exists a $G \in {\mathbb R}$ such that   $|A-GC|=|c -G(\Lambda +c)|<1$, i.e., take $G=0$. This shows (1).\\
(2) Since $K_{Z}=0$,  then $B= 1-K_W\big(K_Z+K_W\big)^{-1}=0,  B^{*,\frac{1}{2}}=K_W^{\frac{1}{2}}B^{\frac{1}{2}}=0$,   $A^*=c-K_W \big(K_W+K_Z\big)^{-1}C=c-\Lambda-c=-\Lambda$, and hence, there exists a $G \in {\mathbb R}$ such that   $|A^*-B^{*,\frac{1}{2}}G|=|-\Lambda|\neq 1$ if and only if $|\Lambda|\neq 1$. This shows (2).\\
(3) Since $K_{Z}=0$, similar to the prove  in (2), there exists a $G \in {\mathbb R}$ such that   $|A^*-B^{*,\frac{1}{2}}G|=|-\Lambda|< 1$ if and only if $|\Lambda|< 1$. This shows (3). \\
(4) Since $c\in (-1,1), K_Z=0$, then,  by (1),  the pair $\{A, C\}$ is detectable, by  (2) the pair $\{A^*,B^{*,\frac{1}{2}}\}$ is unit circle controllable if and only if $|\Lambda| \neq 1$, and by (3) the pair $\{A^*,B^{*,\frac{1}{2}}\}$ is stabilizable  if and only if $|\Lambda|< 1$.  By  Theorem~\ref{thm_ric}.(1) we deduce the claim.  This shows (4). \\
(5) Clearly,  (\ref{i_ric_3_n}) is equivalent to the quadratic equation
\bea
K^2 \big(\Lambda +c\big)^2 + K \big(K_W -K_W \Lambda^2\big)=0.   \label{i_ric_4_nn}
\eea 
Hence, the two solutions are (\ref{i_rae_nu}). The last statement is also obtained by applying Theorem~\ref{thm_ric}.(2), as follows. By (1),  $\{A,C\}$ is detectable, by (2), $\{A^*, B^{*,\frac{1}{2}}\}$  is unit circle controllable if and only if $|\Lambda|\neq 1$, and by (3), $\{A^*, B^{*,\frac{1}{2}}\}$ is stabilizable if and only if $|\Lambda|<1$. By invoking  Theorem~\ref{thm_ric}.(2), then the last statement follows. On the other hand, it is easily verified from  (\ref{i_rae_nu}), that  the uniqueness of solutions $K\geq 0$ holds if and only if $|\Lambda|<1$, because for $|\Lambda|>1$ there are two non-negative solutions. By (\ref{ma_3a}) and (\ref{ma_3}), evaluated at $K_Z=0, K=0$, we have   $F(K,\Lambda,K_Z)\Big|_{K_Z=0,K=0} =-\Lambda$. Hence, the non-negative solution   $K=0$ is unique and stabilizable if and only if $|\Lambda|<1$.

\subsection{Proof of Lemma~\ref{lemma_nc}} 
\label{pr_lemma_nc}
(i) The conditions are a consequence of the optimization problem, while (\ref{nc_9}) is due to the following standard relaxation. For any $(\Lambda^{\infty}, K_Z^{\infty})\in {\cal P}^{\infty}(\kappa)$, then detectability and stabilizability hold, and by properties of time-invariant generalized DRE \cite{caines1988,kailath-sayed-hassibi},   the corresponding sequence $K_n^{o}, n=1,2, \ldots, K_0^{o}=0$ generated by the generalized DRE (\ref{Q_1_10_s1_new}) is non-decreasing, i.e., $K_n^o \leq K_m^o, m \geq n$ and also bounded. Hence, $K_n^o \leq K_{n+1}^o=c^2 K_n^o + K_W -\frac{ \Big( K_W + c  K_n^o\big(\Lambda^\infty + c \big)\Big)^2}{ \Big(K_Z^\infty+ K_W + \big(\Lambda^\infty + c \big)^2 K_n^o\Big)}$, and by taking the limit, as $n \longrightarrow \infty$, the inequality (\ref{nc_9}) is obtained. \\
%
(ii) Suppose $K_Z^{\infty,*}=0$. By (\ref{nc_9}), with $K_Z^{\infty,*}=0$,   then  
\begin{align}
&\Big(K^{\infty,*}-c^2 K^{\infty,*} - K_W\Big)\Big( K_W + \big(\Lambda^{\infty,*} + c \big)^2 K^{\infty,*}\Big) +\Big( K_W + c  K^{\infty,*}\big(\Lambda^{\infty,*} + c \big)\Big)^2\leq0,   \label{ineqprimf}\\
& \Longrightarrow \hso
 K^{\infty,*}\left( K^{\infty,*}\left(\Lambda^{\infty,*} + c \right)^2+K_W\left(1-\big(\Lambda^{\infty,*}\big)^2 \right)\right)\leq 0. \label{ineqprimf_a}
 \end{align}
 By Lemma~\ref{lem_pr_are},   a necessary and sufficient  condition for stabilizability of the pair  $\{A^*, B^{*,\frac{1}{2}}\}$, when $K_Z^\infty=0$,  is $|\Lambda^{\infty,*}|< 1$, and therefore (\ref{ineqprimf_a}) is satisfied if and only if  $K^{\infty,*}$ lies in the region
\bea
\frac{K_W\Big(\big(\Lambda^{\infty,*}\big)^2-1\Big)}{\Big(\Lambda^{\infty,*}+ c\Big)^2} \leq K^{\infty,*}\leq 0, \hso  \mbox{for}\hso  |\Lambda^{\infty,*}|< 1.
\eea
Since it must be that $K^{\infty,*}\geq0$ then necessarily $K^{\infty,*}=0$, which  implies $C^\infty(\kappa)=0, \forall \kappa \in [0,\infty)$. Similarly, if  $K^{\infty,*}=0$ then  necessarily  $K_Z^{\infty,*}=0$, and hence $C^\infty(\kappa)=0, \forall \kappa \in [0,\infty)$. \\
(iii)   By the stationarity conditions   (\ref{nc_3})-(\ref{nc_5}), with $\lambda=(\lambda_1,\lambda_2,\lambda_3, \lambda_4)$:
\begin{align}
&\frac{\partial}{\partial K_Z^\infty}{\cal L}(\Lambda^{\infty,*}, K_Z^{\infty,*},K^{\infty,*},\lambda^*)=  1-\lambda_1^*\Big(K^{\infty,*}-c^2 K^{\infty,*}-K_W\Big)-\lambda_2^*+\lambda_4^*= 0, \label{2.2}\\
&\frac{\partial}{\partial \Lambda^\infty}{\cal L}(\Lambda^{\infty,*}, K_Z^{\infty,*},K^{\infty,*},\lambda^*)=\Big(\Lambda^{\infty,*}+c\Big)K^{\infty,*}-\lambda_1^*\Big\{ \Big(K^{\infty,*}-c^2 K^{\infty,*}-K_W\Big)\Big(\Lambda^{\infty,*}+c\Big)K^{\infty,*}\nonumber \\
&+\Big(K_W+c K^{\infty,*}\big(\Lambda^{\infty,*}+c\big)\Big)c K^{\infty,*}\Big\}-\lambda_2^*\Lambda^{\infty,*} K^{\infty,*}= 0, \label{2.3}\\
&\frac{\partial}{\partial K^\infty}{\cal L}(\Lambda^{\infty}, K_Z^{\infty},K^\infty,\lambda^*)= \Big(\Lambda^{\infty,*}+c\Big)^2-\lambda_1^*\Big\{   \Big(1-c^2\Big)\Big(K_Z^*+K_W+\big(\Lambda^{\infty,*}+c\big)^2 K^{\infty,*}\Big)\nonumber \\
&+ \Big(K^{\infty,*}-c^2 K^{\infty,*}-K_W\Big)\Big(\Lambda^{\infty,*}+c\Big)^2+2c\Big(K_W+c K^{\infty,*}\big(\Lambda^{\infty,*}+c\big)\Big) \Big(\Lambda^{\infty,*}+c\Big)\Big\}\nonumber \\
&-\lambda_2^*\Big(\Lambda^{\infty,*}\Big)^2+\lambda_3^*= 0, \label{2.4} 
\end{align} 
Suppose $\lambda_4^{*}\neq 0$. Then, by  complementary slackness (\ref{nc_6}), we have  $\lambda_4^* K_Z^{\infty,*}=0$, which  implies $K_Z^{\infty,*}=0$, and hence by (ii),  $K^{\infty,*}=0$. By complementary slackness (\ref{nc_6}),  we also have $\lambda_2^*\big(\big(\Lambda^{\infty,*}\big)^2 K^{\infty,*}+K_Z^{\infty,*}-\kappa\big)=\lambda_2^{*}\big(0-\kappa\big)= 0$, hence $\lambda_2^*=0$. By (ii) it follows that   $C^\infty(\kappa)=0, \forall \kappa \in [0,\infty)$, hence the rate is zero.  Similarly, if $\lambda_3^*\neq 0$ then $K_Z^{\infty,*}=0$ and $K^{\infty,*}=0$, which  lead to a zero rate. However, it can be verified (see Theorem~\ref{thm_nfb}),  that for $\Lambda^{\infty,*}=0, K_Z^{\infty,*}\neq 0$,  we exhibit a non-zero rate, which is a lower bound on the non-feedback rate. 
Next, consider the case $\lambda_3^*= 0, \lambda_4^* = 0$ and $\lambda_1^*= 0$. Then, by \eqref{2.2} $\lambda_2^*= 1$. However, when $\lambda_2^*= 1$ equalities \eqref{2.3} and \eqref{2.4} hold only if $c=0$.  Since, $c\neq 0$, otherwise the channel is memoryless, then the only choice is $\lambda_1^*> 0$. Moreover, since $\lambda_1^* >0$, then in order to satisfy the complementary slackness condition \eqref{nc_9}, then inequality  must hold with equality, i.e.,
\begin{align}
\Big(K^{\infty,*}-c^2 K^{\infty,*} - K_W\Big)\Big(K_Z^{\infty,*}+ K_W + \big(\Lambda^{\infty,*} + c \big)^2 K^{\infty,*}\Big) +\Big( K_W + c  K^{\infty,*}\big(\Lambda^{\infty,*} + c \big)\Big)^2=0.   \label{eqcomsla}
 \end{align}
Finally we consider the case $\lambda_3^*= 0, \lambda_4^* = 0$ and $\lambda_2^*= 0$. Solving the system of equations \eqref{nc_6}, \eqref{2.2}-\eqref{eqcomsla}, the following sets of solutions are obtained.\\
The first solution is
\bea
K^{\infty,*}=\frac{K_W+1}{1-c^2}, \ \ \  \Lambda^{\infty,*}=-\frac{c^2+K_W}{c(K_W+1)}, \ \ \   K_Z^{\infty,*}=-\frac{K_W(c^2+K_W)}{c^2(K_W+1)}, \ \ \ \lambda_1^*=1.
\eea
The second solution is
\bea
K^{\infty,*}=0, \ \ \  \Lambda^{\infty,*}=-\frac{c^2+1}{2c}, \ \ \   K_Z^{\infty,*}=0, \ \ \ \lambda_1^*=-\frac{1}{K_W}.
\eea
The first solution is discarded since for $K_W>0$, then, $K_Z^{\infty,*}<0$, while the second solution is discarded due to  (ii). Therefore $\lambda_2^* \neq 0$, which by the complementary slackness condition \eqref{nc_6} implies that 
\begin{align}
(\Lambda^{\infty,*})^2 K^{\infty,*} + K_{Z}^{\infty,*} = \kappa
\label{appcomsla}
\end{align}
Thus, we have shown that a necessary condition for existence of $\kappa \in (0,\infty)$ such that $C^\infty(\kappa)>0$ is $\lambda_1^* >0$, $\lambda_2^* >0$, $\lambda_3^* = 0$ and $\lambda_4^* = 0$. This completes the proof.

\subsection{proof of  the stability condition of Theorem \ref{thm_nfb}}
\label{pr_stab_Lam0}

From \eqref{nf_6} and \eqref{nf_7}, we have
\begin{align}
 F^{nfb}(K^\infty, K_{Z}^\infty) \tri c -M^{nfb}(K^\infty, K_{Z}^\infty)c
=c \left(1- \Big( K_{W} + c^2  K^\infty\Big)\Big(\kappa+ K_{W} + c^2 K^\infty\Big)^{-1}\right).\label{nf_8_new}
\end{align}
We begin the evaluation for the stable case, i.e., $|c|<1$. Then,
\begin{align}
 {|}F^{nfb}(K^{\infty,*}, K_{Z}^{\infty,*}){|} =|c| \Big{|}\Big( \kappa\Big)\Big(\kappa+ K_{W} + c^2 K^{\infty,*}\Big)^{-1}\Big{|}\sr{(a)}{<}|c|<1,\label{nf_9_new}
\end{align}
where $(a)$ holds since $\kappa\geq 0$, $K_W>0$ and $K^{\infty,*}>0$. Thus, the stability condition holds for $|c|<1$. For the unstable case, i.e., $|c|>1$, by substituting $K^{\infty,*}$, we have the following
\begin{align}
 {|}F^{nfb}(K^{\infty,*}, K_{Z}^{\infty,*}){|} =& \bigg{|}\frac{2c\kappa}{\kappa+K_W+c^2\kappa +\sqrt{(K_W+\kappa(c^2-1))^2+4c^2\kappa{K_W}}}\bigg{|}\label{nf_10_new}\\
 \sr{(b)}{\leq}&\bigg{|}\frac{2c\kappa}{\kappa+K_W+c^2\kappa +\sqrt{(K_W+\kappa(c^2-1))^2}}\bigg{|}\label{nf_11_new}\\
 \sr{(c)}{<}&\bigg{|}\frac{2c\kappa}{\kappa+K_W+c^2\kappa +\sqrt{(\kappa(c^2-1))^2}}\bigg{|} =\bigg{|}\frac{2c\kappa}{K_W+2c^2\kappa}\bigg{|}
   \sr{(d)}{<}\frac{|2c\kappa |}{|2c^2\kappa |}=\frac{1}{|c|}<1,\label{nf_14_new}
\end{align}
where $(b)$-$(d)$ hold since  $\kappa\geq 0$ and $K_W>0$. Thus, the stability condition also holds for $|c|>1$. The case $|c|=1$ follows from the equality in (\ref{nf_9_new}).

\subsection{proof of Theorem~\ref{thm_sol}.}
\label{pr_thm_sol}
We prove the statements in several steps.\\
(i) First, we recall  Lemma~\ref{lemma_nc}.(iv) 
that states if $K_Z^{\infty,*}=0$ or $K^{\infty,*}=0$ the $C^\infty(\kappa)=0, \forall \kappa \in [0,\infty)$. However,  by Theorem~\ref{thm_nfb},   if we restrict our optimization to $\Lambda^{\infty,*}=0$, and $K_Z^{\infty,*}\neq 0$,  then we exhibit a non-zero rate without feedback, which is a lower bound on the non-feedback rate.     This shows that the only choice for non-zero rate with feedback is to exists a $\kappa \in (0,\infty)$ such that  
 $K^{\infty,*}>0, K_{Z}^{\infty,*}>0$. If such $\kappa$ does not exists, then,  we can exhibit a non-zero rate without feedback by considering   time-invariant channel inputs strategies  without feedback, ${\bf P}_{X_t^o|X^{o,t-1},V_0}={\bf P}^\infty(dx_t|x^{t-1}, v_0), t=1, \ldots, n$. \\
(ii) Second, we recall Lemma~\ref{lemma_nc}.(v), which states a necessary condition for existence of a non-zero feedback rate for some $\kappa \in (0,\infty)$ is   $\lambda_1^* > 0$, $\lambda_2^* > 0$, $\lambda_3^* = 0$ and $\lambda_4^* = 0$. 
 For the rest of the derivation we characterize the set of all values $\kappa \in {\cal K}^{\infty}(c, K_{W})$ if such exist,  and treat the case when $ {\cal K}^{\infty}(c, K_{W})$ is empty separately, using non-feedback channel inputs. \\
(iii) Consider $\lambda_1^* > 0$, $\lambda_2^* > 0$, $\lambda_3^* = 0$ and $\lambda_4^* = 0$. We solve the system of equations \eqref{2.2}-\eqref{2.4}, \eqref{eqcomsla} and  \eqref{appcomsla}. First we solve the system of equations  \eqref{2.2} and \eqref{2.3} to obtain $\lambda_1^*$ and $\lambda_2^*$ as a function of $\{K^{\infty,*},K_{Z}^{\infty,*}, \Lambda^{\infty,*}\}$. By substituting $\lambda_1^*$, $\lambda_2^*$ and $K_{Z}^{\infty,*}$ from \eqref{appcomsla} in \eqref{2.4}, we obtain \eqref{sol_2}. Finally, by substituting $K_{Z}^{\infty,*}$ and $\Lambda^{\infty,*}$ in \eqref{eqcomsla} we deduce the quadratic equation (\ref{sol_4}). 
The two solutions of the  the quadratic equation (\ref{sol_4}) give rise to the following two solutions.  \\
The first solution is 
\begin{align}
&K^{\infty,*}=K_1^{\infty,*}= \frac{\kappa\Big(c^2-1\Big)^2-K_W}{c^2 \Big(c^2-1\Big)},\label{case_11} \\
&\Lambda^{\infty,*}=\Lambda_1^{\infty,*}= \frac{cK_W}{\kappa\Big(c^2-1\Big)^2-K_W}, \label{case_12}   \\
&K_Z^{\infty,*}=K_{Z_1}^{\infty,*}= \frac{\kappa \Big\{\Big(    \kappa\big(c^2-1\big)^2-K_W\Big)\Big(c^2-1\Big)\Big\}-K_W^2}{\Big( \kappa\big(c^2-1\big)^2-K_W\Big) \Big(c^2-1\Big)}.\label{case_13}\\
&\lambda_1^*=\frac{c^2}{K_W -\kappa\Big(1-c^2\Big)}, \hst \lambda_2^*=c^2.\label{case_13_l}
\end{align}
The second solution is
\begin{align}
&K^{\infty,*}=K_2^{\infty,*}= \frac{\kappa-K_W}{c^2},\label{case_14} \\
&\Lambda^{\infty,*}=\Lambda_2^{\infty,*}= \frac{c\kappa}{K_W-\kappa}, \label{case_15}\\
&K_{Z}^{\infty,*}=K_{Z_2}^{\infty,*}= \frac{\kappa K_W}{K_W-\kappa}.\label{case_16}\\
&\lambda_1^*= \frac{c^2}{\kappa-K_W\Big(1-c^2\Big)}, \hst \lambda_2^*=0.\label{case_16_l}
\end{align} 
 Solution $K^{\infty,*}= K_2^{\infty,*},  \Lambda^{\infty,*}=\Lambda_2^{\infty,*},
K_Z^{\infty,*} = K_{Z_2}^{\infty,*}$, $(\lambda_1^*, \lambda_2^*)$ given by (\ref{case_16_l}), are not valid solutions, because if  $K^{\infty,*}=K_2^{\infty,*}>0$, then  $K_{Z}^{\infty,*}=K_{Z_2}^{\infty,*}<0$, and vice-versa. Thus, the only valid solution is Solution 1, from which all statements of (1) are obtained. It remains to show the statements under (2) and (3).  \\
(iv)  Consider $c^2<1$ and define the set
\beae
{\cal A}_1(c,K_W)\tri \Big\{\kappa \in [0,\infty): \; K^{\infty,*}= K_1^{\infty,*}>0, \; K_{Z}^{\infty,*}= K_{Z_1}^{\infty,*}>0, \; c^2<1, c\neq 0, \;\lambda_1^*>0, \lambda_2^*>0  \Big\}.\nms \label{st_1}
\eeae
Similarly,  define the sets for $c^2>1$ 
\begin{align}
{\cal A}_2(c,K_W)\tri \Big\{\kappa \in [0,\infty): \; K^{\infty,*}= K_1^{\infty,*}>0, \; K_{Z}^{\infty,*}= K_{Z_1}^{\infty,*}>0, \; c^2>1, c\neq 0, \;\lambda_1^*>0, \lambda_2^*>0  \Big\} \label{st_22ck}
\end{align}
and 
the set for $c^2=1$
\begin{align}
{\cal A}_3(c,K_W)\tri \Big\{\kappa \in [0,\infty): \; K^{\infty,*}= K_1^{\infty,*}>0, \; K_{Z}^{\infty,*}= K_{Z_1}^{\infty,*}>0, \; c^2=1, c\neq 0, \;\lambda_1^*>0, \lambda_2^*>0  \Big\} \label{st_23ck}
\end{align}

The proof is then completed by determining the values of $\kappa \in {\cal A}_1(c,K_W)$ and $\kappa \in {\cal A}_2(c,K_W)$, if such exist.  If the set is empty then, we need to consider time-invariant channel input strategies without feedback; such strategies always exists, in view of the lower bound of Theorem~\ref{thm_nfb}.\\
(iv.1) For $c^2<1$, we have the following
\begin{alignat}{3}
& \  K^{\infty,*}=K_1^{\infty,*}>0&&\Longrightarrow \kappa<\frac{K_W}{(c^2-1)^2},\label{st_1ck}\\
& \ \lambda_1^*>0 &&\Longrightarrow \kappa<\frac{K_W}{(1-c^2)},\label{st_2ck}\\
& \ K_{Z}^{\infty,*}=K_{Z_1}^{\infty,*}>0 &&\Longrightarrow \kappa^2 \Big(c^2-1\Big)^3 -\kappa K_W\Big(c^2-1\Big) -K_W^2>0,\\
& &&\Longrightarrow \frac{K_W-K_W \sqrt{4c^2-3}}{2\Big(1-c^2\Big)^2} < \kappa < \frac{K_W+K_W \sqrt{4c^2-3}}{2\Big(1-c^2\Big)^2 }, \hso 4c^2 \geq 3. \label{st_3ck}
\end{alignat}
Next, we show the set ${\cal A}_1(c,K_W)$ is empty. For $4c^2 < 3$, 
$\kappa^2 \Big(c^2-1\Big)^3 -\kappa K_W\Big(c^2-1\Big) -K_W^2<0$, thus $K_{Z}^{\infty,*}<0$. For $4c^2 \geq 3$, and from \eqref{st_2ck} and \eqref{st_3ck}, it suffices to show that
\begin{align}
\frac{K_W-K_W \sqrt{4c^2-3}}{2\Big(1-c^2\Big)^2} >\frac{K_W}{1- c^2}, \hso \mbox{for}\hso  4c^2 \geq  3.\label{st_4ck}
\end{align}
Clearly, after simple algebra, we can show that \eqref{st_4ck} holds iff $\Big(c^2-1\Big)^2 >0, \hso \mbox{for}\hso  4c^2 \geq 3,$ thus it hold for all $c\neq 1$. Therefore, the set defined by ${\cal A}_1(c,K_W)$ is empty. \\
(iv.2) For $c^2>1$, we have the following
\begin{align}
&   K^{\infty,*}=K_1^{\infty,*}> \hso \Longrightarrow \hso \kappa>\frac{K_W}{(c^2-1)^2},\label{st_1ck1}\\
&\lambda_1^*>0 \hso \Longrightarrow \hso  \kappa>\frac{K_W}{(1-c^2)},\label{st_2ck1}\\
&K_{Z}^{\infty,*}=K_{Z_1}^{\infty,*}>0 \hso \Longrightarrow \hso \kappa^2 \Big(c^2-1\Big)^3 -\kappa K_W\Big(c^2-1\Big) -K_W^2>0,\\
& \kappa < \frac{K_W-K_W \sqrt{4c^2-3}}{2\Big(1-c^2\Big)^2 } \ \ \mbox{or} \ \ \kappa>\frac{K_W+K_W \sqrt{4c^2-3}}{2\Big(1-c^2\Big)^2}. \label{st_3ck1}
\end{align}
Next, note that for $c^2>1$ the following inequalities hold
\begin{align}
\frac{K_W}{(1-c^2)}<\frac{K_W-K_W \sqrt{4c^2-3}}{2\Big(1-c^2\Big)^2}<0<\frac{K_W}{(c^2-1)^2}<\frac{K_W+K_W \sqrt{4c^2-3}}{2\Big(1-c^2\Big)^2}.
\label{st_4ck1}
\end{align}
Then, from \eqref{st_1ck1}-\eqref{st_4ck1}, we deduce that the set ${\cal A}_2(c,K_W)$ is non empty, only if the power $\kappa$ satisfies $\kappa>\frac{K_W\left(1+ \sqrt{4c^2-3}\right)}{2\Big(1-c^2\Big)^2}$. For values of $\kappa\leq\frac{K_W\left(1+ \sqrt{4c^2-3}\right)}{2\Big(1-c^2\Big)^2}$, since they do not belong in the set ${\cal A}_2(c,K_W)$, we need to consider time-invariant channel input strategies without feedback.\\
(iv.3) For $c^2=1$ clearly the set ${\cal A}_3(c,K_W)$ is empty, thus we need to consider time-invariant channel input strategies without feedback.\\
Putting all the above together we obtain the statements under (2) and (3). The proof is  completed.

\subsection{proof of Lemma~\ref{lemma_nc-nf}.}
\label{pr_lemma_nc-nf}

 Recall Theorem~\ref{thm_nc}, and let  $K^{\infty,*}=K^{\infty,nc}$ be the solution given in this theorem with $(\Lambda^{\infty,*}, K_Z^{\infty,*})=(-c,K_Z^{\infty,nc})$.  Then 
\begin{align}
&C_{LB}^{\infty,nc}(\kappa)=C^{\infty}(\kappa)\Big|_{\Lambda^{\infty,*}=-c}=\frac{1}{2} \log\Big( \frac{K_Z^{\infty,nc}  +K_{W}}{K_{W}}\Big), \hso c\in (-1,1), \hst \kappa \in [0,\infty), \label{in_nc} \\
& c^2 K^{\infty,nc}+K_Z^{\infty,nc}=\kappa, \hso K^{\infty,nc}= c^2 K^{\infty,nc} + K_W -\frac{ K_W^2}{ K_Z^{\infty,nc}+ K_W}.
\end{align} 
 Recall  Theorem~\ref{thm_nfb}, and let $K^{\infty,*}=K^{\infty,nfb}$ be the solution of $(\Lambda^{\infty,*}, K_Z^{\infty,*})=(0,\kappa)$. Then 
 \begin{align}
&C_{LB}^{\infty,nfb}(\kappa)=C^{\infty}(\kappa)\Big|_{\Lambda^{\infty,*}=0}= \frac{1}{2} \log\Big( \frac{c^2 K^{\infty,nfb}  + K_Z^{\infty,nfb} +K_{W}}{K_{W}}\Big), \label{nfb_in1} \\
&K_{Z}^{\infty,nfb} = \kappa, \hso 
K^{\infty,nfb}= c^2 K^{\infty,nfb} + K_W -\frac{ \Big( K_W + c^2  K^{\infty,nfb}\Big)^2}{ \Big(K_Z^{\infty,nfb}+ K_W +   c^2 K^{\infty,nfb}\Big)}. \label{nfb_in2}
\end{align}
 Then,    $K_Z^{\infty,nc}=\kappa -c^2 K^{\infty,nc}\in (0,\infty),$ for all $\kappa \in (0,\infty)$, and moreover the following inequalities holds. 
\begin{align}
C_{LB}^{\infty,nc}(\kappa)=&\frac{1}{2} \log\Big( \frac{\kappa -c^2 K^{\infty,nc}  +K_{W}}{K_{W}}\Big) \\
 < & \frac{1}{2} \log\Big( \frac{\kappa+K_{W}}{K_{W}}\Big), \hso \forall \kappa \in (0,\infty), \hso c \in (-1,1), c\neq 0\\
 <&\frac{1}{2} \log\Big( \frac{c^2 K^{\infty,nfb} +\kappa +K_{W}}{K_{W}}\Big)
 =  C_{LB}^{\infty,nfb}(\kappa), \hst \mbox{by (\ref{nfb_in2}).}
\end{align}
 This completes the proof.

\bibliographystyle{IEEEtran}
\bibliography{Bibliography_capacity}



\end{document}